\documentclass{article}

\usepackage{arxiv}

\usepackage[utf8]{inputenc} 
\usepackage[T1]{fontenc}    
\usepackage{hyperref}       
\usepackage{url}            
\usepackage{booktabs}       
\usepackage{amsfonts}       
\usepackage{nicefrac}       
\usepackage{microtype}      
\usepackage{lipsum}

\usepackage{natbib}

\usepackage{amsmath,amssymb,amsfonts,url}
\usepackage{graphicx}
\usepackage{textcomp}
\usepackage{xcolor}
\usepackage{graphicx}
\usepackage{subfig}

\usepackage{balance}  
\usepackage{algorithm}
\usepackage{textcomp}

\usepackage{multirow}
\usepackage{xspace}
\usepackage[noend]{algpseudocode}
\usepackage{xcolor}
\usepackage[english]{babel}
\usepackage[utf8]{inputenc}

\usepackage{float}

\usepackage{amsmath}
\usepackage{amsthm}

\newtheorem*{example*}{}
\newtheorem{theorem}{Theorem}
\newtheorem{lemma}{Lemma}
\newtheorem{problem}{Problem}

\newcommand{\StratRec}{\textsf{StratRec}\xspace}
\newcommand{\ADPaR}{\textsf{ADPaR}\xspace}
\newcommand{\ADPaRA}{\textsf{ADPaR-Exact}\xspace}
\newcommand{\BatchStrat}{\textsf{BatchStrat}\xspace}

\title{Recommending Deployment Strategies for Collaborative Tasks}

\author{
  Dong Wei \\
  New Jersey Institute of Technology\\
  Newark, NJ USA \\
  \texttt{dw277@njit.edu} \\
   \And
 Senjuti Basu Roy \\
	New Jersey Institute of Technology\\
	Newark, NJ USA \\
  \texttt{senjutib@njit.edu} \\
   \AND
    Sihem Amer-Yahia \\
	\textit{CNRS, Univ. Grenoble Alpes}\\
   Grenoble, France \\
   \texttt{sihem.amer-yahia@cnrs.fr} \\
}

\begin{document}
\maketitle
\begin{abstract}
  Our work contributes to aiding requesters in deploying collaborative tasks in crowdsourcing. We initiate the study of recommending deployment strategies for  collaborative tasks  to requesters that are consistent with deployment parameters they desire: a lower-bound on the quality of the crowd contribution, an upper-bound on the latency of task completion, and an upper-bound on the cost incurred by paying workers. A deployment strategy is a choice of value for three dimensions: {\em Structure} (whether to solicit the workforce sequentially or simultaneously), {\em Organization} (to organize it collaboratively or independently), and {\em Style} (to rely solely on the crowd or to combine it with machine algorithms).  We propose \StratRec, an optimization-driven middle layer that recommends deployment strategies and alternative deployment parameters to requesters by accounting for worker availability. Our solutions are grounded in discrete optimization and computational geometry techniques that produce results with theoretical guarantees. We present extensive experiments on Amazon Mechanical Turk, and conduct synthetic experiments to validate the qualitative and scalability aspects of \StratRec. 
\end{abstract}

\section{Introduction}\label{intro}
Despite becoming a popular mean of deploying tasks, crowdsourcing offers very little help to requesters. In particular, task deployment requires that requesters identify  {\em appropriate deployment strategies}.  A strategy involves the interplay of multiple dimensions: {\em Structure} (whether to solicit the workforce sequentially or simultaneously), {\em Organization} (to organize it collaboratively or independently), and {\em Style} (to rely on the crowd alone or on a combination of crowd and machine algorithms). A strategy needs to be commensurate to {\em deployment parameters} desired by a requester, namely, a lower-bound on quality, an upper-bound on latency, and an upper-bound on cost. For example, for a sentence translation task, a requester wants the translated sentences to be at least $80$\% as good as the work of a domain expert, in a span of at most $2$ days, and at a maximum cost of \$$100$.  Till date, the burden is entirely on requesters to design deployment strategies that satisfy desired parameters. 
Our effort in this paper is to present a formalism and computationally efficient algorithms to recommend multiple strategies (namely $k$) to the requester that are commensurate to her deployment parameters, primarily  for collaborative tasks.
  
A recent work~\citep{borromeo2017deployment} investigated empirically the deployment of text creation tasks in Amazon Mechanical Turk (AMT). The authors validated the effectiveness of different propose {\em to automate strategy recommendation}. This is particularly challenging because the estimation of the cost, quality and latency of a strategy for a given deployment request must account for many factors. 

To realize our contributions, we develop \StratRec (refer to Figure~\ref{fig:framework}), an optimization-driven middle layer that sits between requesters, workers, and platforms. \StratRec has two main modules: {\em Aggregator} and {\em Alternative Parameter Recommendation} (\ADPaR in short).  {\em Aggregator} is responsible for recommending $k$ strategies to a batch of such incoming strategies for different collaborative tasks, such as text summarization and text translation, and provided evidence for the need to guide requesters in choosing the right strategy. In this paper, we deployment requests, considering worker availability. If the platform does not have enough qualified workers to satisfy all requests, {\em Aggregator} {\em  triages them by optimizing platform-centric goals, i.e., to maximize throughput or pay-off (details  in Section~\ref{sec:framework})}.  Unsatisfied requests are sent to the  {\em Alternative Parameter Recommendation} module (\ADPaR), that recommends different deployment parameters for which $k$ strategies are available.

In principle, {\em recommending deployment strategies involves modeling worker availability considering their skills for the tasks that require deployment}. This gives rise to a complex function that estimates parameters (quality, latency, and cost) of a strategy considering worker skills, task types, and worker availability. As the first ever principled investigation of strategy recommendation in crowdsourcing, we first make a binary match between workers' skills and task types and then estimate strategy parameters considering those workers' availability. Worker availability is captured as  a probability  distribution function (pdf)  by leveraging historical  data on a platform.  For example,  the pdf can capture that there  is a 70\% chance  of  having $7\%$ of  the  workers and a 30\%  chance of having  2\%  of the  workers available who  are  suitable to  undertake a  certain type of task.  In expectation, this gives rise to  $5.5$\% of  available workers. If a platform has $4000$ total workers available to undertake a certain type of  task, that  gives rise to a total of $220$ available workers in an expected sense. \StratRec works with such {\em expected values}.

\begin{figure}
\centering
  \includegraphics[width=8.5cm,height=5.2cm]{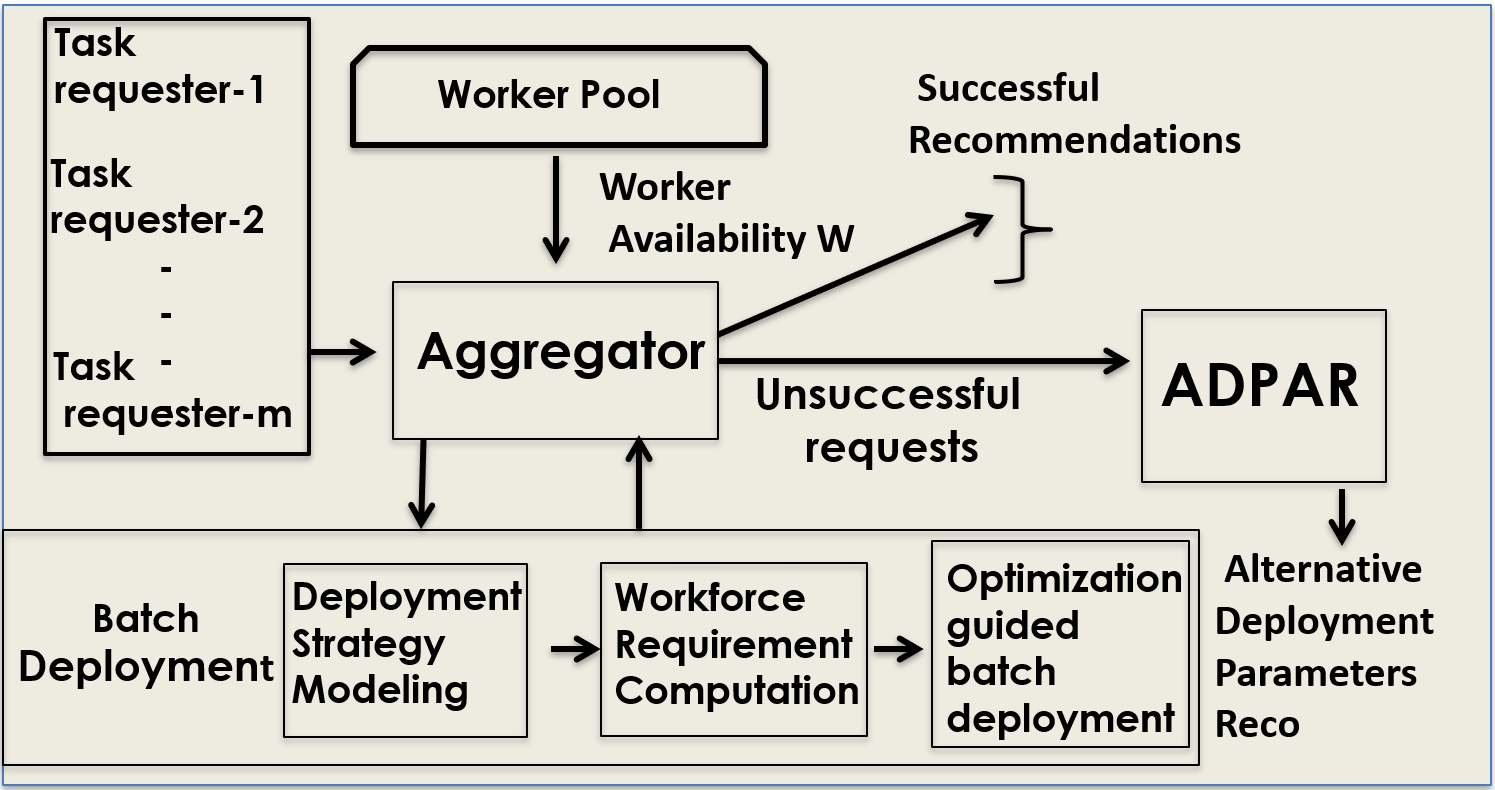}
  \caption{\StratRec Framework}
 \label{fig:framework}
\end{figure}

{\bf Contribution 1. Modeling and Formalism}: We present a general framework \StratRec for modeling quality, cost, and latency of a set of collaborative tasks, when deployed based on a strategy considering worker availability (Section~\ref{depmod}). The first problem we study is {\em Batch Deployment Recommendation} inside to deploy a batch of tasks to maximize two different platform-centric criteria: task throughput and pay-off.  After that, unsatisfied requests are sent one by one to the {\em Alternative Parameter Recommendation} module (\ADPaR). \ADPaR solves an optimization problem that recommends alternative parameters for which $k$ deployment strategies exist. For instance, if a request has a very small latency threshold that cannot be attained based on worker availability, \ADPaR may recommend to increase the latency and cost thresholds to find $k$ legitimate strategies. \ADPaR does not arbitrarily choose the alternative deployment parameters. It recommends those alternative parameters that are {\em closest}, i.e., minimizing the $\ell_2$ distance to the ones specified. 

{\bf Contribution 2.  Algorithms}:  In Section~\ref{sec:batch}, we design \BatchStrat, a unified algorithmic framework to solve the {\em Batch Deployment Recommendation} problem. \BatchStrat is greedy in nature and provides exact results for the throughput maximization problem, and a $1/2$-approximation factor for the pay-off maximization problem (which is NP-hard). In Section~\ref{sec:adpar}, we develop \ADPaRA to solve \ADPaR that is geometric and exploits the fact that our objective function is monotone (Equation~\ref{multiobjective-prob}). Even though the original problem is defined in a continuous space, we present a discretized technique that is exact. \ADPaRA employs a sweep-line technique~\cite{de1997computational} that gradually relaxes quality, cost, and latency, and is guaranteed to produce the tightest alternative parameters for which $k$ deployment strategies exist. 

{\bf Contribution 3. Experiments}: We conduct comprehensive real-world deployments for text editing applications  with real workers and rigorous synthetic data experiments (Section~\ref{exp}). The former validate that worker availability {\em varies over time, and could be reasonably estimated through multiple real world deployments}. It also shows {\em with statistical significance that cost, quality, latency have a linear relationship with worker availability for text editing tasks}. 
Our real data experiments (Section~\ref{realdata2}) also validate that when tasks are deployed considering recommendation of \StratRec, with statistical significance, they achieve higher quality and lower latency, under the fixed cost threshold on an average, compared to the deployments that do not consult  \StratRec. These results validate the effectiveness of deployment recommendations of our proposed frameworks and its algorithms. 

\section{Framework and Problem}\label{sec:model}

\subsection{Data Model}
           {\bf Crowdsourcing Tasks:} A platform is designed  to crowdsource tasks, deployed by a set of requesters and undertaken by crowd workers.  We consider collaborative tasks such as sentence translation, text summarization, and puzzle solving~\cite{julien,habib}.  

{\bf Deployment Strategies: } A deployment strategy~\cite{thesis} instantiates three dimensions:
{\em Structure} (sequential or simultaneous), {\em Organization} (collaborative or independent), and {\em Style} (crowd-only or crowd and algorithms). We rely on common deployment strategies~\cite{thesis, borromeo2017deployment} and refer to them as $\mathcal{S}$. Figure~\ref{fig:strategies} enlists some strategies that are suitable for text translation tasks (from English to French in this example). For instance, {\em SEQ-IND-CRO} in Figure~\ref{fig:strategies}(a) dictates that workers complete tasks sequentially ({\em SEQ}), independently ({\em IND}) and with no help from algorithms ({\em CRO}).  In {\em SIM-COL-CRO} (Figure~\ref{fig:strategies}(b)), workers are solicited in parallel ({\em SIM}) to complete a task collaboratively ({\em COL}) and with no help from algorithms ({\em CRO}). The last strategy {\em SIM-IND-HYB} dictates a hybrid work style ({\em HYB}) where workers are combined with algorithms, for instance with Google Translate.

A platform could provide the ability to implement some strategies. For instance, communication between workers enables {\tt SEQ} while collaboration enables {\tt COL}. Additionally, coordination between machines and humans may enable {\tt HYB}. Therefore, strategies could be implemented inside or outside platforms. In the latter, a platform could be used solely for hiring workers who are then redirected to an environment where strategies are implemented. In all cases, we will assume a set of strategies $\mathcal{S}$ for a given platform.

\begin{figure*}[htpb]
	\subfloat[{\bf SEQ-IND-CRO}]{
		\includegraphics[height=3cm, width=.22\textwidth]{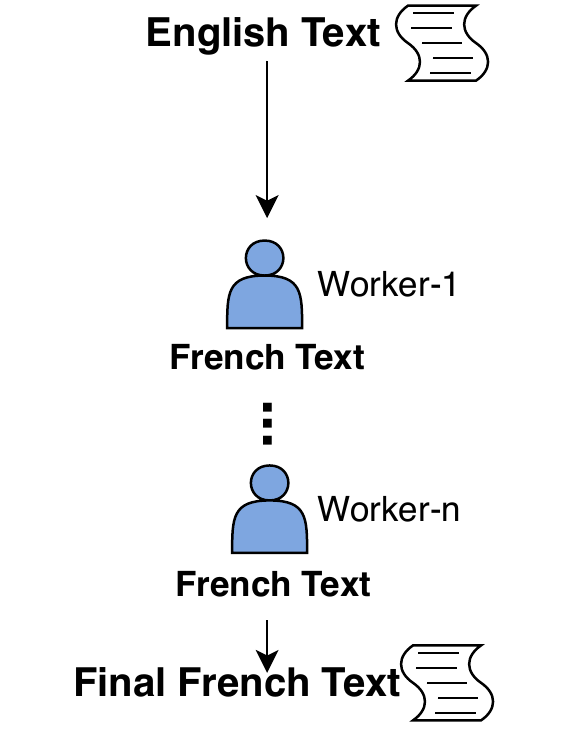}
	}
	\subfloat[{\bf SIM-COL-CRO}]{
		\includegraphics[height=3cm, width=.22\textwidth]{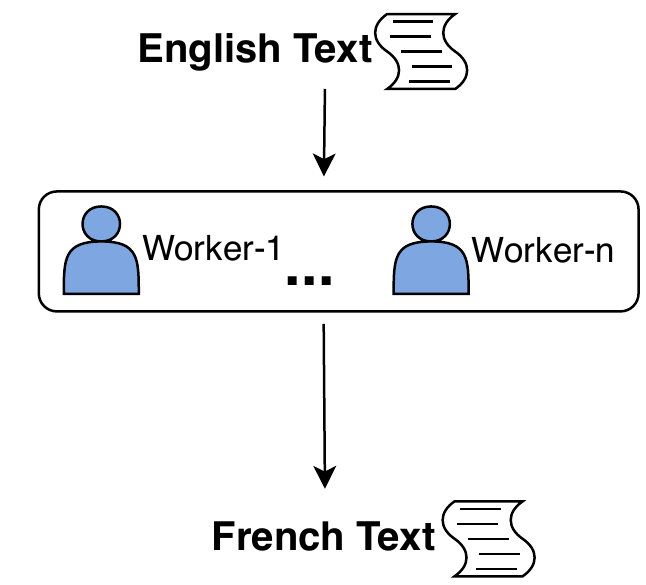}
	}
	\subfloat[{\bf SIM-IND-CRO}]{
		\includegraphics[height=3cm, width=.22\textwidth]{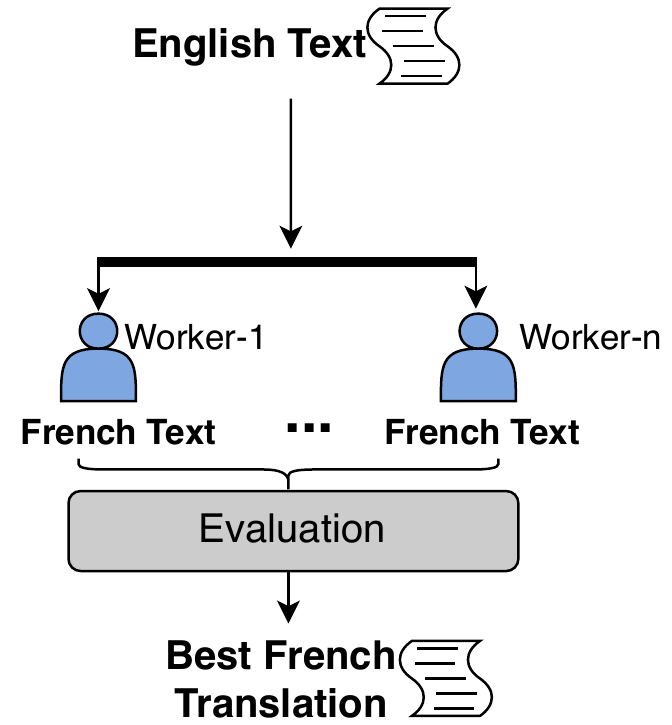}
	}
	\subfloat[{\bf SIM-IND-HYB}]{
		\includegraphics[height=3cm, width=.22\textwidth]{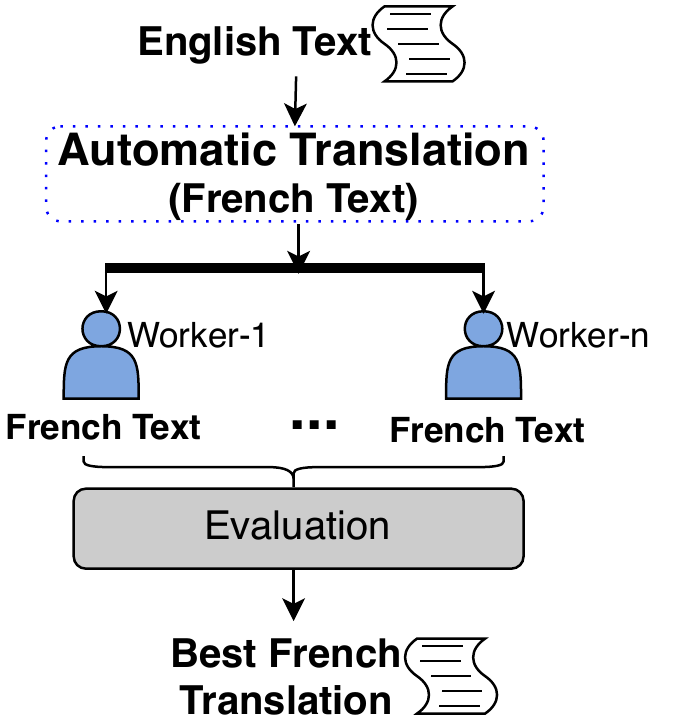}
	}
  	\caption{Deployment Strategies}
  	\label{fig:strategies}
\end{figure*}


For the purpose of illustration, we will only use a few strategies in this paper. However, in principle, the number of possible strategies could be very large. The closest analogy is query plans in relational databases in which joins, selections and projections could be combined any number of times and in different orders. Let us assume a collaborative task that involves $n$ workers and that $|Structure| \times |Organization| \times | Style|=v$. Even when the same combination of {\em Structure}, {\em Organization}, and {\em Style} appears at most once in a strategy, the number of possible strategies is essentially of the order of 
$(v^n \times v!)$. That is, if {\em Structure}, {\em Organization}, and {\em Style}, have $2$ unique choices each, the number of possible strategies is in the order of  $(8^n \times 8!)$. Now,  if the same combination of {\em Structure}, {\em Organization}, and {\em Style} can appear any number of times, the number of possible strategies becomes infinite. Additionally, there exists multiple real world tools Turkomatic~\cite{turkomatic} or Soylent~\cite{Bernstein10soylent:a}, that aid  requesters  in  planning  and solving collaborative tasks. In Turkomatic, while workers decompose and solve tasks, requesters can  view  the status  of worker-designed  workflows  in  real  time;  intervene  to  change  tasks; and request new solutions. For instance, with $x$ tasks in the workflow, there are $8^x$ possible strategies (e.g., $1,073,741,824$ strategies for $x=10$ and $v=8$). Such tools would certainly benefit from strategy recommendation.\\
\color{black}
{\bf Task Requests and Deployment Parameters:} 
A requester intends to find one or more strategies (notationally $k$, a small integer) for a deployment $d$ with parameters on quality, cost, and latency ($d.quality$, $d.cost$, $d.latency$) such that,  when a task in $d$ is deployed using strategy $s \in S$,  it is estimated to achieve a crowd contribution quality $s.quality$, by spending at most $s.cost$, and the deployment will last at most $s.latency$.
\begin{table}[!htbp]
\centering
    \begin{tabular}{ | l | l | l | l |}
    \hline
     & Quality & Cost & Latency \\ \hline
    $d_1$ & $0.4$ & $0.17$  & $0.28$ \\
    $d_2$ & $0.8$ & $0.2$ & $0.28$ \\
    $d_3$ & $0.7$ & $0.83$ & $0.28$ \\
   \hline \hline
   $s_1$ & $0.5$ & $0.25$ & $0.28$ \\
   $s_2$ & $0.75$ &  $0.33$ &$0.28$ \\
   $s_3$ & $0.8$  & $0.5$ & $ 0.14$ \\
   $s_4$ & $0.88$ & $0.58$ & $0.14$ \\
    \hline
    \end{tabular}
    \caption{Deployment Requests and Strategies}\label{def}
\end{table}

\begin{example*}\label{ex1}
 {\bf Example 1}: Assume there are $3$ ($m=3$) task deployment requests for different types of collaborative sentence translation tasks. The first requester $d_1$ is interested in deploying sentence translation tasks for $2$ days (out of $7$ days), at a cost up to \$$100$ (out of $\$600$ max), and expects the quality of the translation to reach at least $40\%$ of domain expert quality. Table~\ref{def} presents these after normalization between $[0-1]$. We set $k=3$. 
\end{example*}


A strategy $s$ is suitable to be recommended to $d$, if $s.quality \geq d.quality$ \& $s.cost \leq d.cost$ \& $s.latency \leq d.latency$. Estimating the parameters $s.quality$, $s.cost$, $s.latency$  for each $s$ and deployment $d$ requires accounting for the worker pool and their skills who are available to undertake  tasks in $d$. {\em A simple yet reasonable approach to that is to first match task types in a deployment request with workers' skills to select a pool of workers. Following that,  we account for {\em worker availability} from this selected pool, since the deployed tasks are to be done by those workers. Thus, the (estimated) quality, cost and latency of a strategy for a task is a function of worker availability, considering a selected pool of workers who are suitable for the tasks}. 


{\bf Worker Availability:} Worker availability is a discrete random variable and is represented by its corresponding distribution function (pdf), which gives the probability of the proportion of workers who are suitable and available to undertake tasks of a certain type within a specified time $d.latency$ (refer to Example~\ref{ex1}).  This pdf is computed from historical data on workers' arrival and departure on a platform.  \StratRec computes the expected value of this pdf to represent the available workforce $W$, as a normalized value in $[0,1]$.  In the remainder of the paper, worker availability stands for worker availability  in expectation, unless otherwise specified. How  to  accurately estimate worker availability  is an  interesting yet orthogonal problem and  not our focus here.

\subsection {Illustration of \StratRec}\label{sec:framework}
StratRec is an optimization-driven middle layer that sits between requesters, workers, and platforms. At any time, a crowdsourcing platform has a batch of $m$ deployment requests each with its own parameters as defined above, coming from different requesters. \StratRec is composed of two main modules -  {\em Aggregator} and  {\em Alternative Parameter Recommendation} (or \ADPaR).  

For the purpose of illustration, continuing with Example~\ref{ex1}, $\mathcal{S}$ consists of the set of $4$ deployment strategies, as shown in Figure~\ref{fig:strategies}: {\em SIM-COL-CRO, SEQ-IND-CRO, SIM-IND-CRO, SIM-IND-HYB.} To ease understanding, we name them as $s_1$, $s_2$, $s_3$, $s_4$, respectively. 

These requests, once received by \StratRec, are sent to the {\em Aggregator}. First,  it analyzes the Worker Pool to estimate worker availability. There is a $50\%$ probability of  having $700$ workers and a $50\%$ probability of  having $900$ workers out of $1000$ suitable  workers  for sentence translation tasks available for the next $7$ days. Thus, the expected worker availability $W$ is $0.8$. After that, it consults the {\em Deployment Strategy Modeling} in {\em Batch Deployment module} to estimate quality, cost, and latency of a strategy (more in Section~\ref{depmod}) for a deployment.  Since all deployments are of same type, Equation~\ref{eq2},  could be used to estimate those Strategy parameters (also presented in Table~\ref{def}). Then, it consults the {\em Workforce Requirement Computation} to estimate workforce requirement of each strategy (more in Section~\ref{wr} and Figure~\ref{WR}).  Finally, the 
{\em Optimization Guided Batch Deployment} (refer to Section~\ref{od}) is invoked to select a subset of requests that optimizes the underlying goal and recommends $k$ strategies for each. Each unsatisfied request $d_i$ is sent to \ADPaR that recommends an alternative deployment $d'_i$ to the requester for which there exists $k$ deployment strategies. 

Using Example~\ref{ex1}, out of the three deployment requests, only  $d_3$ could be fully served  (considering either throughput or pay-off objective) and $s_2$, $s_3$, $s_4$ are recommended.  $d_1$ and $d_2$ are then sent to \ADPaR one by one.

\subsection{Problem Definitions}
\label{sec:pbms}

\begin{problem}\label{pbm1}
{\bf Batch Deployment Recommendation:}
Given an optimization goal $F$, a set $\mathcal{S}$ of strategies, a batch of $m$ deployment requests from different requesters, where the $i$-th task deployment $d_i$ is associated with parameters $d_i.quality$, $d_i.cost$ and $d_i.latency$, and worker availability $W$, distribute $W$ among these requests by recommending $k$ strategies for each request, such that $F$ is optimized.

The high level problem optimization problem could be formalized as:

\begin{equation}\label{opt}
\begin{aligned}
& \text{Maximize }
F= \sum f_i \\
& \text{s.t.}   \sum \vec{w_i} \leq W  \text{ AND } \\
&   d_i \text{ is successful}
\end{aligned}
\end{equation}
\end{problem}

where $f_i$ is the optimization value of deployment $d_i$ and $\vec{w_i}$ is the workforce required to successfully recommend $k$ strategies it. A deployment request $d_i$ is successful, if for each of the $k$ strategies in the recommended set of strategies $S^i_d$, the following three criteria are met:  $s.cost \leq d_i.cost$, $s.latency \leq d_i.latency$ and $s.quality \geq d_i.quality$.

Using Example~\ref{ex1},  $d_3$ is successful, as it will return $S^3_d = \{s_2,s_3,s_4\}$, such that $d_3.cost \ge s_4.cost \ge s_3.cost \ge s_2.cost $ \& $d_3.latency \ge s_4.latency\ge s_3.latency\ge s_2.latency$ \& $d_3.quality \le s_4.quality \le s_3.quality \le s_2.quality$, and it could be deployed with the available workforce  $W=0.8$.

In this work, $F$ is designed to maximize one of two different platform centric-goals: task throughput and pay-off.

{\em Throughput} maximizes the total number of successful strategy recommendations without exceeding $W$. Formally,
\begin{equation}
\begin{aligned}
& \text{Maximize}
  \sum_{i = 1}^{m} x_i  \\
& \text{s.t.}   \sum x_i \times \vec{w_i} \leq W \\
&   x_i =
\begin{cases}
1& d_i.cost \leq s_j.cost  \text{ AND } \\
 &   d_i.latency \leq s_j.latency \text{ AND} \\
& d_i.quality \geq s_j.{quality} \text{ AND} \\
& |S^i_{d}| = k , \forall i=1,\ldots, m; j=1,\ldots,|\mathcal{S}|\\
0& otherwise
\end{cases} \\
\end{aligned}
\label{throughput}
\end{equation}
{\em Pay-off} maximizes $d_i.cost$, if $d_i$ is a successful deployment request without exceeding $W$. The rest of the formulation is akin to Equation~\ref{throughput}.

\begin{problem}\label{pbm2}
{\bf Alternative Parameter Recommendation:} 
Given a deployment $d$, worker availability $W$, a set of deployment strategies $\mathcal{S}$, and a cardinality constraint $k$,  \ADPaR recommends an alternative deployment $d'$ and associated $k$ strategies, such that,  the Euclidean distance ($\ell_2$)  between $d$ and $d'$ is minimized.

Formally, our problem could be stated as a constrained optimization problem:
\begin{equation}\label{multiobjective-prob}
\begin{aligned}
& \text{min}
& & (d'.cost - d.cost)^2 + (d'.latency - d.latency)^2 \\ & & &+ (d'.quality - d.quality)^2\\
& \text{s.t.} & &  \sum_{j = 1}^{|\mathcal{S}|} x_j = k \\
& & & x_j =
\begin{cases}
1& d'.cost \leq s_j.cost  \text{ AND } \\
 &   d'.latency \leq s_j.latency \text{ AND} \\
& d'.quality \geq s_j.{quality}\\
0& otherwise
\end{cases} \\
\end{aligned}
\end{equation}
\end{problem}

Based on Example~\ref{ex1}, if \ADPaR takes the following input values $d_1:(0.4,0.17,0.28)$ and $\mathcal{S}$. For $d_1$, the alternative recommendation should be $(0.4,0.5,0.28)$ with three strategies $s_1, s_2, s_3$.

%

\section{Deployment Recommendation}\label{sec:batch}
We describe our proposed  solution for {\em Batch Deployment Recommendation} (Problem~\ref{pbm1}). Given $m$ requests and $W$, the {\em Aggregator}  invokes \BatchStrat, our  unified solution to solve the batch deployment recommendation problem. There are three major steps involved. \BatchStrat first obtains model parameters of a set of candidate strategies (Section~\ref{depmod}), then computes workforce requirement to satisfy these requests (Section~\ref{wr}), and finally performs optimization to select a subset of $m$ deployment requests, such that different platform-centric optimization goals could be achieved (Section~\ref{od}).

We first provide an abstraction which serves the purpose of designing \BatchStrat. Given $m$ deployment requests and $W$ workforce availability, we intend to compute a two dimensional matrix $\mathcal{W}$, where there are 
$|\mathcal{S}|$ columns that map to available deployment strategies and $m$ rows of different deployment requests. Figure~\ref{mapping} shows the matrix built for Example~\ref{ex1}. A cell $w_{ij}$  in this matrix estimates the workforce required to deploy $i$-th request using $j$-th strategy. This matrix $\mathcal{W}$ is crucial to enable platform centric optimization for batch deployment. 

\subsection{Deployment Strategy Modeling} \label{depmod}
\BatchStrat first performs deployment strategy modeling to estimate quality, cost, latency of a strategy $s$ for a given deployment request $d$. As the first principled solution, it models these parameters as a  linear function of worker availability, from the filtered pool of workers whose profiles match tasks in the deployment request \footnote{\small We note that \StratRec could be adapted for  tasks that do not exhibit such linear relationships.}.  Therefore, if $d$ is deployed using strategy $s$, the quality parameter of this deployment is modeled as:

\begin{equation}\label{eq2}
s_d.quality = \alpha_{qds}.(w_{qds})+\beta_{qds}
\end{equation}

Our experimental evaluation (Table~\ref{mc}) in Section~\ref{realdata}, performed on AMT validates this linearity assumption with $90\%$ statistical significance for two text editing tasks.

Model parameters  $\alpha$ and $\beta$ are obtained for every $s$, $d$, and parameter (quality, cost, latency) combination, by fitting historical data to this linear model.  Once these parameters are known, \BatchStrat uses Equation~\ref{eq2} again to estimate workforce requirement  $w_{qds}$ to satisfy quality threshold (cost and latency like-wise) for deployment $d$ using strategy $s$.  We repeat this exercise for each $s \in \mathcal{S}$, which comprises our set of candidate strategies for a deployment $d$.

\subsection{Workforce Requirement Computation}\label{wr}
The goal of the {\em Workforce Requirement Computation} is to estimate workforce requirement per (deployment, strategy) pair.  It performs that in two sub-steps, as described below. \\
{\bf (1) Computing Matrix $\mathcal{W}$: } The first step is to compute $\mathcal{W}$, where $w_{i,j}$ represents the workforce requirement of deploying $d_i$ with strategy $s_j$. Recall that in Equation~\ref{eq2}, as long as for a deployment $d_i$, the deployment parameters on quality, cost, and latency, i.e., $d_i.quality$, $d_i.cost$ and $d_i.latency$ are known, for a strategy, $s_j$, we can compute $w_{i,j}$, i.e., that is the minimum workforce needed to achieve those thresholds, by considering the equality condition, i.e.,  $s_j.quality=d_i.quality$ (similarly for cost and latency), and solving Equation~\ref{eq2} for $w$, with known ($\alpha, \beta$) values.
Using Example~\ref{ex1}, the table in Figure~\ref{mapping} shows the rows and columns of matrix $\mathcal{W}$ and how a workforce requirement could be calculated for  $w_{11}$. Basically, once we solve the workforce requirement of quality, cost, and latency($w_{qij}$, $w_{cij}$, $w_{lij}$), the overall workforce requirement of deploying $d_i$ using $s_j$ is the maximum over these three requirements. Formally, they could be stated as follows:

\[
  w_{ij}=Max \begin{cases}
               d_i.quality=\alpha_{qij} w_{qij}+\beta_{qij}\\
               d_i.cost=\alpha_{cij} w_{cij}+ \beta_{cij}\\
               d_i.latency=\alpha_{lij} w_{lij}+ \beta_{lij}
            \end{cases}
\]
Using Example~\ref{ex1}, $w_{11}$ is the maximum over $\{w_{q11},w_{c11}, w_{l11}\}$.  Figure~\ref{mapping} shows how $w_{11}$ needs to be computed for deployment $d_1$ and strategy $s_1$ for the running example.

{\bf Running Time:} Running time of computing  $\mathcal{W}$ is $O(m |\mathcal{S}|)$, since computing each cell $w_{ij}$ takes constant time.


\begin{figure*}[h]
	\subfloat[Requirement for ($d_1,s_1$) ]{
		\includegraphics[height=3cm, width=.3\textwidth]{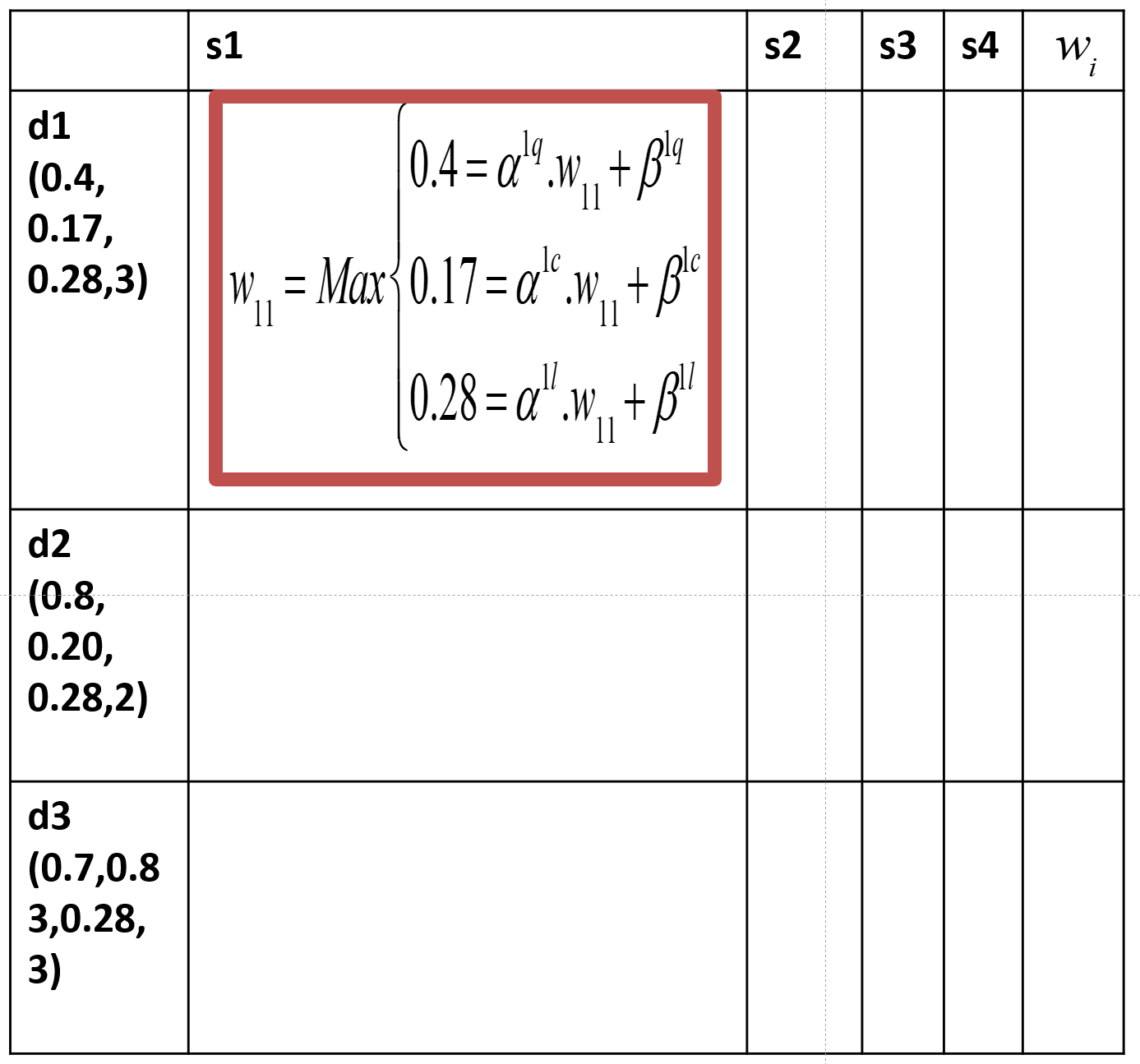}\label{mapping}
	}
	\hfill
	\subfloat[Aggregated requirement per request (Sum)]{
		\includegraphics[height=3cm, width=.3\textwidth]{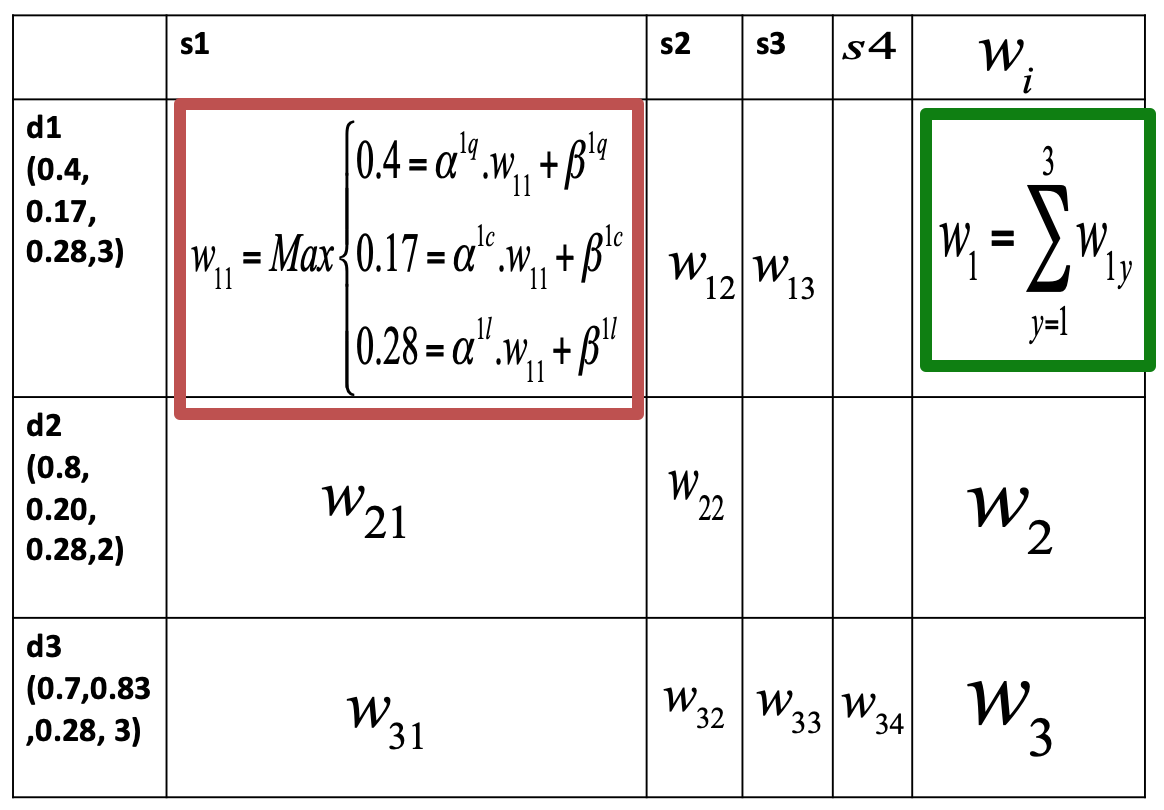}\label{sum}
	}
	\hfill
	\subfloat[Aggregated requirement per request (Max)]{
		\includegraphics[height=3cm, width=.3\textwidth]{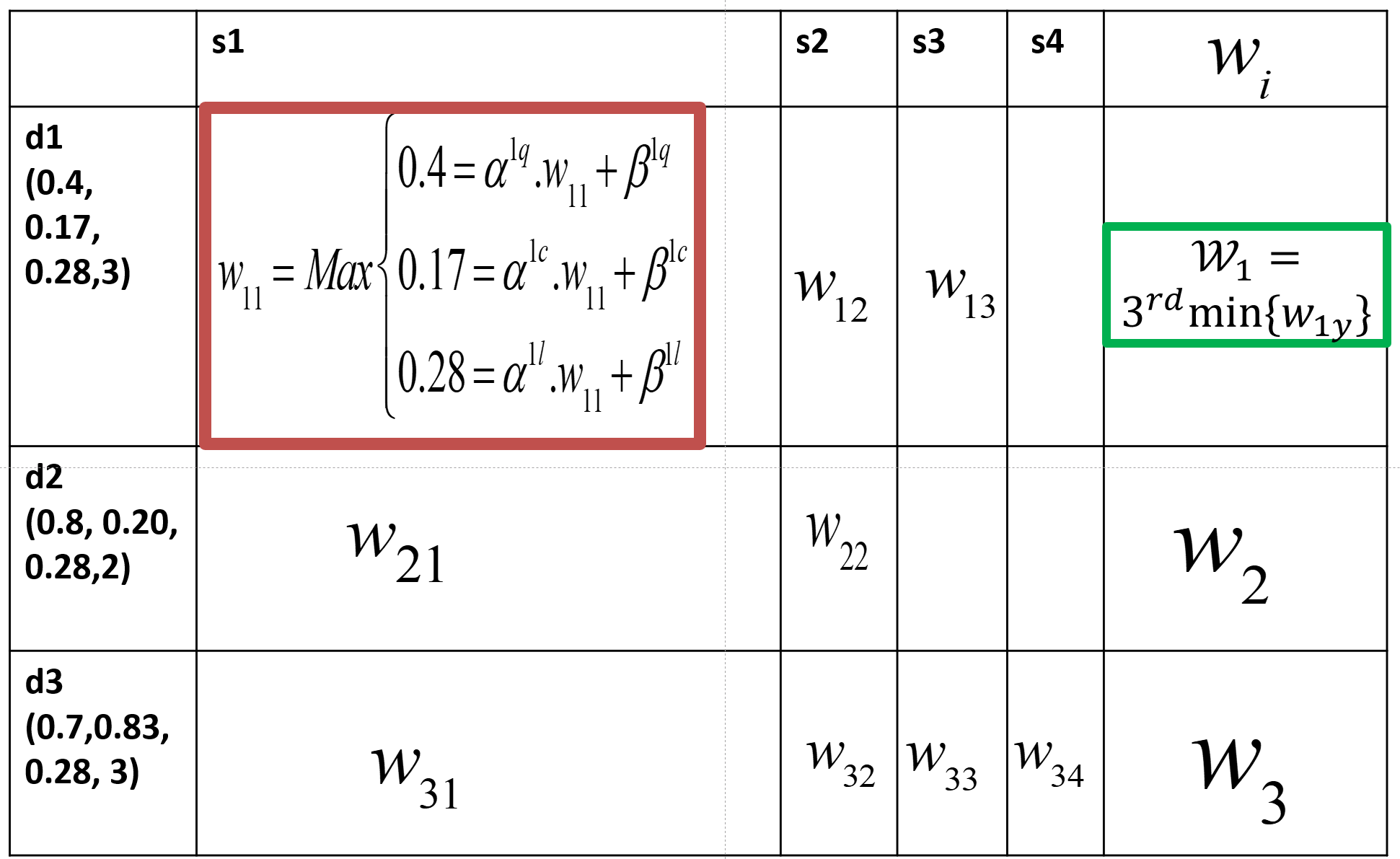}\label{max}
	}
	\caption{Computing Workforce Requirement}
	\label{WR}
\end{figure*}

{\bf (2) Computing Workforce Requirement per Deployment:} For a deployment request $d_i$ to be successful, \BatchStrat has to find $k$ strategies, such that each satisfies the deployment parameters.  In step (2), we investigate how to make compute workforce requirement for all $k$ strategies, for each $d_i$.  The output of this step produces a vector $\vec{W}$ of length $m$, where the $i$-th value represents the aggregated workforce requirement for request $d_i$. Computing $\vec{W}$ requires understanding of two cases:
\begin{itemize}
\item {\bf Sum-case:} It is possible that the task designer intends to perform the deployment using all $k$ strategies.  Therefore, the minimum workforce ($w_i$) needed to satisfy cardinality constraint $k_i$ is $\Sigma_{ y=1}^{ k} w_{iy}$ (where $w_{iy}$ is the $y$-th smallest workforce value in row $i$ of matrix $\mathcal{W}$.
\item {\bf Max-case:} The task designer intends to only deploy one of the $k$ recommended strategies - in that case, $w_i =  w_{iy}$, (where $w_{iy}$ is the $k$-th smallest workforce value in row $i$ of matrix $\mathcal{W}$).
\end{itemize}
Figures~\ref{sum} and~\ref{max} represent how $\vec{W}$ is calculated considering sum-case and max-case, respectively.

{\bf Running Time:} The running time of computing the aggregated workforce requirement of the $i$-th deployment request is $O(|\mathcal{S}| k log |\mathcal{S}|)$, if we use min-heaps to retrieve the $k$ smallest numbers. The overall running time is again $O(m  k \ log |\mathcal{S}|)$.

\subsection{Optimization-Guided Batch Deployment}\label{od}
Finally, we focus on the optimization step of \BatchStrat, where, given $\vec{W}$, the objective is to distribute the available workforce $W$ among $m$ deployment requests such that it optimizes a platform-centric goal $F$.  Since $W$ can be limited, it may not be possible to successfully satisfy all deployment requests in a single batch. This requires distributing $W$ judiciously among competing deployment requests and satisfying the ones that maximize platform-centric optimization goals, i.e., throughput or pay-off.

 At this point, a keen reader may notice that the batch deployment problem bears resemblance to a well-known discrete optimization problem that falls into the general category of assignment problems, specifically, Knapsack-type of problems~\cite{garey2002computers}. The objective is to maximize a goal (in this case, throughput or pay-off), subject to the capacity constraint of worker availability $W$. In fact, depending on the nature of the problem, the optimization-guided batch deployment problem could become intractable.

Intuitively, when the objective is only to maximize throughput (i.e., the number of satisfied deployment requests), the problem is polynomial-time solvable. However, when there is an additional dimension, such as pay-off, the problem becomes NP-hard problem, as we shall prove next.

	\begin{figure}[h]
		\centering
		\includegraphics[width=6cm,height=6cm]{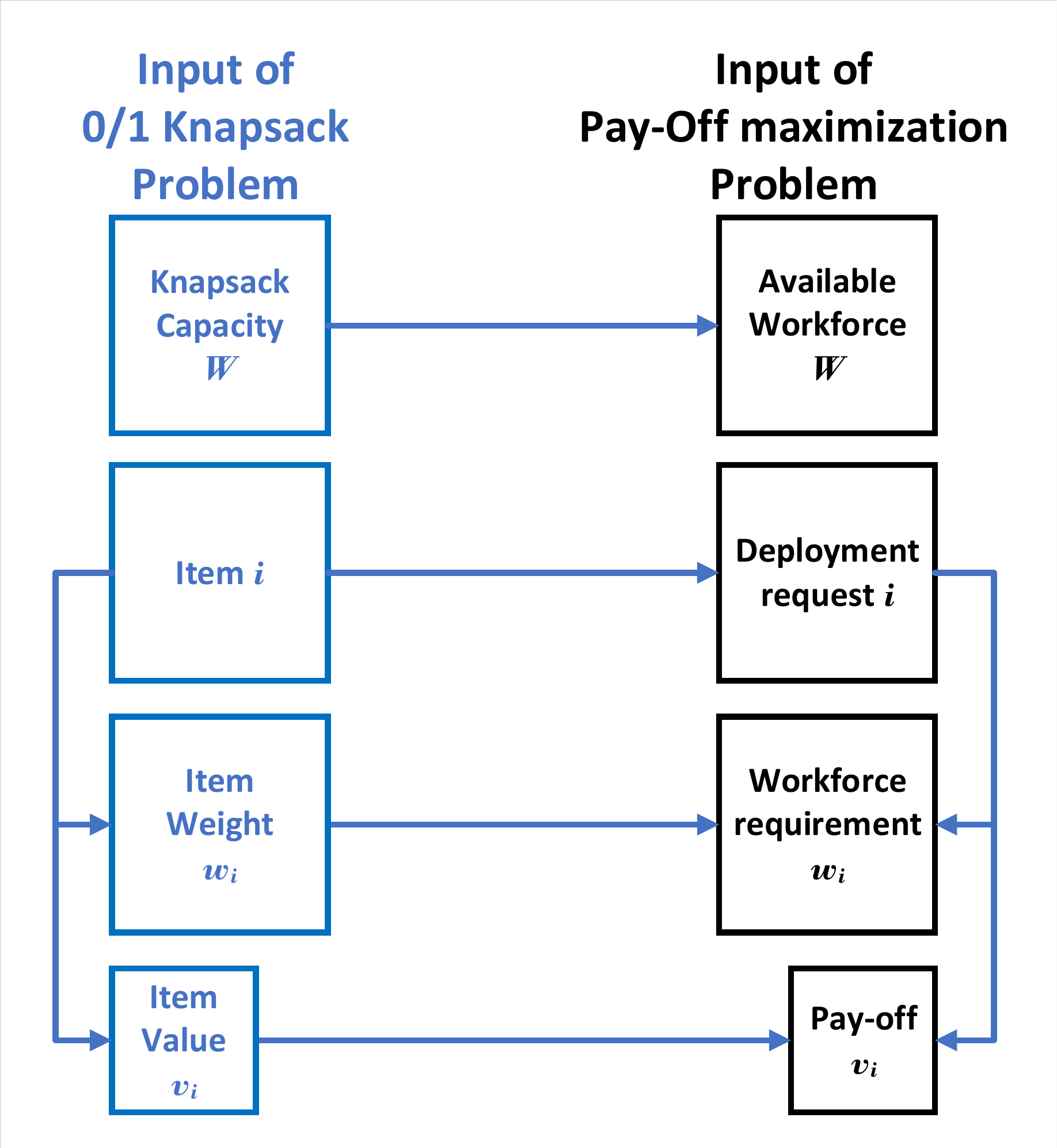}
		\caption{Reduction from the $0/1$ Knapsack problem to the Pay-off Maximization problem}
		\label{p1}
	\end{figure}

\begin{theorem}
The Pay-Off maximization problem is NP-hard.
\end{theorem}

\begin{proof}
		To prove NP-hardness, we reduce an instance of the known NP-hard problem $0/1$ Knapsack problem~\cite{garey2002computers} problem to an instance of the Pay-off Maximization problem. Given an instance of the 0/1 Knapsack problem, an instance of our problem could be created as follows: an item $i$ of weight $w_i$ and value $v_i$ represents a deployment request $d_i$ with minimum workforce requirement $w_i$ and pay-off $v_i$, as the Figure \ref{p1} shown.  Clearly this transformation is performed in polynomial time. After that, it is easy to notice that a solution exists for the 0/1 Knapsack problem iff the same solution solves our Pay-Off maximization problem.
\end{proof}

Our proposed solution bears similarity to the greedy approximation algorithm of the Knapsack problem~\cite{ibarra1975fast}. The objective is to sort the deployment strategies in non-increasing order of $\frac{f_i}{\vec{w_i}}$.  The algorithm greedily adds deployments based on this sorted order until it hits a deployment $d_i$ that can no longer be satisfied by $W$, that is, $\Sigma_{i=1..x}\ d_i > W$.  At that step, it chooses the better of $\{d_1,d_2,d_{i-1}\}$ and $d_i$ and the process continues until no further deployment requests could be satisfied based on $W$. Lines $4-8$ in Algorithm \BatchStrat describe those steps.

{\bf Running Time:} The running time of this step is dominated by the sorting time of the deployment requests, which is $O(m \ log m)$.

\begin{algorithm}
\caption{Algorithm \BatchStrat}\label{alg:batch}
\begin{algorithmic}[1]
\State {\bf Input:} $m$ deployment requests, $\mathcal{S}$,  objective function $F$, available workforce $W$
\State {\bf Output:} recommendations for a subset of deployment requests.
\State Estimate model parameters for each (strategy, deployment) pair.
\State Compute Workforce Requirement Matrix $\mathcal{W}$
\State Compute Workforce Requirement per Deployment Vector $\vec{W}$
\State Compute the objective function value $f_i$ of each deployment request $d_i$
\State Sort the deployment strategies in non-increasing order of  $\frac{f_i}{\vec{w_i}}$
\State Greedily add deployments until we hit $d_i$, such that $\Sigma_{i=1..x}\ d_i > W$
\State Pick the better of $\{d_1,d_2,d_{i-1}\}$ and $d_i$
\end{algorithmic}
\end{algorithm}

\subsubsection{Maximizing Throughput}
When task throughput is maximized, the objective function $F$ is computed simply by counting the number of deployment requests that are satisfied by the {\tt Aggregator}. Therefore, $f_i$, the objective function value of deployment $d_i$ is the same for all the deployment requests and is $1$. Our solution, \BatchStrat-ThroughPut, sorts the deployment requests in increasing order of workforce requirement $\vec{w_i}$ to make 
$\frac{1}{\vec{w_i}}$ non-increasing. Other than that, the rest of the algorithm remains unchanged.
\begin{theorem}
Algorithm \BatchStrat-ThroughPut gives an exact solution to the problem.
\end{theorem}
\begin{proof}
We proof this theorem by the method of contradiction.

Assume the solution $S$ produced by the Algorithm \BatchStrat-ThroughPut has the objective function value $x$. Assume $O$ is the optimal solution with the objective function value $y$, such that $S\ne O$(that is, $x\le y$). Let $d_i$ be the first deployment request, such that, $d_i \notin O$. Therefore, $d_1,\dots,d_{i-1}\in O$.

Let the $i$-th deployment request present in $S$ and $O$  have the workforce requirement  of $w_i$ and $w'_i$, respectively. Since, \BatchStrat-ThroughPut iteratively selects the deployments in ascending order of workforce requirement, therefore, we can say, for all $j\ge i$, $\frac{1}{w_j}\ge \frac{1}{w'_j}$. Therefore, to satisfy $x\le y$,  $O$ must exceed the available workforce constraint $W$. That is a contradiction. Therefore, $x =y$ and Algorithm \BatchStrat-ThroughPut gives an exact solution of the throughput problem.
\end{proof}

\subsubsection{Maximizing Pay-Off}
Unlike throughput, when pay-off is maximized, there is an additional dimension involved that is different potentially for each deployment request. $f_i$ for deployment request $d_i$ is computed using $d^i.cost$, the amount of payment deployment $d_i$ is willing to expend. Other than that, the rest of the algorithm remains unchanged.
\begin{theorem}
Algorithm \BatchStrat-PayOff has a 1/2-approximation factor.
\end{theorem}
	\begin{proof}
		Assume $OPT$ be the optimal Payoff value of the problem. According to the algorithm, the output of algorithm \BatchStrat-PayOff is either $\{d_1,...,d_{i-1} \}$ or $\{d_i\}$. Now, if we perform integer relaxation of this problem (that is, it is allowed to satisfy a deployment request partially), the optimal value of the relaxed problem becomes, \\ 
		\begin{equation}
			\text{Payoff}_i\ge\sum_{x = 1...i-1}\text{Payoff}_x+\frac{\vec{W}-\sum_{x=1...i-1}w_x}{w_i}\text{Payoff}_i
		\end{equation}
	
	This is larger than $OPT$, because deployments are allowed to be partially satisfied due to integer relaxation.
	
	 In addition, 
\begin{equation}
 \text{Payoff}_i\ge\frac{\vec{W}-\sum_{x=1...i-1}w_x}{w_i}\text{Payoff}_i
\end{equation}	 
Therefore:
		\begin{equation}
		\sum_{x = 1...i-1}\text{Payoff}_x+\text{Payoff}_i\ge\sum_{x = 1...i-1}\text{Payoff}_x+\frac{\vec{W}-\sum_{x=1...i-1}w_x}{w_i}\text{Payoff}_i\ge OPT
		\end{equation}
		
Since,$\sum_{x = 1...i-1}\text{Payoff}_x+\text{Payof}f_i\ge OPT$, either $\sum_{x = 1...i-1}\text{Payoff}_x \ge \frac{1}{2}OPT$ or $\text{Payoff}_i\ge \frac{1}{2}OPT$. Since the output solution of algorithm \BatchStrat-PayOff is the better one in either $\{d_1,...,d_{i-1} \}$ or $\{d_i\}$, algorithm \BatchStrat-PayOff holds the $\frac{1}{2}$ approximation factor~\cite{lawler1979fast}.

	\end{proof}
	
%
\section{ADPaR}\label{sec:adpar}
We discuss our solution to the \ADPaR problem, that takes a deployment $d$ and strategy set $\mathcal{S}$ as inputs, and is designed to recommend alternative deployment parameters $d'$ to optimize the goal stated in Equation~\ref{multiobjective-prob} (Section~\ref{sec:pbms}), such that $d'$ satisfies the cardinality constraint of $d$.

Going back to Example~\ref{ex1} with the request $d_2$, \StratRec there is no strategy that satisfies $d_2$ (refer to Figure~\ref{fig:3dspace}).

At a high level, \ADPaR bears resemblance to {\em Skyline and Skyband queries}~\cite{sky1, sky2, sky3} - but as we describe in Section~\ref{related}, there are significant differences between these two problems - thus the former solutions do not adapt to solve \ADPaR. Similarly, \ADPaR is significantly different from existing works on query refinement~\cite{mishra,query2,query3,query4}, that we further delineate in Section~\ref{related}.

\begin{figure}
	\subfloat[Deployment parameters in 3-D space]{
		\includegraphics[height=6cm, width=.45\textwidth]{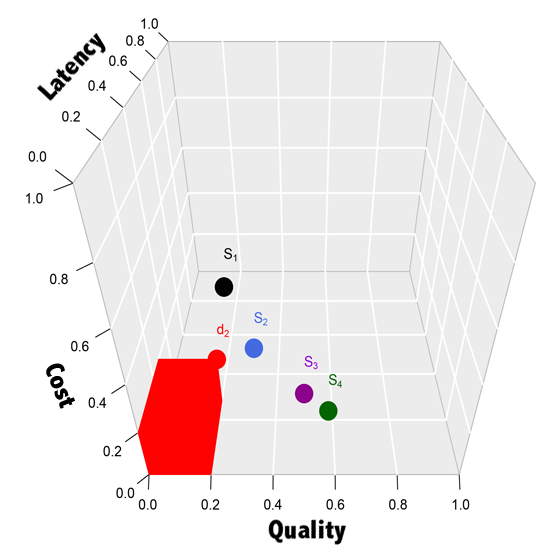}\label{fig:3dspace}
	}
	\hfill
	\subfloat[Projection of $d'$ on ({\tt L, Q}) plane]{
		\includegraphics[height=6cm, width=.45\textwidth]{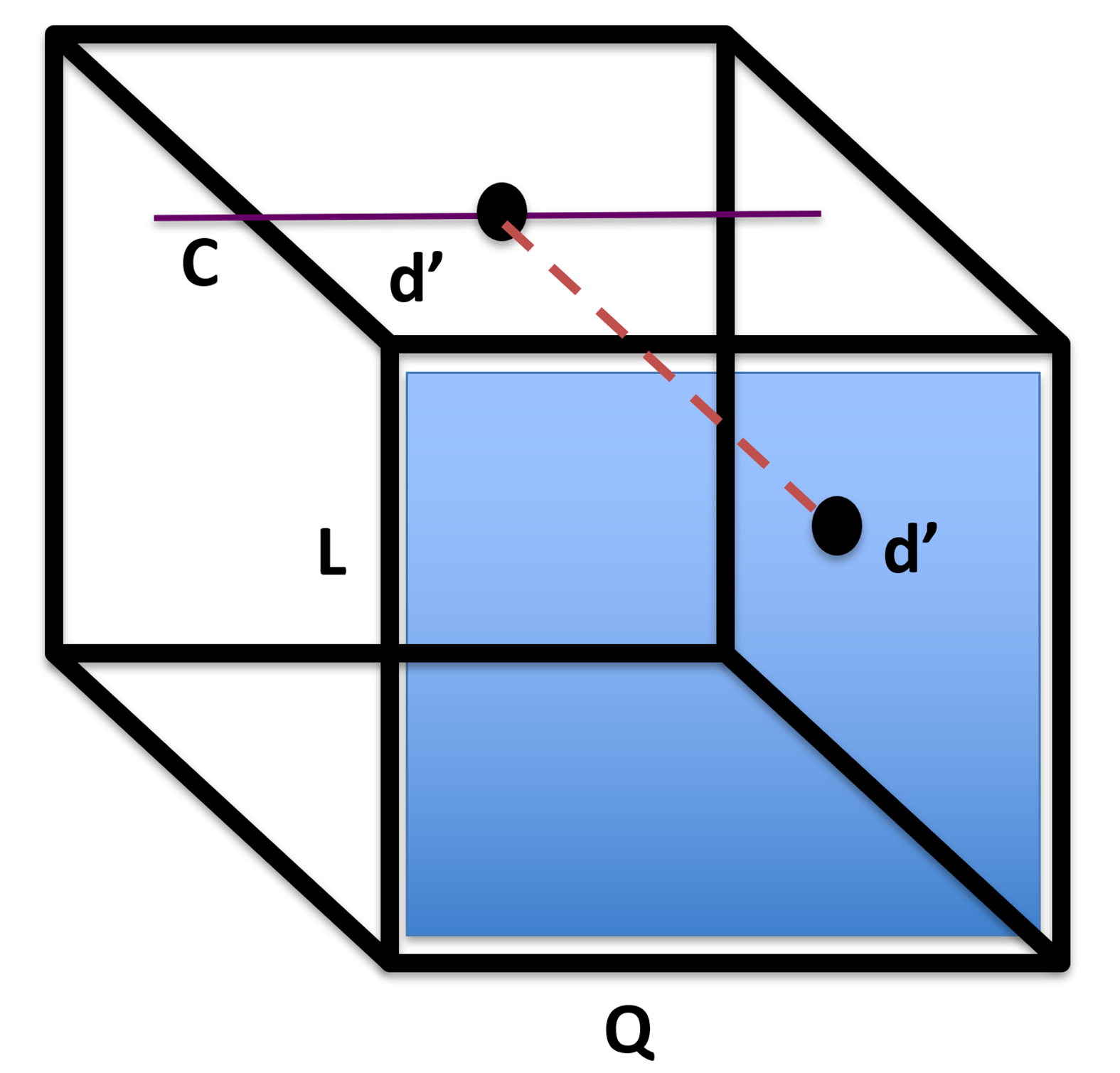}\label{fig:proj}
	}
	\caption{ \ADPaR}
	\label{CC}
\end{figure}

\subsection{Algorithm \ADPaRA}
Our treatment is geometric and exploits the monotonicity of our objective function (Equation~\ref{opt} in Section~\ref{sec:pbms}). Even though the original problem is defined in a continuous space, we present a discretized technique that is exact.  \ADPaRA, employs three sweep-lines~\cite{de1997computational}, one for each parameter, quality, cost, and latency and gradually relaxes the parameters to produce the tightest alternative parameters that admit $k$ strategies. By its unique design choice,  \ADPaRA is empowered to select the parameter that is most suitable to optimize the objective function, and hence, produces exact solutions to \ADPaR.

\ADPaRA has four main steps. Before getting into those details, we present a few simplifications to the problem for the purpose of elucidation. As we have described before, we normalize quality, cost, latency thresholds of a deployment or of a strategy in $[0,1]$, and inverse quality to $(1-quality)$. This step is just for unification, making our treatment for all three parameters uniform inside \ADPaR, where smaller is better, and the deployment thresholds are considered as upper-bounds.  With this, each strategy is a point in a $3$-dimensional space and a deployment parameter (modulo its cardinality constraint) is an axis-parallel hyper-rectangle\cite{de1997computational}  in that space.  Consider Figure~\ref{fig:3dspace} that shows the $4$ strategies in Example~\ref{ex1} and $d_2$ as a hyper-rectangle.

Step-1 of \ADPaRA computes the relaxation (increment) that a deployment requires to satisfy a strategy among each deployment parameter. This is akin to computing $s_i.cost - d_2.cost$ (likewise for quality and latency) and when the strategy cost is smaller than the deployment threshold, it shows no relaxation is needed - hence we transform that to $0$. The problem is studied for quality, cost, and latency (referred to as {\tt Q}, {\tt C}, {\tt L}) (Table~\ref{sp1}). It also initializes $d'= \{1,1,1\}$, the worst possible relaxation.

\smallskip \noindent Step-2 of \ADPaRA involves sorting the strategies based on the computed relaxation values from step-1 in an increasing order across all parameters, as well as keeping track of the index of the strategies and the parameters of the relaxation values. The sorted relaxation scores are stored in list $R$, the corresponding $I$ data structure provides the strategy index, and $D$ provides the parameter value. In other words, $R[j]$ represents the $j$-th smallest relaxation value, where  $I[j]$ represents the index of the strategy and $D[j]$ represents the parameter value corresponding to that. A cursor $r$ is initialized to the first position in $R$ (Table~\ref{sp2}).
Another data structure, a  boolean matrix $M$ of size $|\mathcal{S}| \times 3$ (Table~\ref{csm}) is used that keeps track of the number of strategies that are covered by the current movement of cursor $r$ in list $R$. This matrix is initialized to $0$ and the entries are updated to $1$, as $r$ advances.

\smallskip \noindent Step-3 involves designing three sweep-lines along  {\tt Q}, {\tt C}, and {\tt L} (Table~\ref{sp3}). A  sweep line is an imaginary vertical line which is swept across the plane rightwards. The  {\tt Q} sweep-line sorts the $\mathcal{S}$ in  {\tt C L} plane in  increasing order of  {\tt Q} (the other two works in a similar fashion).  \ADPaRA sweeps the line as it encounters strategies, in order to discretize the sweep. At the beginning, each sweep-line points the $k$-th strategy along {\tt Q}, {\tt C}, {\tt L}, respectively.  $d'$ is updated and contains the current {\tt Q}, {\tt C}, {\tt L} value i..e, $d'.quality= {\tt Q}$, $d'.cost= {\tt C}$, and $d'.latency= {\tt L}$.  Cursor $r$ points to the smallest of these three values in R. Matrix $M$ is updated to see what parameters of which strategies are covered so far.

\smallskip \noindent  At step-4, \ADPaRA checks if the current $d'$ covers $k$ strategies or not.  This involves reading through $I$ and checking if there exists $k$ strategies such that for each strategy $s.quality \leq d'.quality$ and $s.cost \leq d'.cost$ and $s.latency \leq d'.latency$.  If there  are not $k$ such strategies, it advances $r$ to the next position and resets $d'= \{1,1,1\}$ again.

If there are more than $k$ strategies, the new $d'$, however, does not ensure that it is the tightest one to optimize Equation~\ref{multiobjective-prob}. Therefore, \ADPaRA cannot halt.  \ADPaRA needs to check if there exists another $d^{''}$ that still covers $k$ strategies better than $d'$. This can indeed happen as we are dealing with a $3$-dimensional problem and these three values in combination determine the objective function.

\ADPaRA takes turn in considering the current values of each parameter based on $d'$, and creates a  projection on the corresponding $2$-D plane, for the fixed value of the third parameter. Figure~\ref{fig:proj} shows an example in  ({\tt Q, L}) plane for a fixed cost. It then considers all strategies whose $s.cost \leq d'.cost$.  After that, it finds the largest expansion among the two parameters such that this new $d^{''}$ covers $k$ strategies. This gives rise to three new deployment parameters, $d^{''}_{C}$, $d^{''}_{Q}$, $d^{''}_{L}$. It chooses the best of these three and updates $d'$. At this point it checks if $M$ has $k$ strategies covered. If it does, it stops processing and returns the new $d'$ and the $k$ strategies. If it does not, it advances the cursor $r$ to the right.

Using Example~\ref{ex1}, the alternative parameters are $(0.75,0.5,$ $0.28)$ for $d_2$ and  $s_1, s_2, s_3$ are returned.

\begin{table}[h!]
	\centering
	\begin{tabular}{|l|l|l|l|}
\hline
&Cost&Quality&Latency \\ \hline
$s_1$&0&0&1  \\ \hline
$s_2$&0&0&1  \\ \hline
$s_3$&0&0&0  \\ \hline
$s_4$&0&0&0  \\ \hline
\end{tabular}
\caption{matrix $M$}\label{csm}
	\begin{tabular}{|l|l|l|l|}
		\hline
		&Cost&Quality&Latency \\ \hline
		$s_1$&0.3&0.05&0  \\ \hline
		$s_2$&0.05&0.13&0  \\ \hline
		$s_3$&0&0.3&0  \\ \hline
		$s_4$&0&0.38&0  \\ \hline
	\end{tabular}
		\caption{Step 1}\label{sp1}
			\begin{tabular}{|l|l|l|l|l|l|l|}
		\hline
		Relaxation $R$&0&\underline{0}&0&0&0&0 \\ \hline
		Strategy Index $I$&1&\underline{2}&3&4&3&4 \\\hline
		Parameter $D$&L&\underline{L}&L&L&C&C \\\hline
		Relaxation R&0.05&0.05&0.13&0.3&0.3&0.38 \\ \hline
		Strategy Index $I$&1&2&2&1&3&4 \\\hline
		Parameter $D$&Q&C&Q&C&Q&Q \\\hline
	\end{tabular}
		\caption{Step 2}\label{sp2}
			\begin{tabular}{|l|l|l|l|l|l|}
		\hline
		sweep-line(Q)&$C,L$ plane&0.05&0.13&0.3&0.38\\\hline
		&$s.cost$ &0.3&0.05&0&0\\\hline
		&$s.latency$ &0&0&0&0\\\hline
		sweep-line(C)&$Q,L$ plane &0&0&0.05&0.3\\\hline
		&$s.quality$ &0.38&0.3&0.13&0.05\\\hline
		&$s.latency$ &0&0&0&0\\\hline
		sweep-line(L)&$C,Q$ plane &0&0&0&0\\\hline
		&$s.cost$ &0.3&0.05&0&0\\\hline
		&$s.quality$ &0.05&0.13&0.3&0.38\\\hline
	\end{tabular}
		\caption{Step 3}\label{sp3}
\end{table}

\begin{lemma}\label{th1}
 To cover $k$ strategies, $d'$ needs to be initialized  at least to the  $k^{th}$ smallest values on each parameter.
\end{lemma}
\begin{proof}
	We prove this by the method of contradiction. Assume there is an alternative deployment $d'$, which has a a parameter that is smaller than the corresponding $k^{th}$ smallest value. On the other hand, to be able to satisfy as an alternative deployment parameter, $d'$ must cover at least $k$ strategies on each parameter per its definition. However, to be able to cover at least  $k$ strategies, each parameter of $d'$ must be equal or larger than the corresponding parameters of the strategies that it covers. Hence, the contradiction.  
		
		In Figure \ref{p2}, the red dot shows an alternative deployment request for $d_2$ in Example 1. Figure \ref{p3} is the 2D projection of Figure \ref{p2}. It illustrates further, to cover $3$ strategies, the alternative $d_2$ has to cover at least $3$ strategies on each dimension.  
	\end{proof}
		\begin{figure}[h]
		\centering
		\includegraphics[width=6cm,height=6cm]{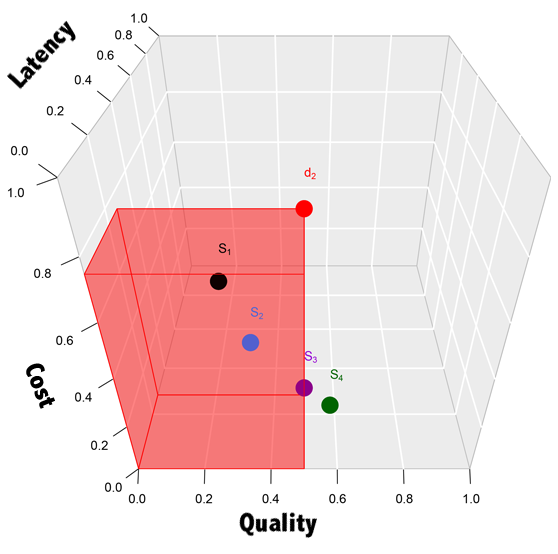}
		\caption{The returned alternative deployment for $d_2$ and  corresponding strategies}
		\label{p2}
	\end{figure}
	\begin{figure}[h]
		\centering
		\includegraphics[width=6cm,height=6cm]{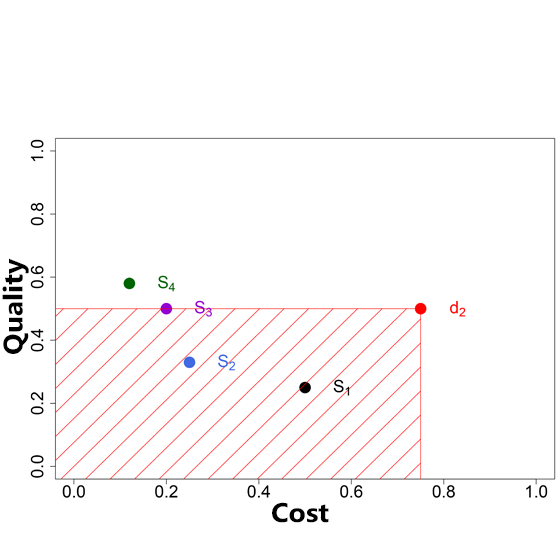}
		\caption{The 2D projection of the returned alternative deployment for $d_2$ and corresponding strategies}
		\label{p3}
	\end{figure}

\begin{lemma}\label{th2}
Going by the relaxation value and parameter order of $R$ and $D$, it ensures the tightest increase in the objective function in \ADPaRA.
\end{lemma}
	\begin{proof}
The proof is done by induction on the number of relaxation. 

\smallskip \noindent		

{\bf The basis (base case):} we first prove that the statement is true for the very first relaxation.
After the relaxation in step 1, all the values are sorted in increasing order of the  parameters, as shown in Table \ref{sp2}. If there are $n$ strategies under consideration, each with its corresponding quality, cost, and latency parameters, then $\{r_1,...,r_{3n} =R\}$ and $\{d_1,...,d_{3n} = D\}$ are sorted relaxation values and corresponding parameters $r_1\le r_2\le\dots\le r_{3n}$. Assume the scan starts at $d'=$\{ $r_a$, $r_b$ and $r_c$\}. The objective value of the initial setting is $\sqrt{r_a^2 + r_b^2 + r_{c}^2}$, since the objective function is $\sqrt{Cost^2 + Quality^2 + Latency^2}$. 

\smallskip \noindent
Let $r_d$ be the new cursor position  of $d'$, once one of the current parameter (without loss of generality, let $r_a$ be that one) of  $d'$ gets relaxed. Since (Recall Table~\ref{sp2}) $R$ is sorted in increasing order of parameter values of the strategies, after the first relaxation, the increment in the objective function value is smaller than that of any other relaxation. Thus, the statement is true for the first relaxation. \\
\smallskip \noindent		
{\bf Inductive Hypothesis:} Assume the statement to be true for the first $x$ relaxations: \\
\smallskip \noindent		
{\bf Inductive Step:} Now we intend to prove that the statement is true for the $x+1$-th relaxation.	Assume, at the $x$-th iteration, $d' =$\{$r_a', r_b'$,$r_c'$\}. Assume $r_d'$ is the next parameter value in $R$. Therefore, $r_d'$ will be scanned at the $x+1$-th iteration. Since $r_d'$ is the smallest relaxation on the corresponding parameter, therefore the increase in the objective function value from \{$r_a', r_b'$,$r_c'$\} to \{$r_d', r_b'$,$r_c'$\} is smaller than all other relaxation. Therefore, Lemma \ref{th2} is true for the $x+1$-th iteration. Hence the proof.
		
	\end{proof}


\begin{theorem}\label{th3}
\ADPaRA produces an exact solution to the \ADPaR problem.
\end{theorem}
	\begin{proof}
		For this, we need to prove that there is no alternative deployment parameters $d^{''}$ that has a smaller objective value than that of $d'$ and still covers $k$ strategies. We prove this by contradiction.
	Assume that is indeed the case, that is, there is an alternative deployment parameter $d^{''}$ whose objective value is smaller than that of $d'$'s and it still covers $k$ strategies. Inside Algorithm \ADPaRA,  when a parameter of $d'$ needs to be updated, it is increased to the corresponding value of its closest strategy, based on Lemma 2, i.e., update $d'$ based of $R$ and $D$, as described in Algorithm \ADPaRA. Therefore, in such cases, \ADPaRA will decide not to relax the next parameter, and $d^{''}$ will be returned, instead of $d'$. However, that is not the case. Hence the contradiction and the proof.
	
In fact, by using the sweep-lines, all possible alternative deployment parameters  that have smaller objective value (than that of $d'$) are checked before $d'$ is returned, just like the Figure \ref{th} shows. All the possible alternative deployment requests in the unchecked area have a larger objective value.  
	\end{proof}
	
		\begin{figure}
		\centering
		\includegraphics[width=6cm,height=6cm]{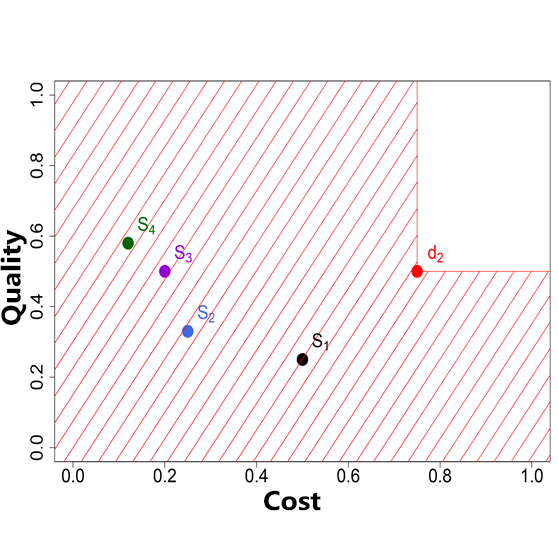}
		\caption{The scanned area when  $d'$ is returned}
		\label{th}
	\end{figure}
	
%

{\bf Running Time:}  Step-1 of Algorithm \ADPaRA takes $O(|\mathcal{S}|)$. Step-2 and 3 are dominated by sorting time, which takes $O(|\mathcal{S}| \ log|\mathcal{S}|)$. Step-4 is the most time-consuming and takes $O(|\mathcal{S}^3|)$. Therefore, the overall running time of the algorithm is cubic to the number of strategies.

\begin{algorithm}[htpb]
\caption{Algorithm \ADPaRA for alternative deployment parameter recommendation}\label{checkk}
	\begin{algorithmic}[1]
		\Require$\mathcal{S}$, $k$, $W$, $d$, $k$.
         \State Compute relaxation values $s.quality - d.quality$, $s.cost - d.cost$, $s.latency-d.latency$, $\forall s \in \mathcal{S}$.
         \State Compute $R$ by sorting $3|\mathcal{S}|$ numbers in increasing order.
         \State Compute $I$ and $D$ accordingly.
        \State Initialize $M$ to all $0$'s and $d'=\{1,1,1\}$
        \State Initialize Cursor $r = R[0]$
         \State Sort ({\tt C L}) , ({\tt Q L}), and ({\tt Q C}) planes based on the {\tt Q}, {\tt C}, {\tt L} sweep-lines respectively.
         \State $x$= $k$-th value in ( {\tt C L}), $y$= $k$-th value in ( {\tt Q L}), $z$= $k$-th value in ({\tt Q C}) plane
         \State Update $d'$= $\{x,y,z\}$
         \State $r$ = minimum $\{x,y,z\}$
         \State Update matrix $M$
            \If{$d'$ covers $ \ge k$ strategies}
             \State Compute the best $d^{''}$ better than $d'$ that covers $k$ strategies
             \If {$M$ covers $k$ strategies}
             \State  $d'=d^{''}$ and return
             \EndIf
             \If {$M$ covers  $<k$ strategies}
             \State  move $r$ to the right
             \EndIf
            \EndIf
            \If{$d'$ covers $< k$ strategies}
            \State Move $r$ to the right
            \State Update $d'$'s one of the parameters by consulting $R$ and $D$
            \EndIf
         \State go back to line $10$
            \end{algorithmic}
\end{algorithm}

\section{Experimental Evaluation}\label{exp}
In our real-world deployments, we estimate worker availability and demonstrate the need for optimization (Section~\ref{realdata}). In synthetic data experiments (Section~\ref{synExp}), we present results to validate the qualitative and scalability aspects of our algorithms.

\subsection{Real Data Experiments}\label{realdata}
We perform two different real data experiments that involve workers from AMT focusing on text editing tasks. The first experiments (Section~\ref{realdata1}) empirically validate key assumptions in designing \StratRec. the second experiments (Section~\ref{realdata2}) validate the effectiveness of \StratRec when compared to the case where no recommendation is made.

\subsubsection{Validating Key Assumptions}\label{realdata1}
We consider two types of tasks: a) sentence translation (translating from English to Hindi) and text creation (writing 4 to 5 sentences on some topic) to validate the following questions:

\noindent {\em 1. Can worker availability be estimated and does it vary over time?} We performed $3$ different deployments for each task. The first deployment was done on the weekend (Friday 12am to Monday 12am), the second deployment was done at the beginning to the middle of the week (Monday to Thursday), the last one is from the middle of the week until the week-end (Thursday to Sunday). 
We design the HITs (Human Intelligence Tasks) in AMT such that each task needs to be undertaken by a maximum number of workers $x$. Worker availability is computed as the ratio of $\frac{x'}{x}$, where $x'$ is the actual number of workers who undertook the task during the deployment time (although this does not fully conform to our formal worker availability definition, it is our sincere attempt to quantify worker availability using public platforms).

\noindent {\em 2. How does worker availability impact deployment parameters?} We need to be able to calculate the quality, cost, and latency, along with worker availability. Latency and cost are easier to calculate, basically, it is the total amount of money that was paid to workers and the total amount of time the workers used to make edits in the document. Since text editing tasks are knowledge-intensive, to compute the quality of the crowd contributions, we ask a domain expert to judge the quality completed tasks as a percentage.
Once worker availability, quality, cost, and latency are computed, we perform curve fitting that has the best fit to the series of data points.

\noindent {\em 3. How do deployment strategies impact different task types?} We deployed both types of text editing tasks using two different deployment strategies {\em SEQ-IND-CRO} and {\em SIM-COL-CRO}  that were shown to be effective with more than 70\% of quality score for short texts~\cite{borromeo2017deployment}. Since our effort here was to evaluate the effectiveness of these two strategies considering quality, cost, and latency, we did not set values for deployment parameters and we simply observed them through experimentation.  

{\bf Tasks and Deployment Design:} We chose three popular English nursery rhymes for sentence translation. Each rhyme consisted of 4-5 lines that were to be translated from English to Hindi (one such sample rhyme is shown in Figure~\ref{fig:sample-sentence}).  For text creation, we considered three popular topics, {\em Robert Mueller Report}, {\em Notre Dame Cathedral}, and {\em 2019 Pulitzer prizes}. One sample text creation is shown in Figure~\ref{fig:sample-text}.


\begin{figure}
	\centering
	\begin{minipage}{0.45\textwidth}
	\includegraphics[width=0.9\textwidth,height=8cm]{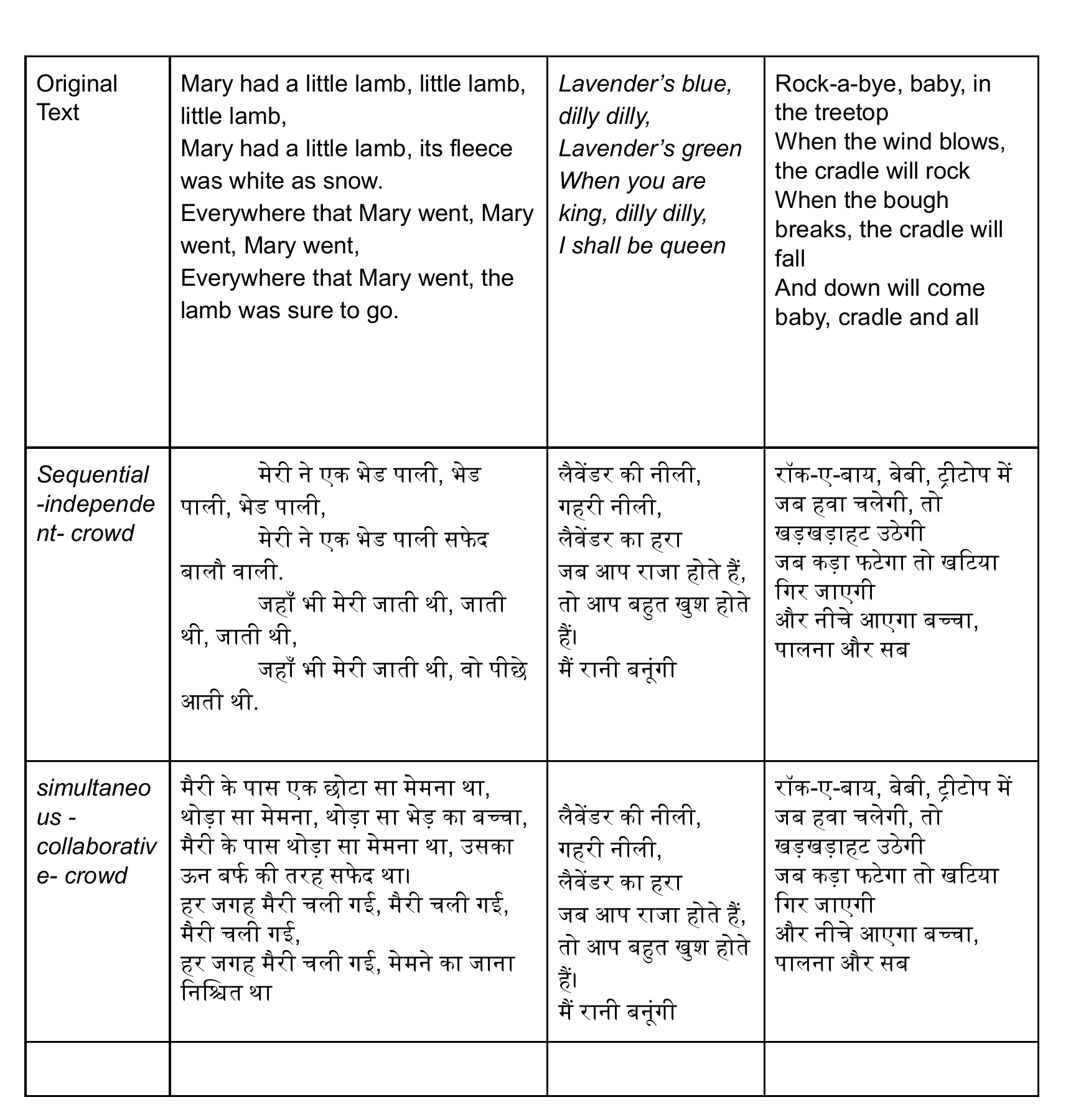}
	\caption{Translation: Original Texts and Translation}
	\label{fig:sample-sentence}
	\end{minipage}\hfill
	\begin{minipage}{0.45\textwidth}
	\centering
	\includegraphics[width=0.9\textwidth,height=8cm]{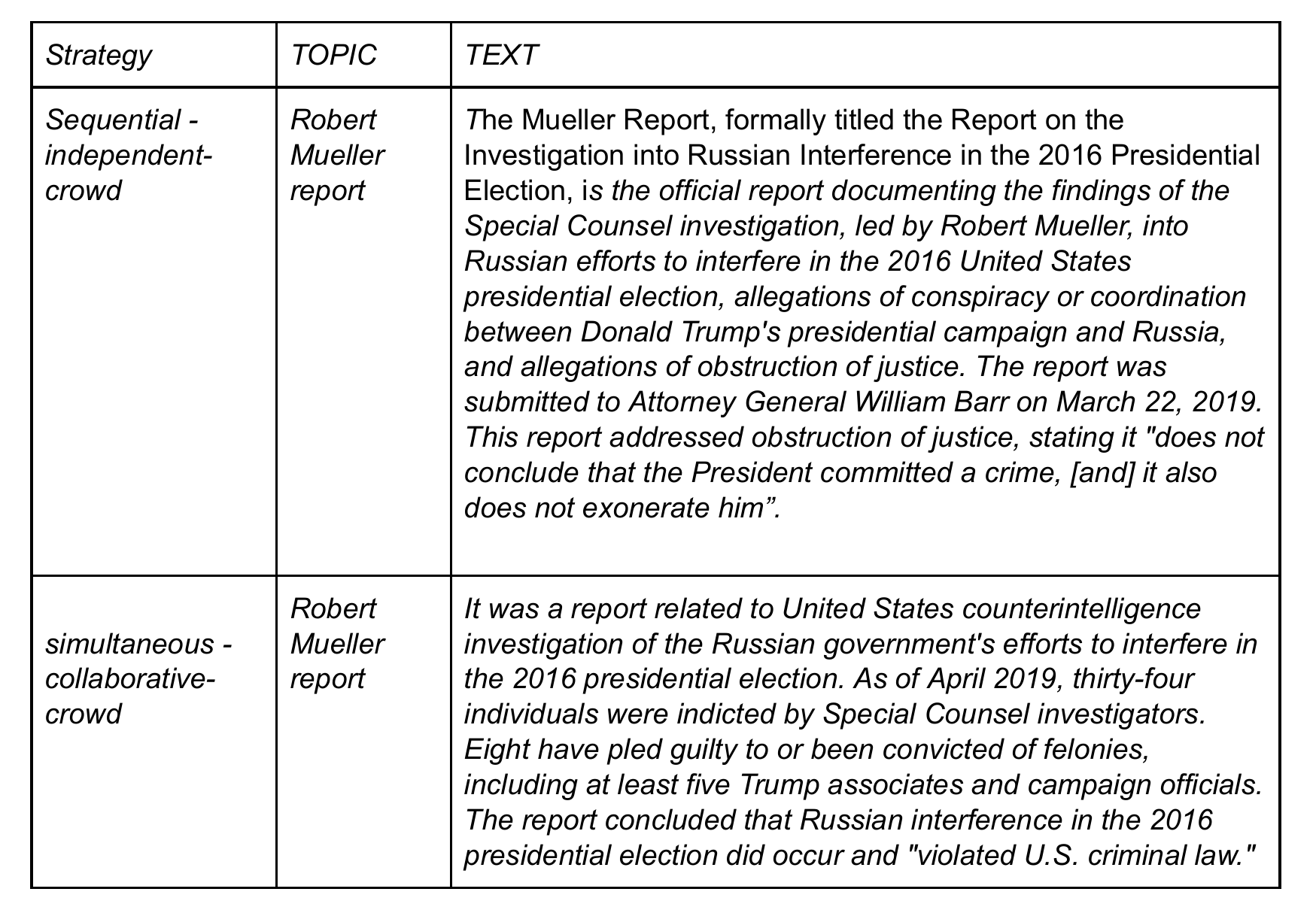}
	\caption{Text Creation: Robert Mueller Report}
	\label{fig:sample-text}
	
	\end{minipage}\hfill
\end{figure}

\begin{figure}
	\centering
	\includegraphics[width=7cm,height=4cm]{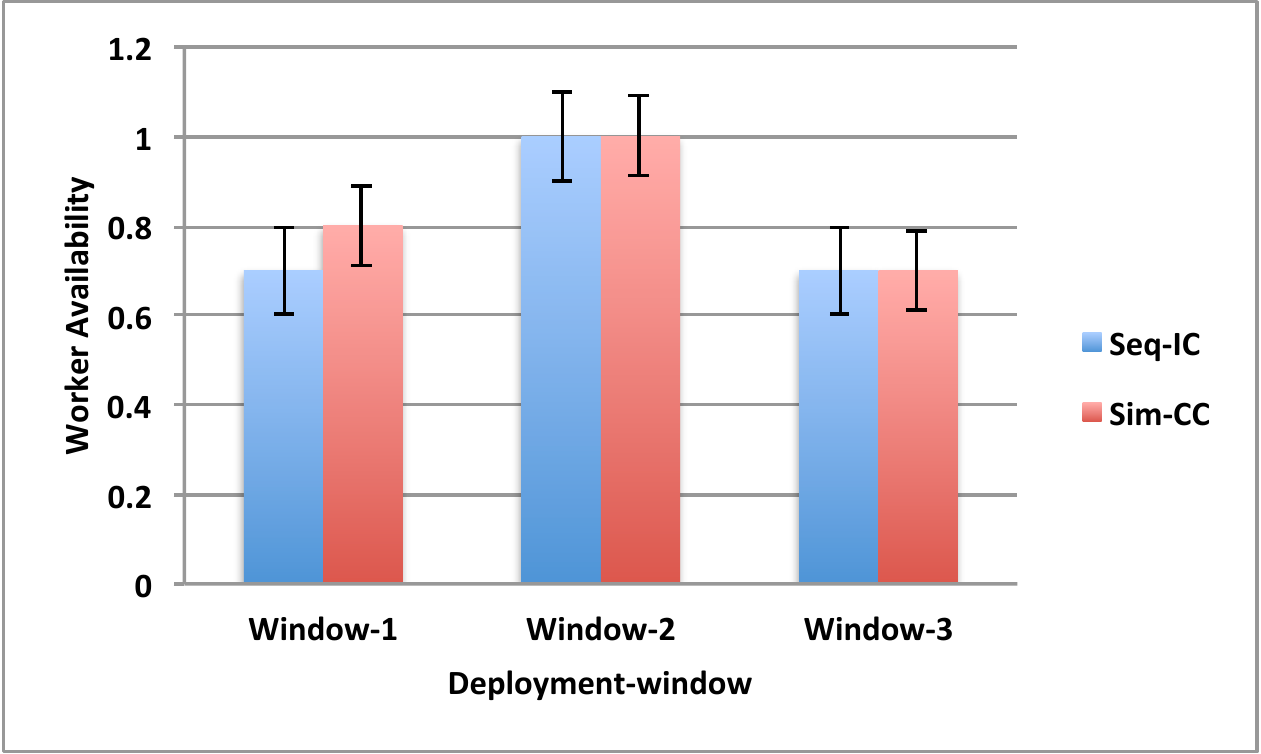}
	\caption{\small Worker Availability Estimation}
	\label{fig:wa}
\end{figure}

We designed three deployment windows at different days of the week. Unlike micro-tasks in AMT, text editing tasks require significantly more time to complete (we allocated $2$ hours per HIT). A HIT contains  either 3 sentence translation tasks or three text creation tasks as opposed to micro-tasks, where a HIT may contain tens of tasks. For each task type, we validated $2$ deployment strategies - in {\em SEQ-IND-CRO}, workers were to work in sequence and independently, whereas, in {\em SIM-COL-CRO}, workers were asked to work simultaneously and collaboratively. We created $2$ different samples of the same study resulting in a total of $8$ HITs deployed inside the same window. Each HIT was asked to be completed by $10$ workers paid $\$2$ each if the worker spent enough time (more than $10$ minutes). This way, a total of $80$ unique workers were hired for each deployment window, and a total of $240$ workers were hired for all three deployments.

{\bf Worker Recruitment:} For both task types, we recruited workers with a HIT approval rate greater than $90\%$. For sentence translation, we additionally filtered workers on geographic locations, either US or India. For text creation tasks, we recruited US-based workers with a Bachelor's degree. 

{\bf Enabling collaboration:} After workers were recruited from AMT, they were directed to Google Docs where the tasks were described and the workers were given instructions. The docs were set up in editing mode, so edits could be monitored. 

{\bf Experiment Design:} An experiment is comprised of three steps. In Step-1, all initially recruited workers went through qualification tests. For text creation, a  topic (Royal Wedding) was provided and the workers were asked to write $5$ sentences related to that topic. For sentence translation, the qualification test comprised of $5$ sample sentences to be translated from English to Hindi. Completed qualification tests were evaluated by domain experts and workers with more than 80\% or more qualification scores were retained and invited to work on the actual HITs. In Step-2, actual HITs  were deployed for $72$ hours and the workers were allotted $2$ hours for the tasks. In Step-3, after $72$ hours of deployment, results were aggregated by domain experts to obtain a quality score. Cost and latency were easier to calculate directly from the raw data.

{\bf Summary of Results:} 
{\bf Our first observation is} that worker availability {\em can be estimated and does vary over time} (standard error bars added). We observed that for both task types, workers were more available during Window 2 (Monday-Thursday), compared to the other two windows. Detailed results are shown in Figure~\ref{fig:wa}. 

{\bf Our second observation} is that each deployment parameter has a linear relationship with worker availability for text editing tasks. Quality and cost increases linearly with worker availability. Latency decreases with increasing worker availability. This linear relationship could be captured and the parameters $(\alpha,\beta)$ could be estimated. Table~\ref{mc} presents these results and the estimated $(\alpha,\beta)$ always lie within $90\%$ confidence interval of the fitted line.

{\bf Our final observation} is that {\em SEQ-IND-CRO} performs better than {\em SIM-COL-CRO} for both task types. However, this difference is not statistically significant. On the other hand, {\em SEQ-IND-CRO} has higher latency. Upon further analysis, we observe that when workers are asked to collaborate and edit simultaneously, that gives rise to an edit war and an overall poor quality. Figure~\ref{relationship} presents these results.

\begin{table}[h]
\centering
	\begin{tabular}{ |l|l|l| }
		\hline
		\multicolumn{3}{ |c| }{Worker Availability and Deployment Parameters} \\
		\hline
		Task-Strategy & Parameters & $\alpha$,$\beta$ \\ \hline
		\multirow{3}{*}{Translation {\em SEQ-IND-CRO}} & Quality & $0.09$,  $0.85$ \\
		& Cost & $1.00,0.00$ \\
		& Latency & $-0.98,1.40$ \\ \hline
		\multirow{3}{*}{Translation {\em SIM-COL-CRO}} & Quality & $0.09,0.82$ \\
		& Cost &$0.82,0.17$ \\
		& Latency & $-0.63,1.01$ \\ \hline
		\multirow{3}{*}{Creation {\em SEQ-IND-CRO}} & Quality & $0.10,0.80$ \\
		& Cost & $1.00,0.00$  \\
		& Latency & $-1.56,2.04$ \\ \hline
		\multirow{3}{*}{Creation {\em SIM-COL-CRO}} & Quality & $0.19,0.70$ \\
		& Cost & $1.00,-0.00$ \\
		& Latency & $-1.38,1.81$ \\ \hline
	\end{tabular}
	\vspace{0.1in}
	\caption{ $\alpha,\beta$ Estimation}\label{mc}
\end{table}

\begin{figure*}
	\subfloat[Translation {\em SEQ-IND-CRO}]{
		\includegraphics[height=3cm, width=.22\textwidth]{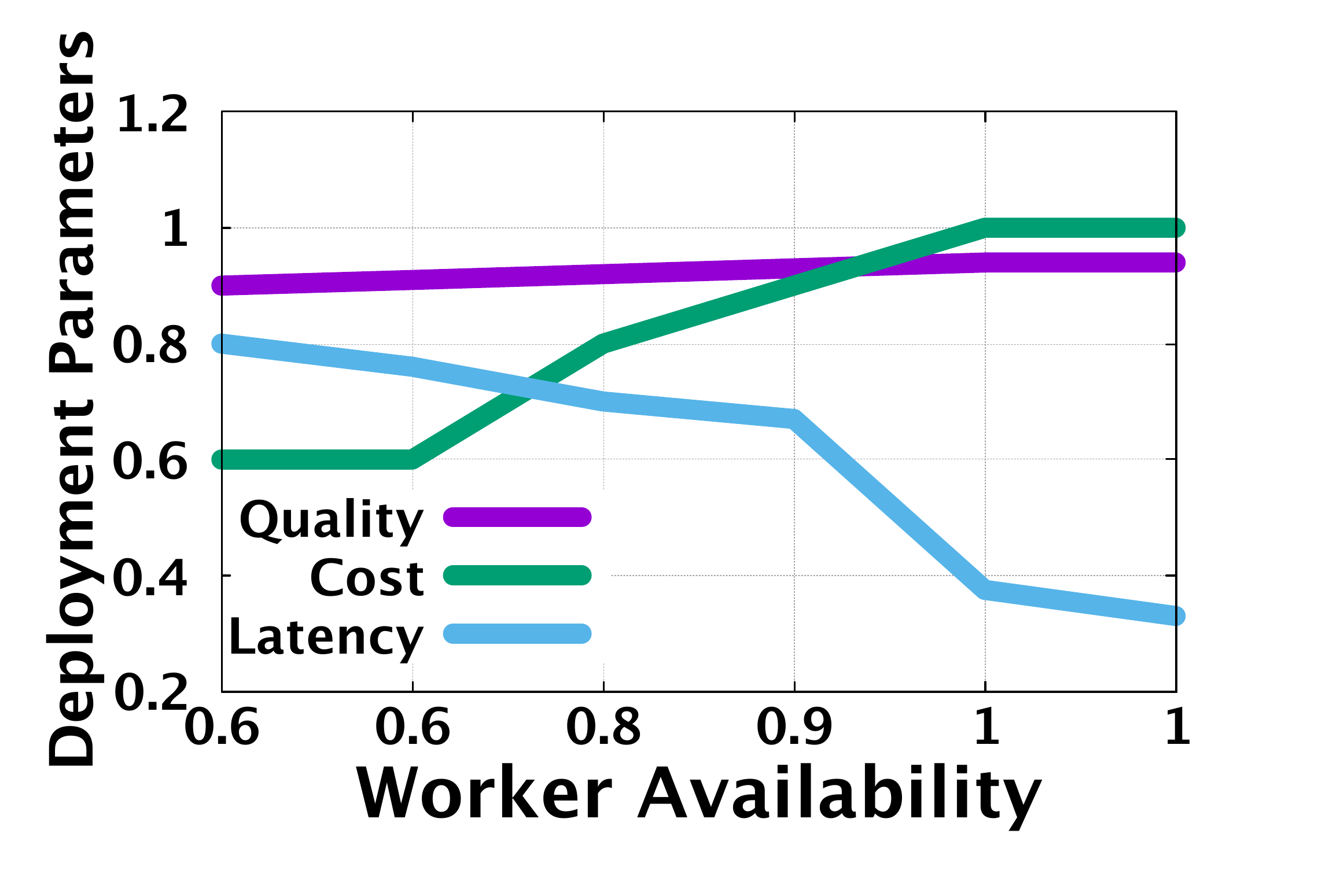}
	}
	\hfill
	\subfloat[Translation {\em SIM-COL-CRO}]{
		\includegraphics[height=3cm, width=.22\textwidth]{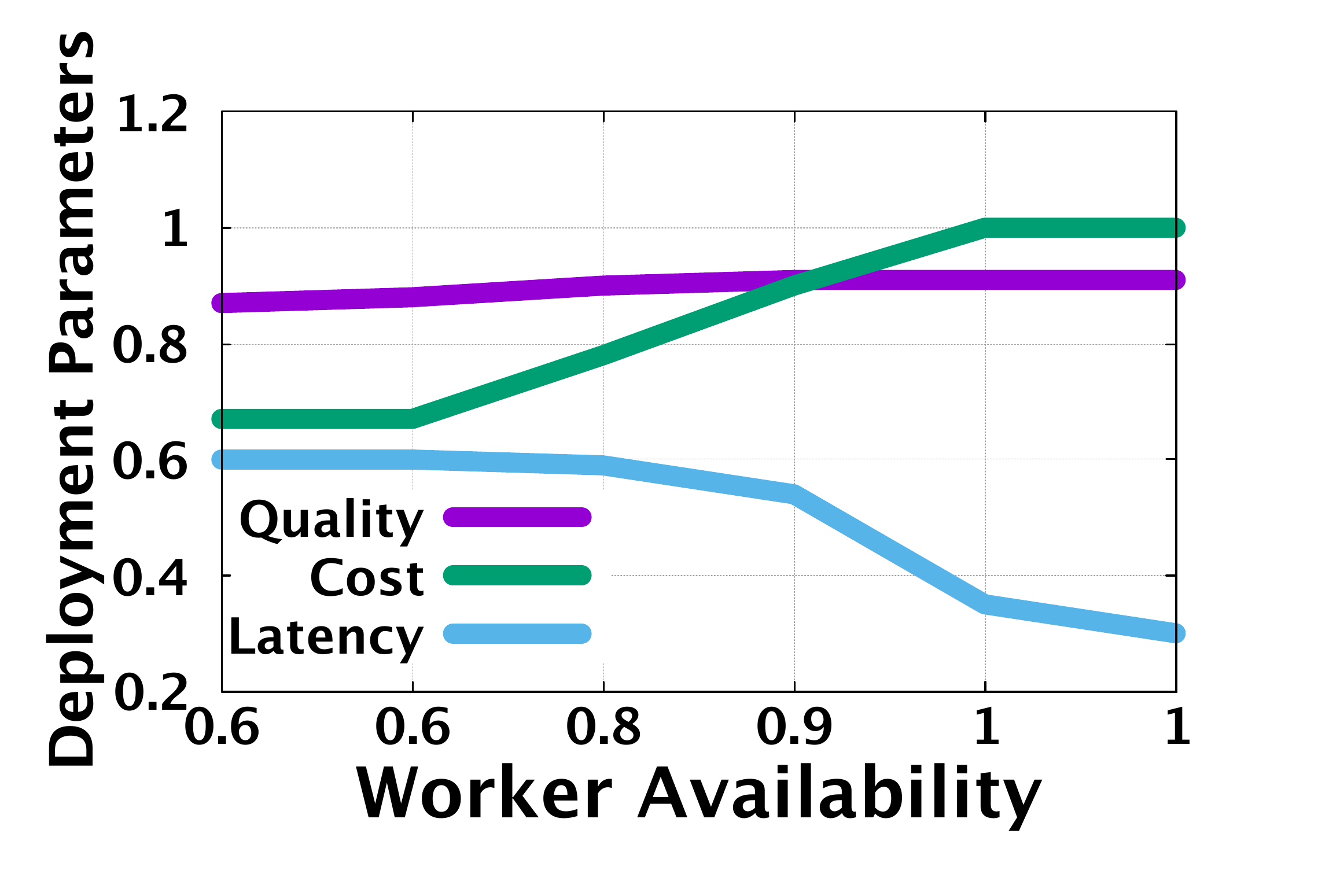}
	}
	\hfill
	\subfloat[Creation {\em SEQ-IND-CRO}]{
		\includegraphics[height=3cm, width=.22\textwidth]{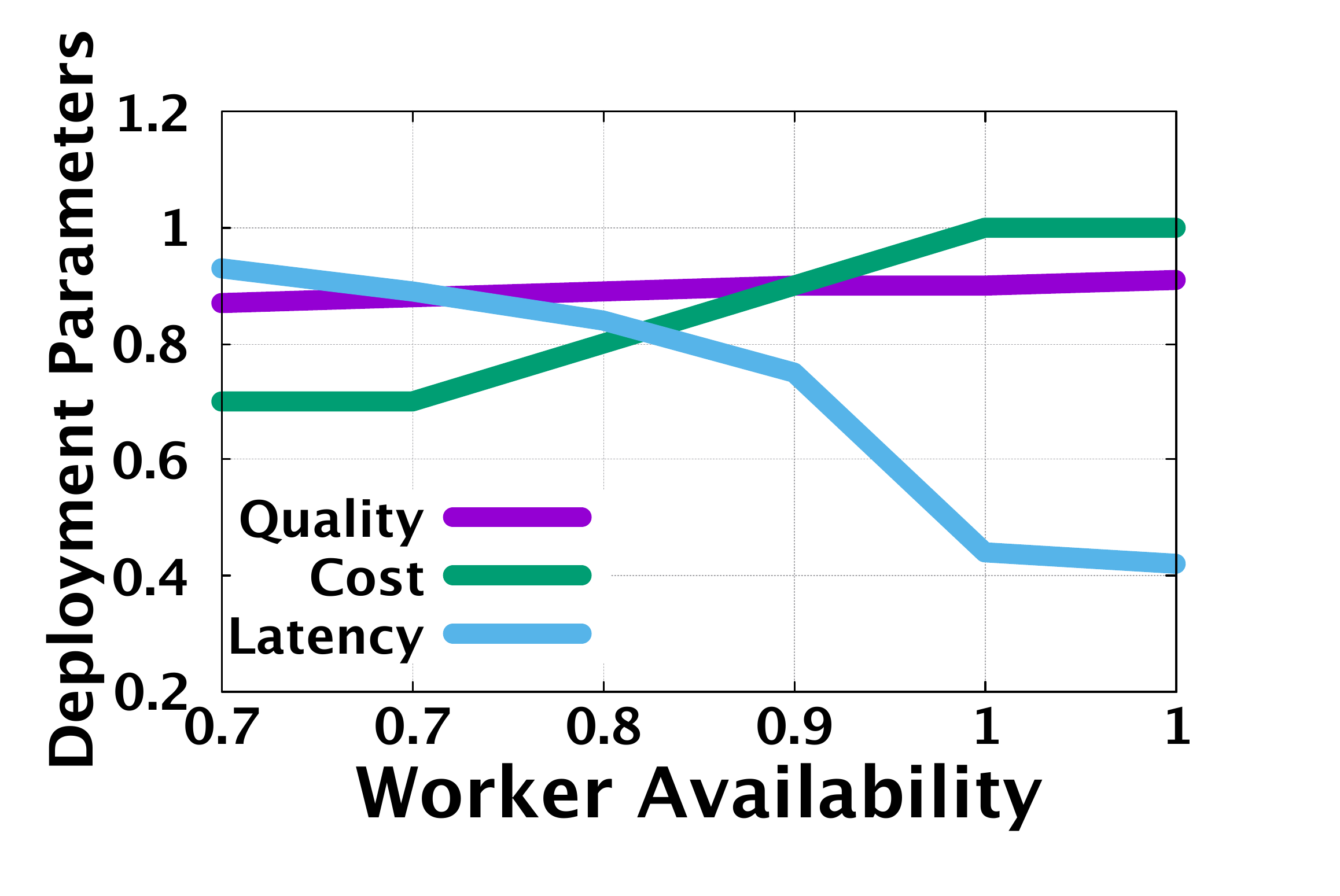}
	}
	\hfill
	\subfloat[Creation {\em SIM-COL-CRO}]{
		\includegraphics[height=3cm, width=.22\textwidth]{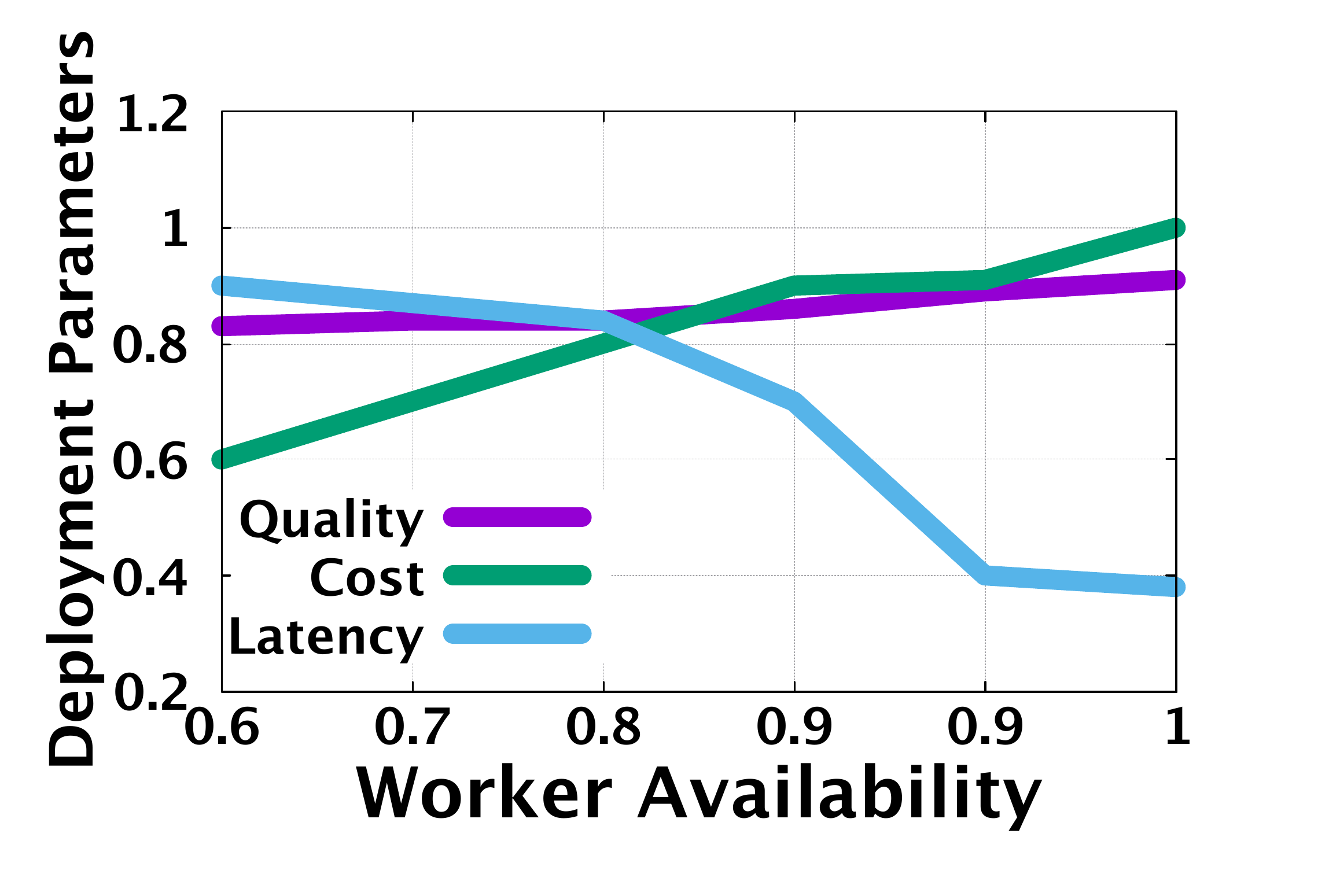}
	}
	\caption{ Relationship Between Deployment Parameters and Worker Availability}
	\label{relationship}
\end{figure*}

\subsubsection{Validating the Effectiveness of \StratRec} \label{realdata2}
We are unable to ask specific user (task designer's) satisfaction questions in this experiment, simply because AMT does not allow  to recruit additional task designers and only workers could be recruited. 
For this purpose, we deploy $10$ additional sentence translation (translating nursery rhymes from English to Hindi) and $10$ additional text creation tasks considering a set of $8$ strategies.  

We create $2$ mirror deployments for the same task (one according to \StratRec recommendation and the other without) resulting in a total of $40$ HITs deployed. For the latter scenario, the deployments were not recommended any structure, organization, or style and the workers were given the liberty to complete the task the way they preferred. Each HIT was asked to be completed by $7$ workers paid $\$2$ each if the worker spent enough time (more than $10$ minutes). This way, a total of $280$ unique workers are hired during this experiment. The quality, cost, and latency thresholds of each deployment are set to be $70\%$, $\$14$, $72$ hours. 

The worker recruitment, and the rest of the experiment design, and result aggregation steps are akin to those steps that are described in Section~\ref{realdata1}. Figure~\ref{recnorec} represents the average quality, cost, and latency results of these experiments with statistical significance.

\begin{figure}
	\centering
	\includegraphics[width=10cm,height=5cm]{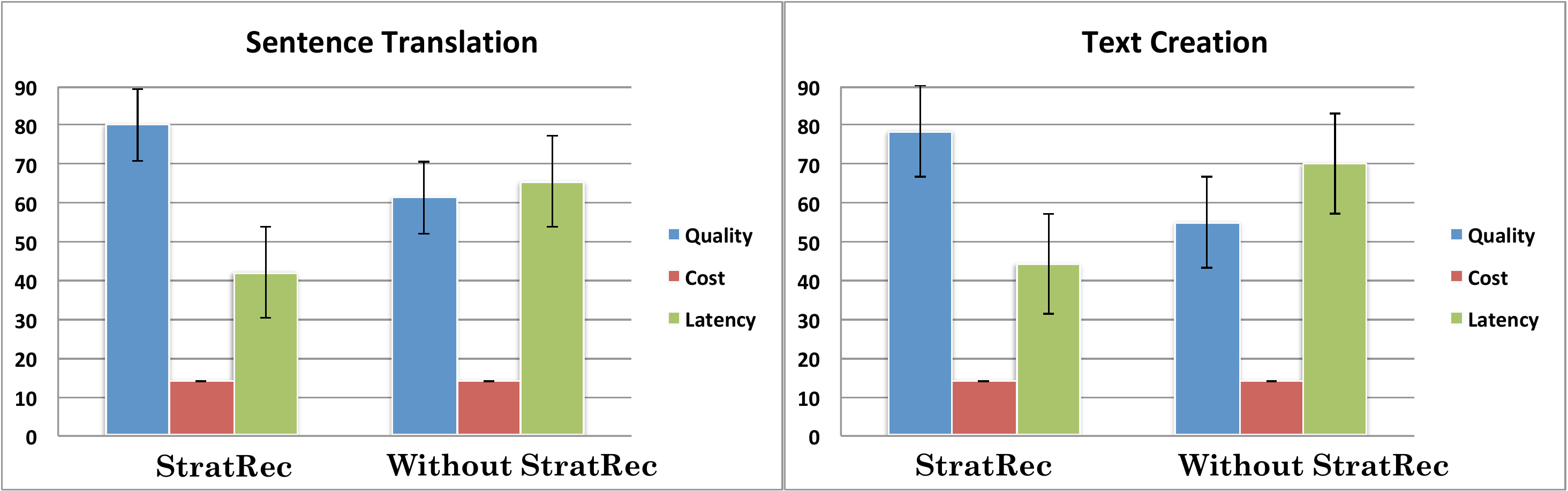}
	\caption{\small Average Quality, Cost, Latency Comparison of Deployments with and without \StratRec}
	\label{recnorec}
\end{figure}

{\bf Summary of Results:} We have two primary observations from these experiments.
{\bf Our first observation is} that (Figure~\ref{recnorec}), when tasks are deployed considering recommendation of \StratRec, with statistical significance, they achieve higher quality and lower latency, under the fixed cost threshold on an average compared to the deployments that do not consult \StratRec. These results validate the effectiveness of deployment recommendations of our proposed frameworks and its algorithms.

{\bf Our second observation} (upon further investigating the Google Docs where the workers undertook tasks),  is that the deployments that do not consider \StratRec recommendations have more edits, compared to that are deployed considering \StratRec. In fact, on average, \StratRec deployments have an average of $3.45$ edits for sentence translation, compared to $6.25$ edits on average for those deployed with no recommendations. Indeed, when workers were not guided, they repeatedly overrode each other's contributions, giving rise to an edit war.

\subsection{Synthetic Experiments}\label{synExp}
We aim to evaluate the qualitative guarantees and the scalability. Algorithms are implemented in Python 3.6  on Ubuntu 18.10. Intel Core i9 3.6 GHz CPU, 16GB of memory.

\subsubsection{Implemented Algorithms}
We describe different algorithms that are implemented.
\paragraph{Batch  Deployment  Algorithms}
{\tt Brute Force:}
An exhaustive algorithm which compares all possible combinations of deployment requests and returns the one that optimizes the objective function.\\
{\tt BaselineG:} This algorithm sorts the deployment requests in decreasing order of $\frac{f_i}{\vec{w_i}}$ and greedily selects requests until worker availability $W$ is exhausted.\\
{\BatchStrat:} Our proposed solution described in Section~\ref{sec:batch}.

\paragraph{\ADPaR Algorithms}
{\tt ADPaRB:} This is a brute force algorithm that examines all sets of strategies of size $k$. It returns the one that has the smallest distance to the task designer's original deployment parameters. While it returns the exact answer, this algorithm takes exponential time to run.\\
{\tt Baseline2:}  This baseline algorithm is inspired by a related work~\cite{mishra}. The main difference though, the related work modifies the original deployment request by just one parameter at a time and is not optimization driven. In contrast, \ADPaRA returns an alternative deployment request, where multiple parameters may have to be modified. \\
{\tt Baseline3:} This one is designed by modifying space partitioning data structure R-Tree~\cite{beckmann1990r}. We treat each strategy parameters as a point in a $3$-D space and index them using an R-Tree. Then, it scans the tree to find if there is a minimum bounding box (MBB) that exactly contains $k$ strategies. If so, it returns the top-right corner of that MBB as the alternative deployment parameters and corresponding $k$ strategies. If such an MBB does not exist, it will return the top right corner of another MBB that has at least $k$ strategies and will randomly return $k$ strategies from there. \\
{\ADPaRA:} Our proposed solution in Section~\ref{sec:adpar}. \\

{\bf Summary of Results:}
Our simulation experiments highlight the following findings:
{\bf Observation 1:} Our solution \BatchStrat returns exact answers for throughput optimization, and the approximation factor for pay-off maximization is always above $90\%$, significantly surpassing its theoretical approximation factor of $1/2$.
{\bf Observation 2:}  Our solution \BatchStrat is highly scalable and takes less than a second to handle millions of strategies, and hundreds of deployment requests, and $k$.
{\bf Observation 3:} Our algorithm \ADPaRA returns exact solutions to the \ADPaR problem, and significantly outperforms the two baseline solutions in objective function value.
{\bf Observation 4:} \ADPaRA is scalable and takes a few seconds to return alternative deployment parameters, even when the total number of strategies is large and $k$ is sizable.
\smallskip
\begin{figure*}[htpb]
	\subfloat[Varying $k$]{
		\includegraphics[height=6 pc, width=9 pc]{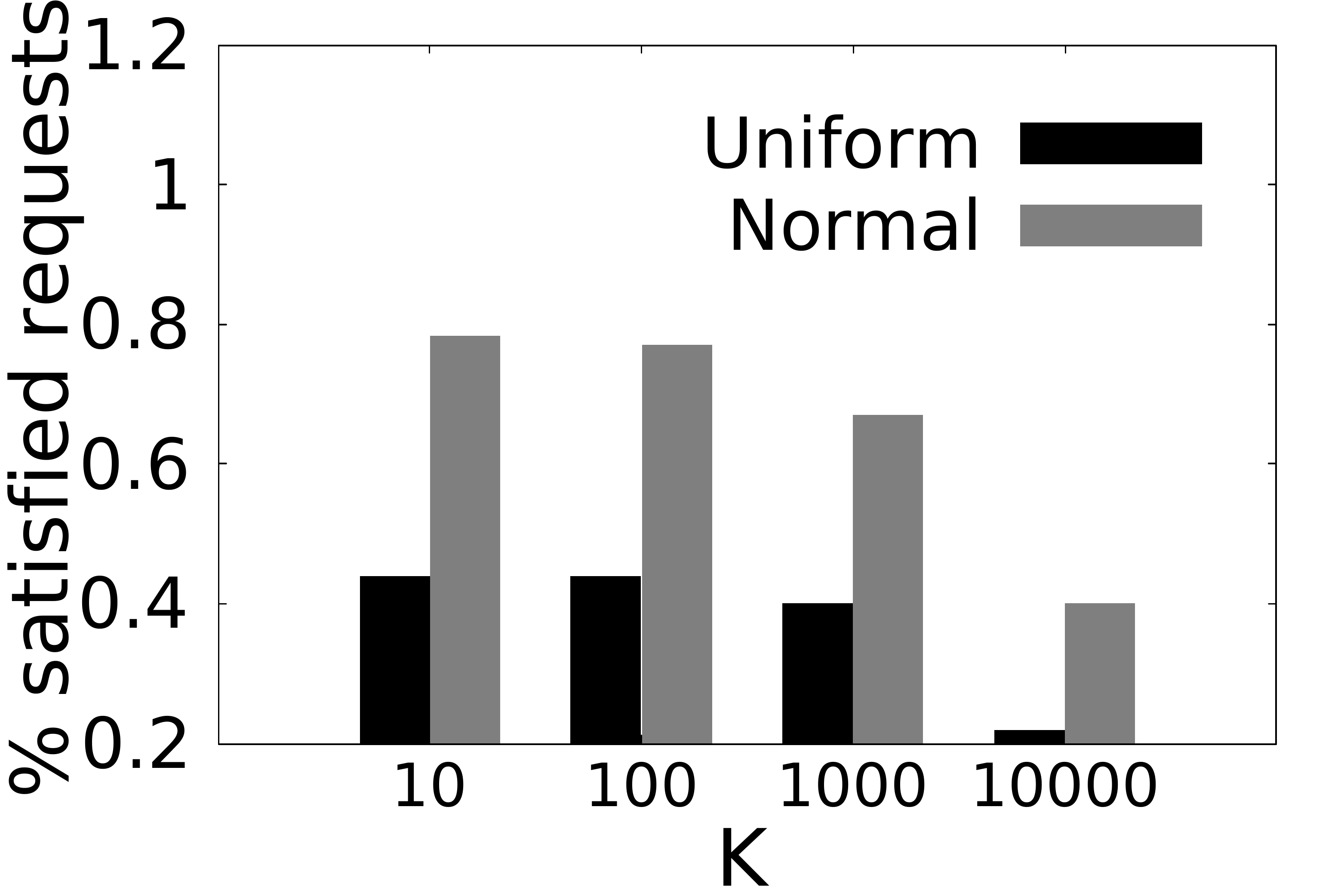}
	}
	\subfloat[Varying $m$]{
		\includegraphics[height=6 pc, width=9 pc]{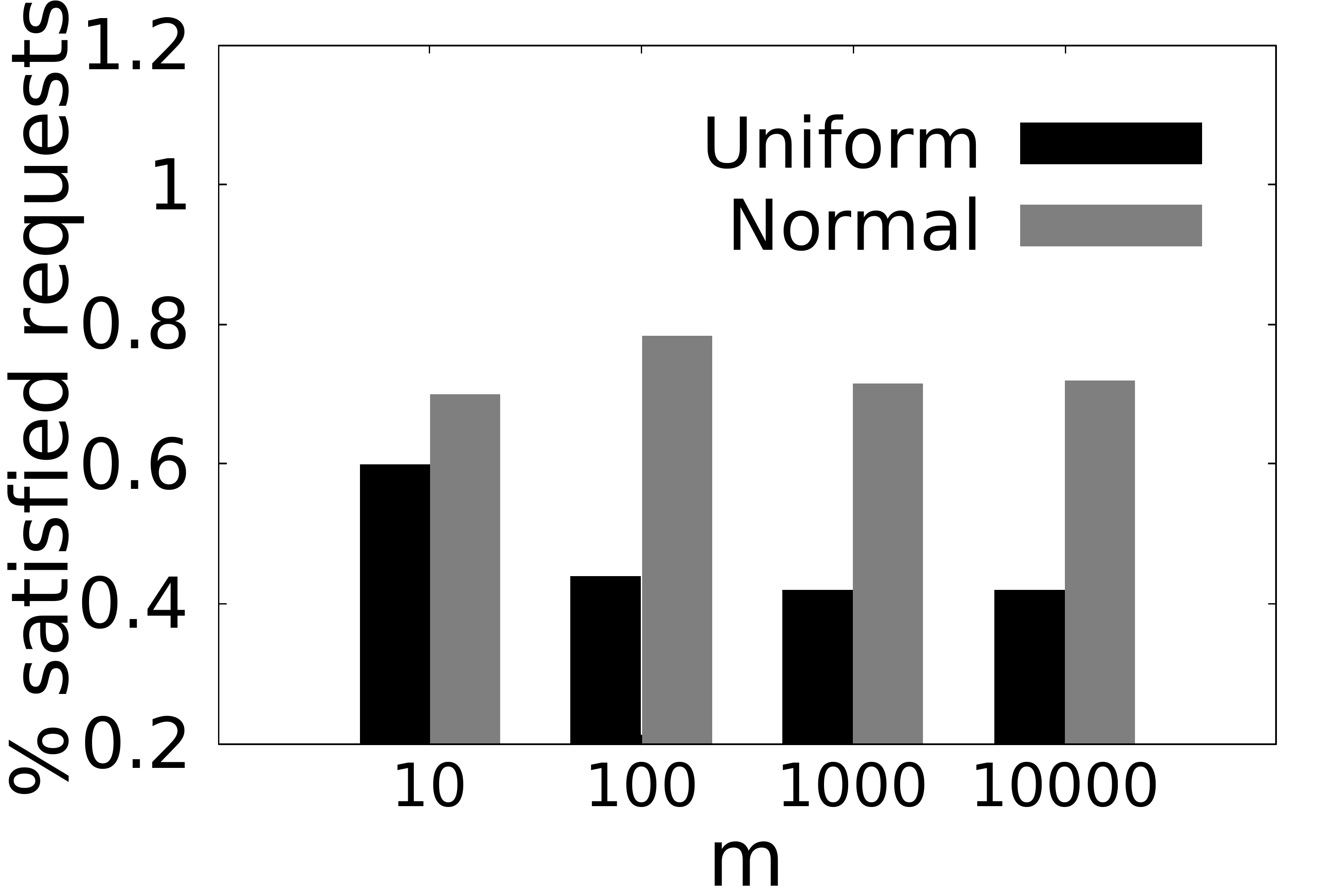}
	}
	\subfloat[Varying $\mathcal{S}$]{
		\includegraphics[height=6 pc, width=9 pc]{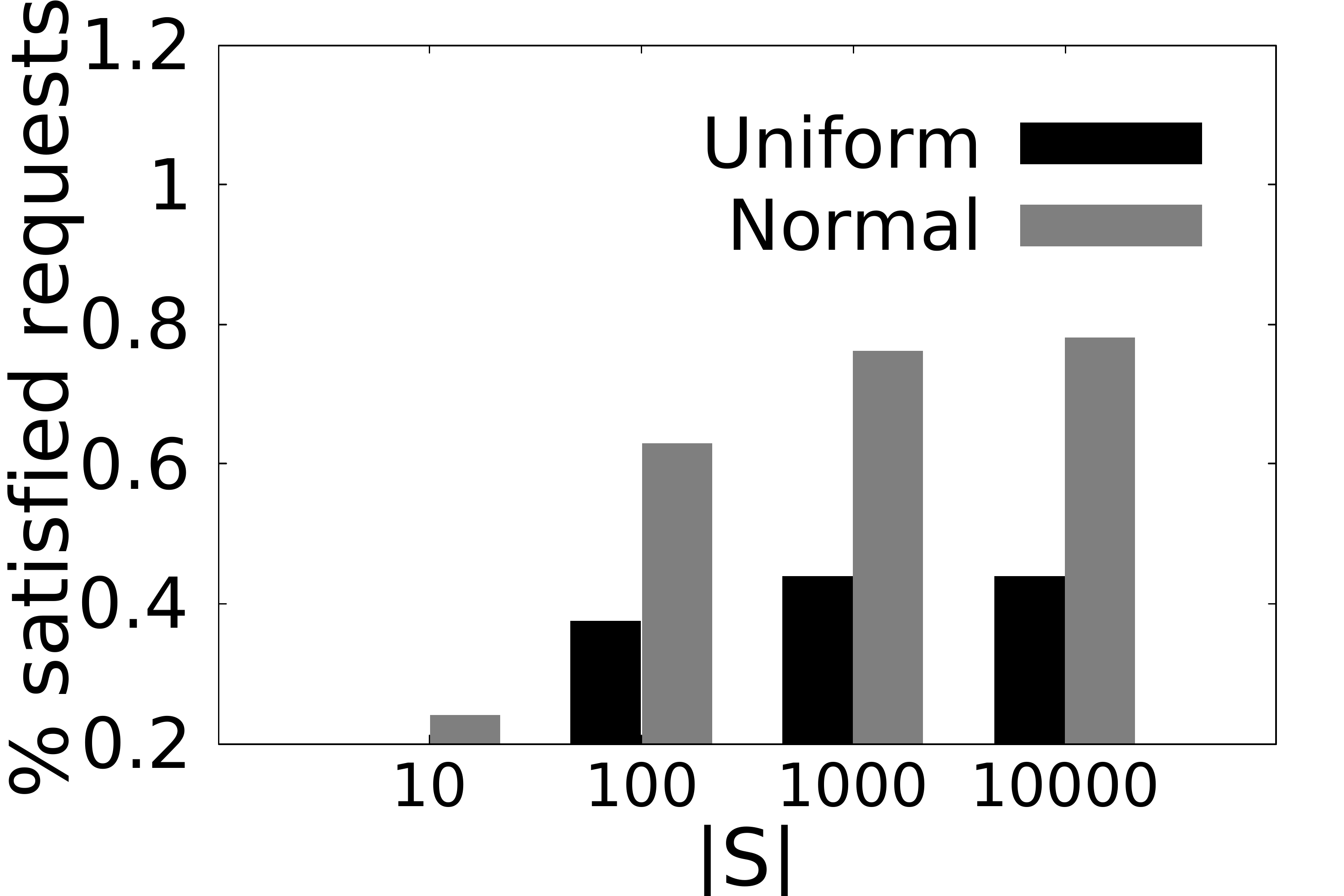}
	}
	\subfloat[Varying $W$]{
		\includegraphics[height=6 pc, width=9 pc]{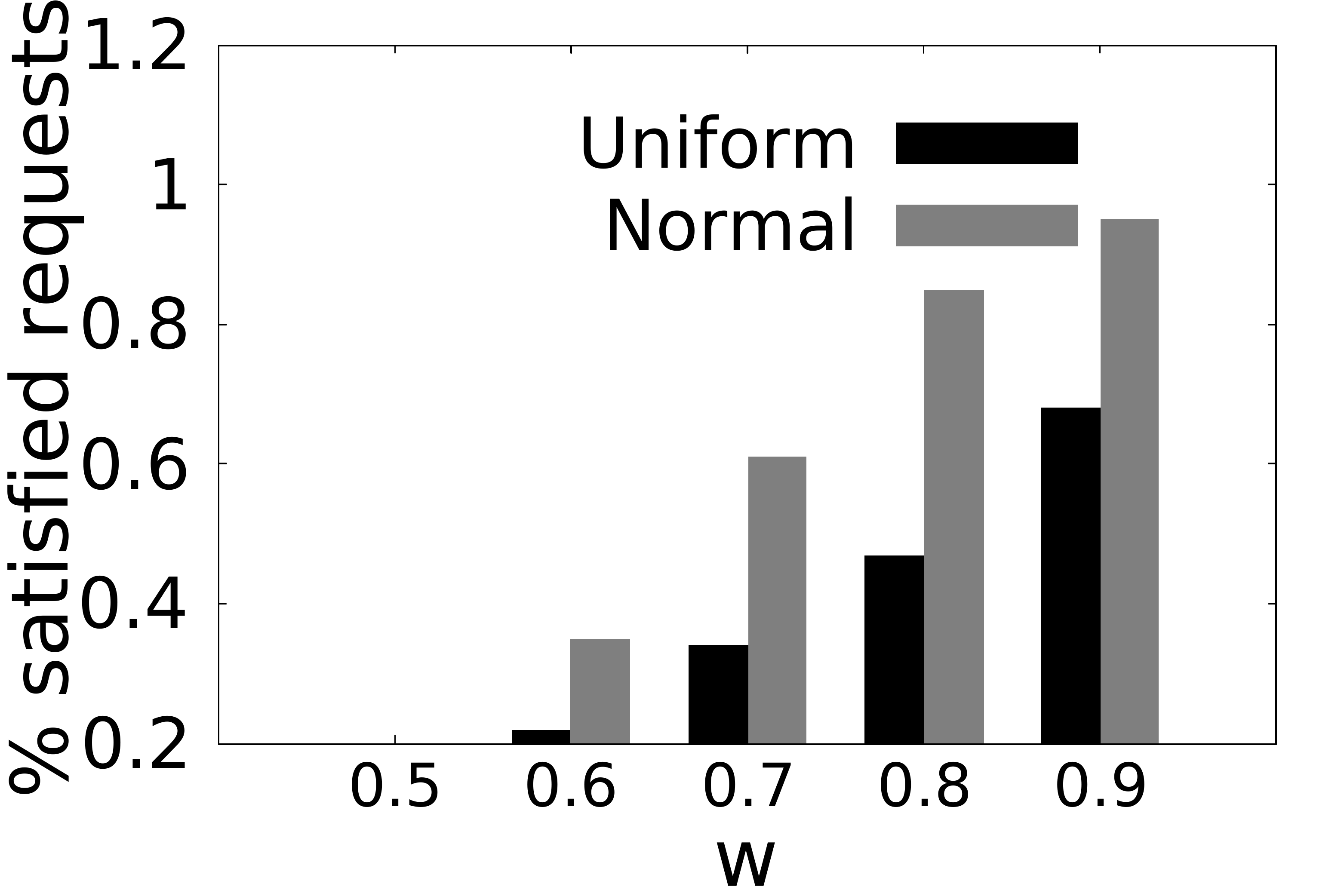}
	}
	\caption{Percentage of satisfied requests before invoking \ADPaR}
	\label{CC2}
\end{figure*}
\subsubsection{Quality Experiment}
\indent\\
\indent{\bf Batch Deployment Recommendation.}
{\bf Goal:} We validate the following two aspects: (i) {\em how many deployment requests} \BatchStrat { }{\em can satisfy without invoking} \ADPaR? (ii) {\em How does} \BatchStrat{ }{\em fare to optimize different platform-centric goals?} We compare \BatchStrat with the other two baselines, as appropriate.

{\bf Strategy Generation:}  The dimension values of a strategy  are generated considering uniform and normal distributions. For the normal distribution, the mean and standard deviation are set to $0.75$ and $0.1$, respectively. We randomly pick the value from $0.5$ to $1$ for the uniform distribution.

{\bf Worker Availability:} For a strategy, we generate $\alpha$ uniformly from an interval $[0.5,1]$. Then, we set $\beta = 1- \alpha$ to make sure that the estimated worker availability $W$ is within $[0,1]$. These numbers are generated in consistence with our real data experiments.

{\bf Deployment Parameters:} Once $W$ is estimated, the quality, latency, and cost - i.e., the deployment parameters, are generated in the interval $[0.625,1]$. For each experiment, $10$ deployment parameters are generated, and an average of $10$ runs is presented in the results.
\begin{figure*}[htpb!]
\centering
	\subfloat[Varying $k$]{
		\includegraphics[height=6 pc, width=9 pc]{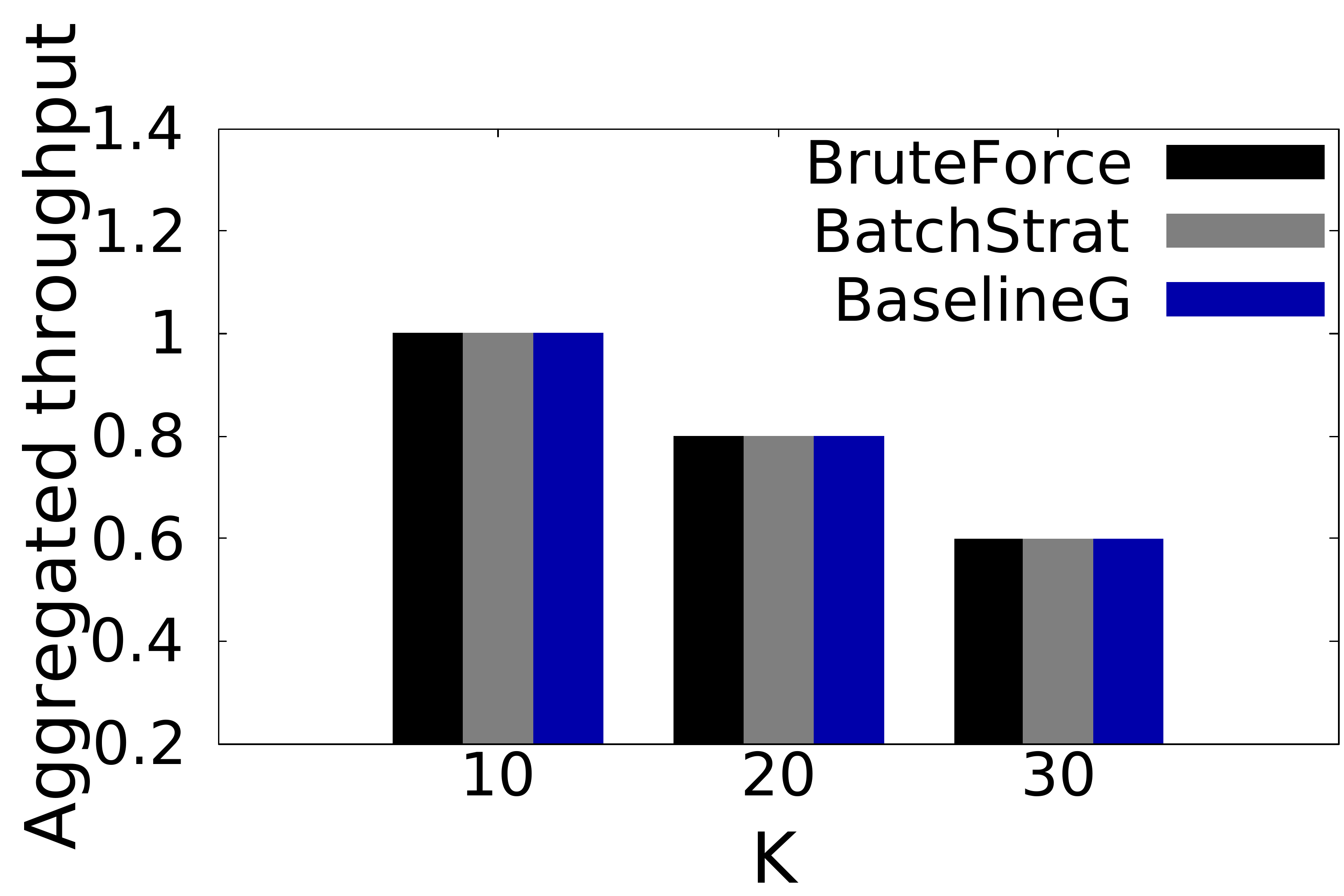}
	}
	\subfloat[Varying $m$]{
		\includegraphics[height=6 pc, width=9 pc]{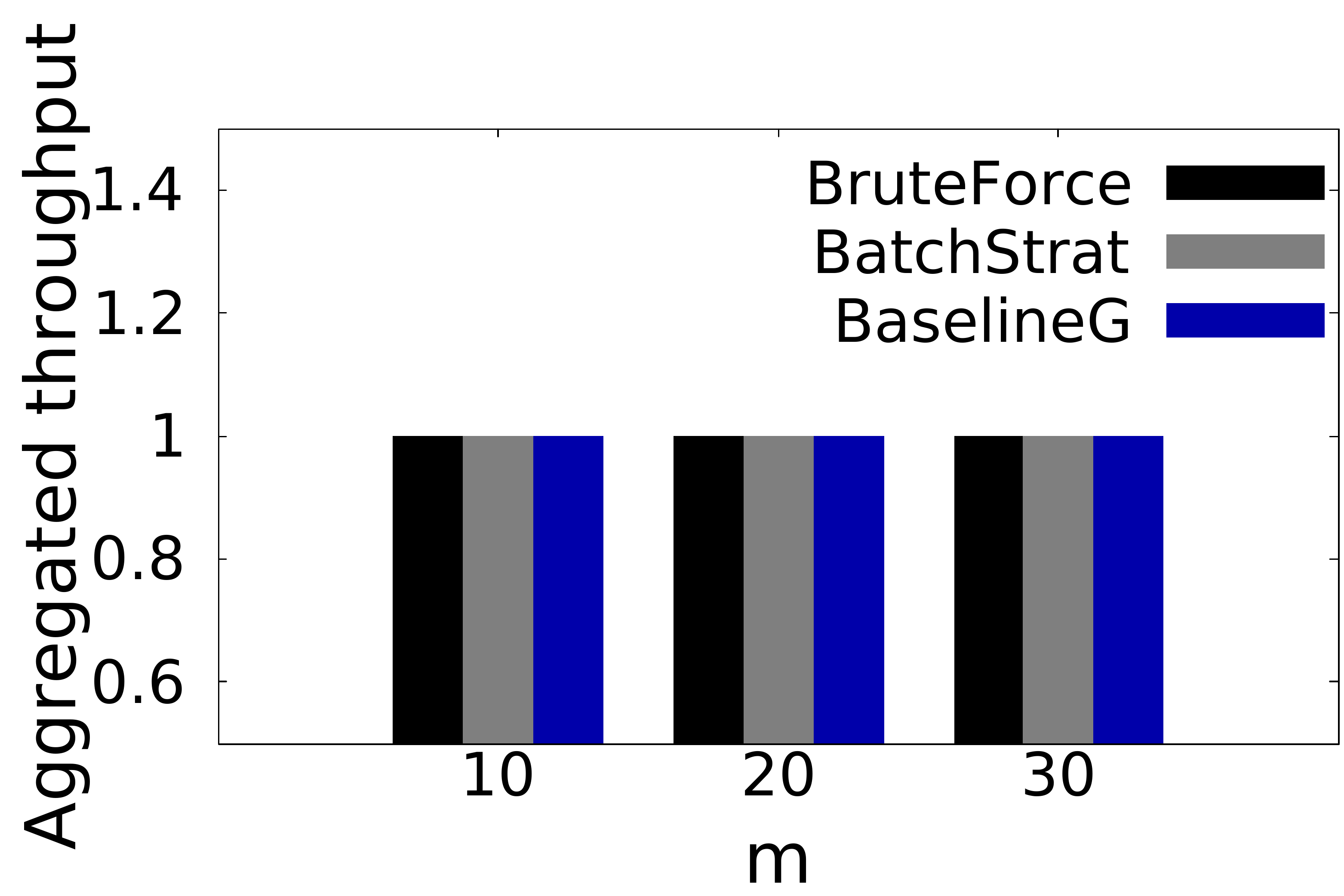}
	}
	\subfloat[Varying $\mathcal{S}$]{
		\includegraphics[height=6 pc, width=9 pc]{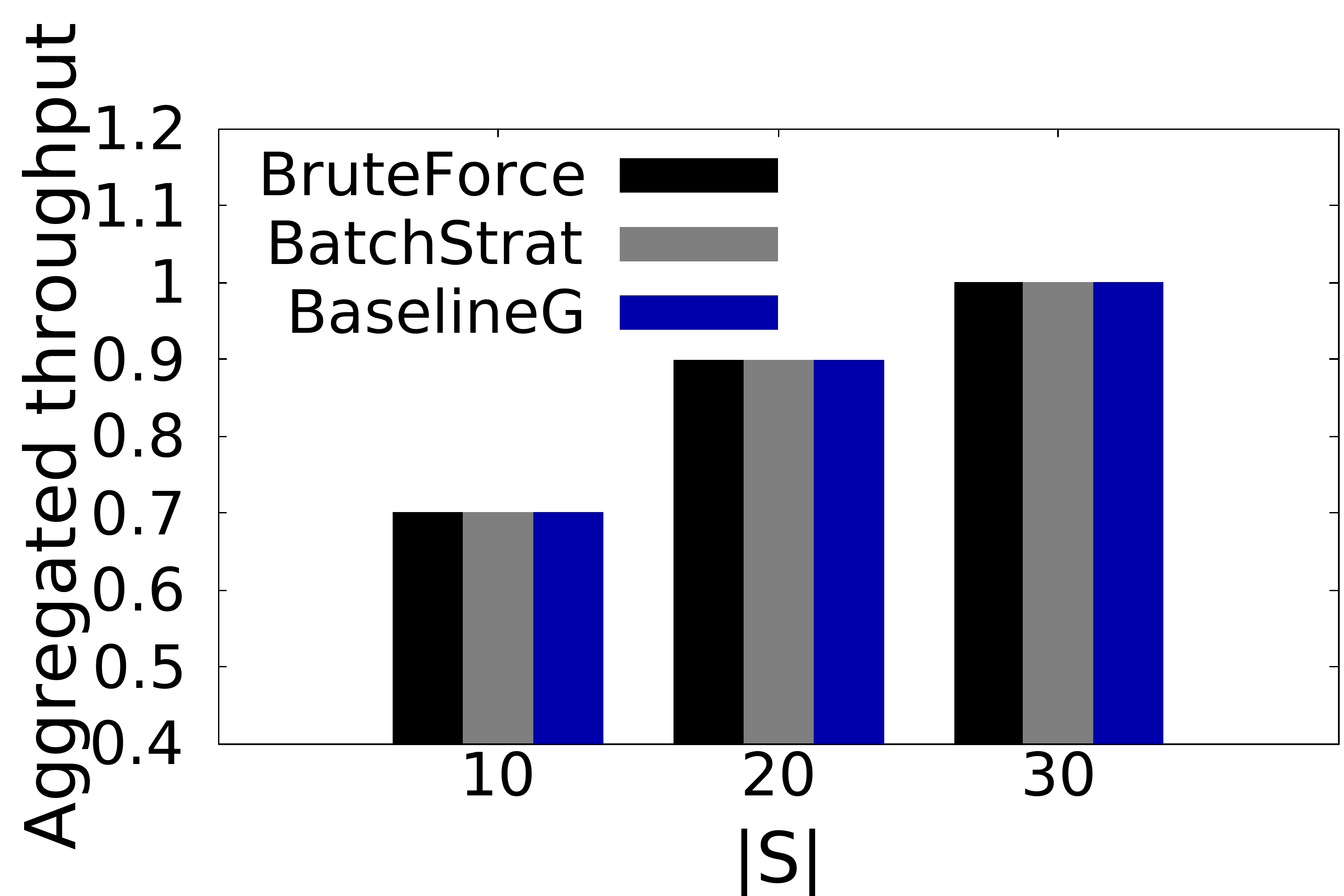}
	}
	\caption{Objective Function for Throughput}
	\label{throughout}
\end{figure*}

\begin{figure*}[htpb]
\centering
	\subfloat[Varying $k$]{
		\includegraphics[height=6 pc, width=9 pc]{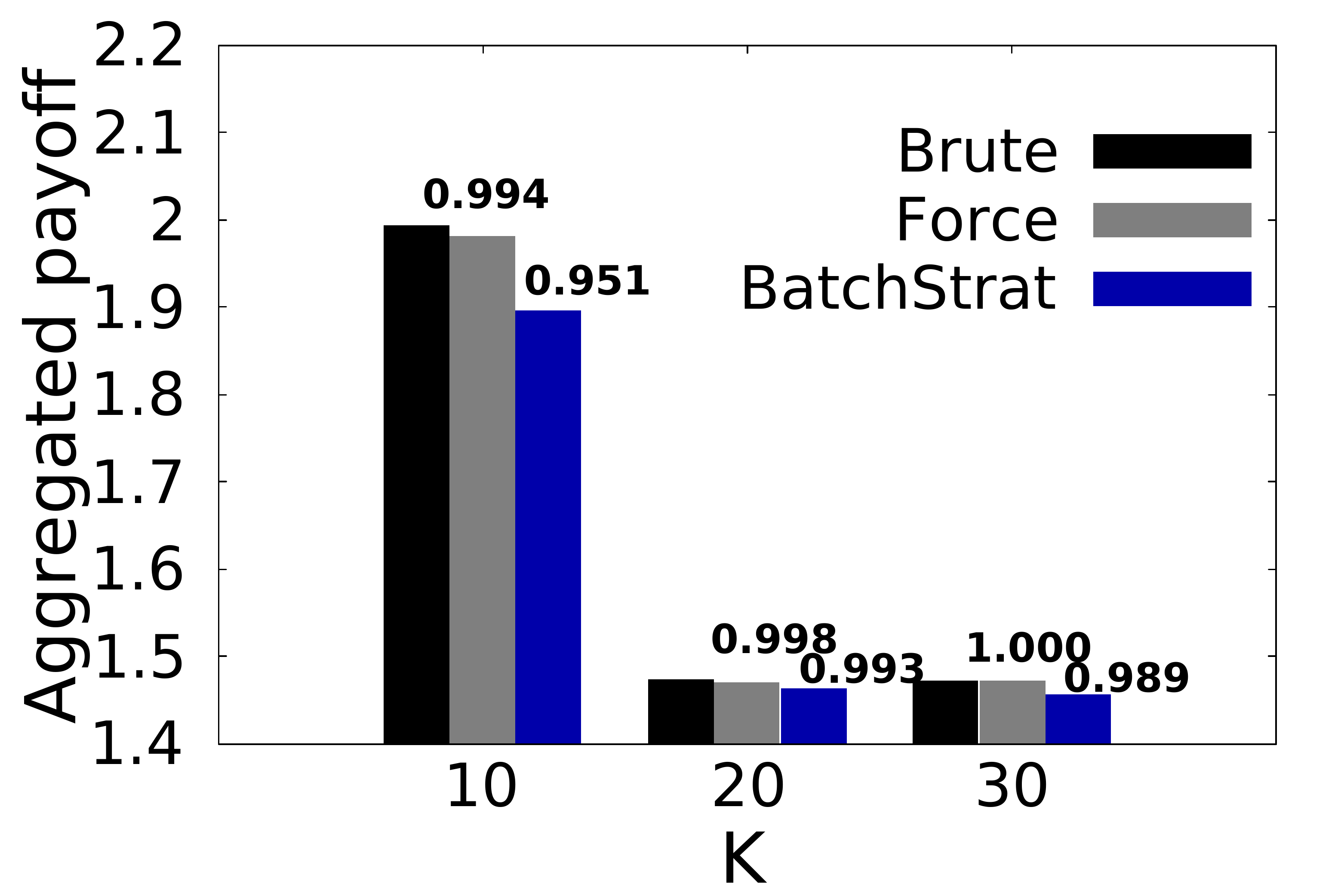}
	}
	\subfloat[Varying $m$]{
		\includegraphics[height=6 pc, width=9 pc]{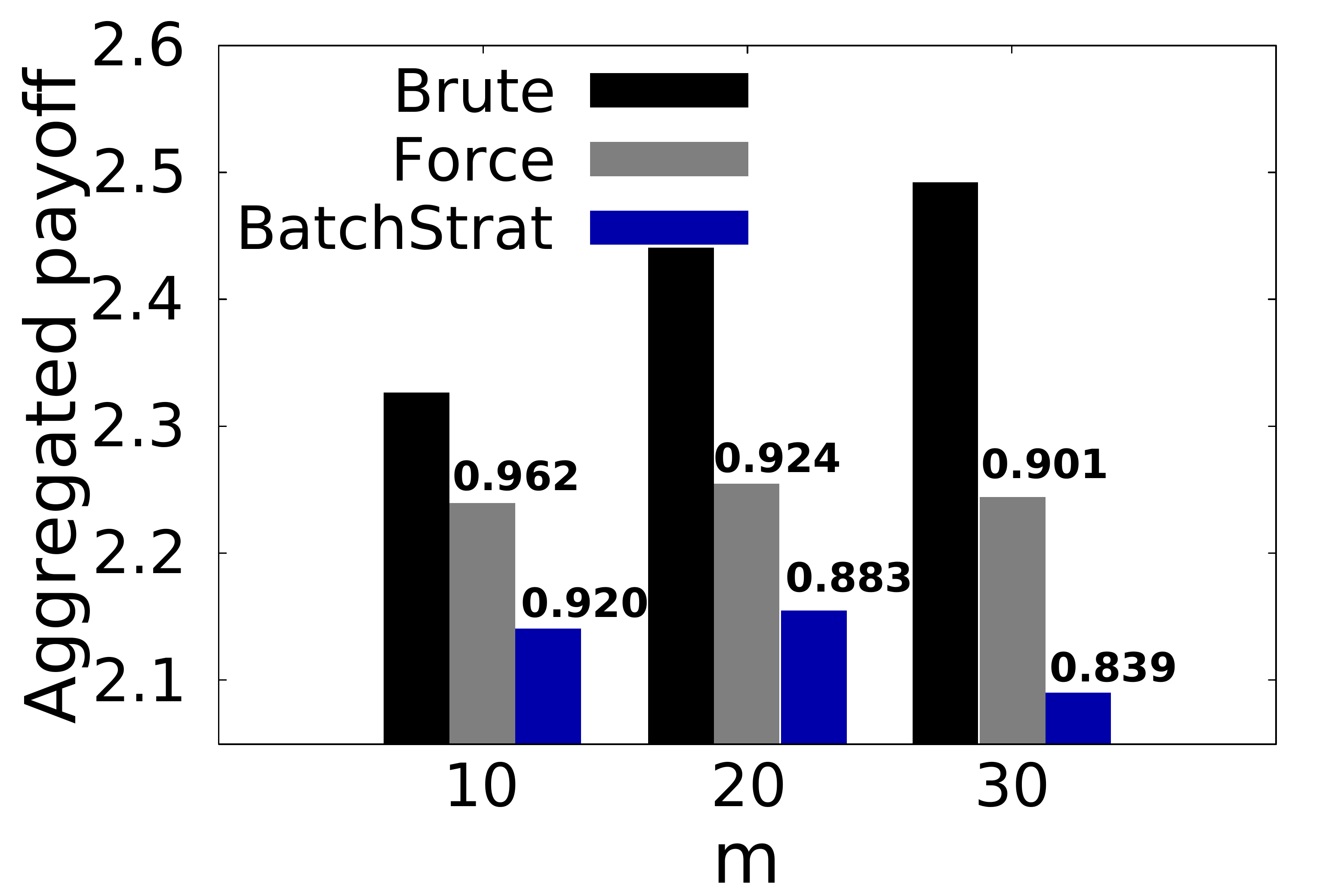}
	}
	\subfloat[Varying $\mathcal{S}$]{
		\includegraphics[height=6 pc, width=9 pc]{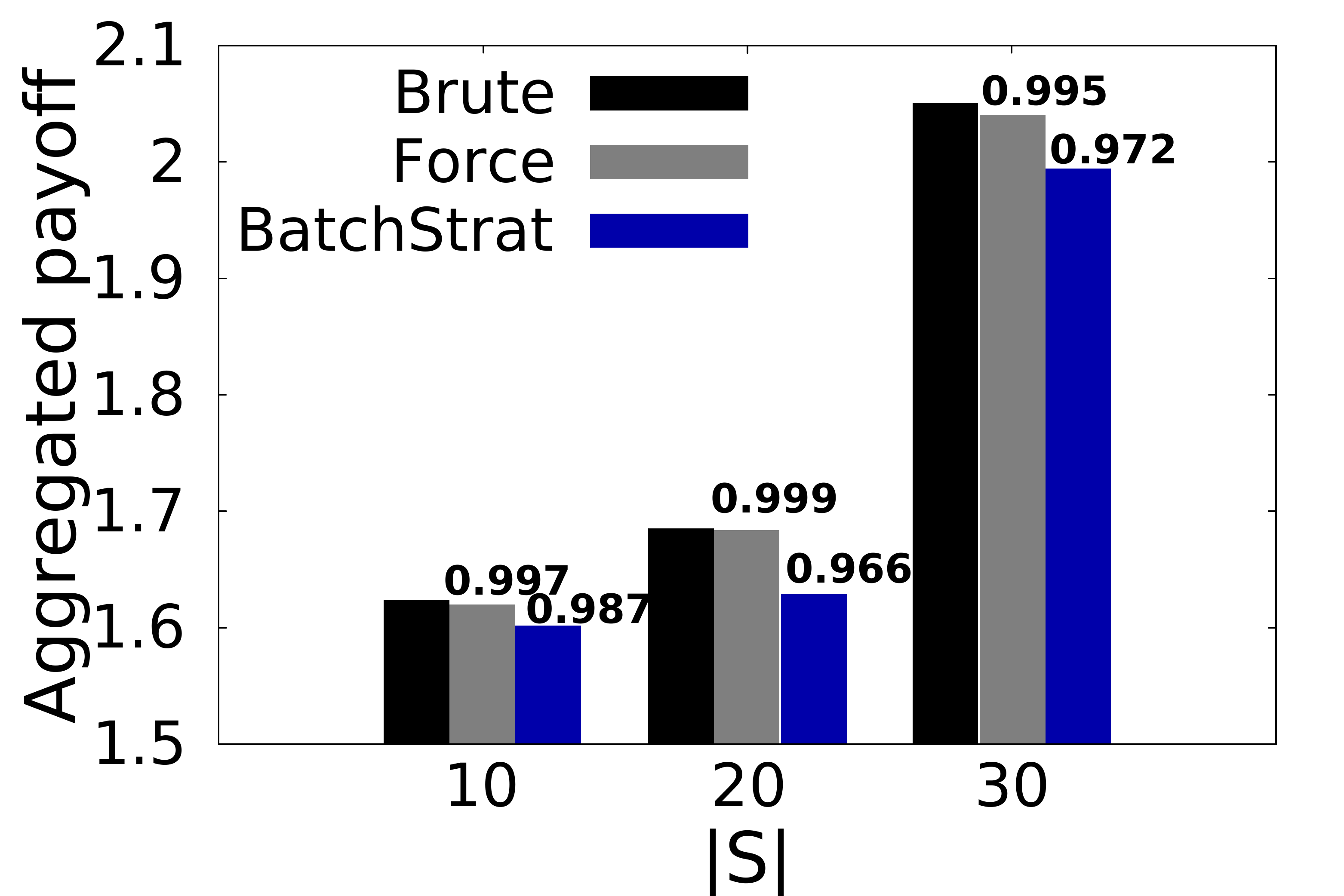}
	}
	\caption{Objective Function and Approximation Factor for Payoff}
	\label{payoff}
\end{figure*}

\begin{figure*}[htpb]

	\subfloat[without Brute Force]{
		\includegraphics[height=6 pc, width=9 pc]
		{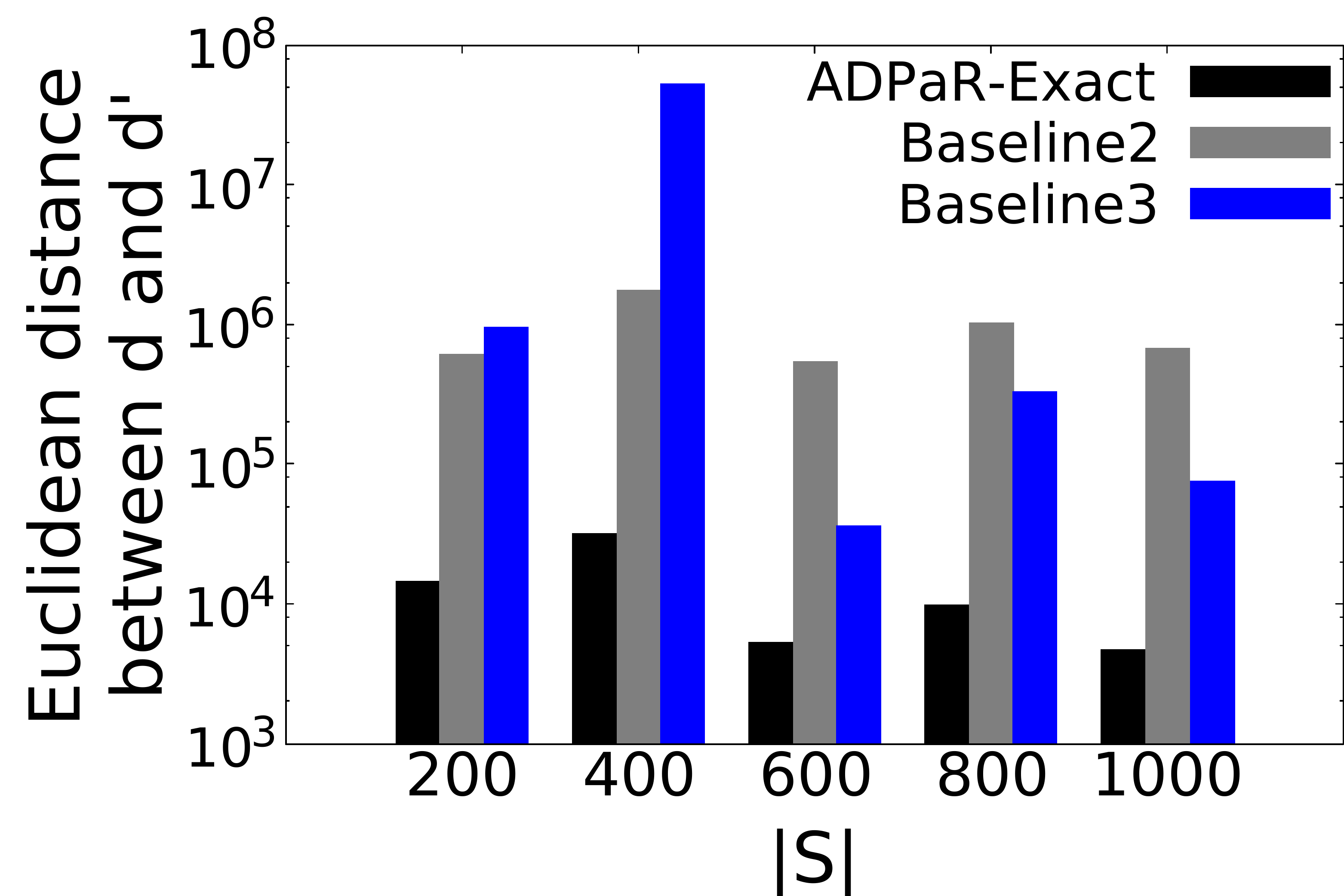}}
	\subfloat[with Brute Force]{
		\includegraphics[height=6 pc, width=9 pc]{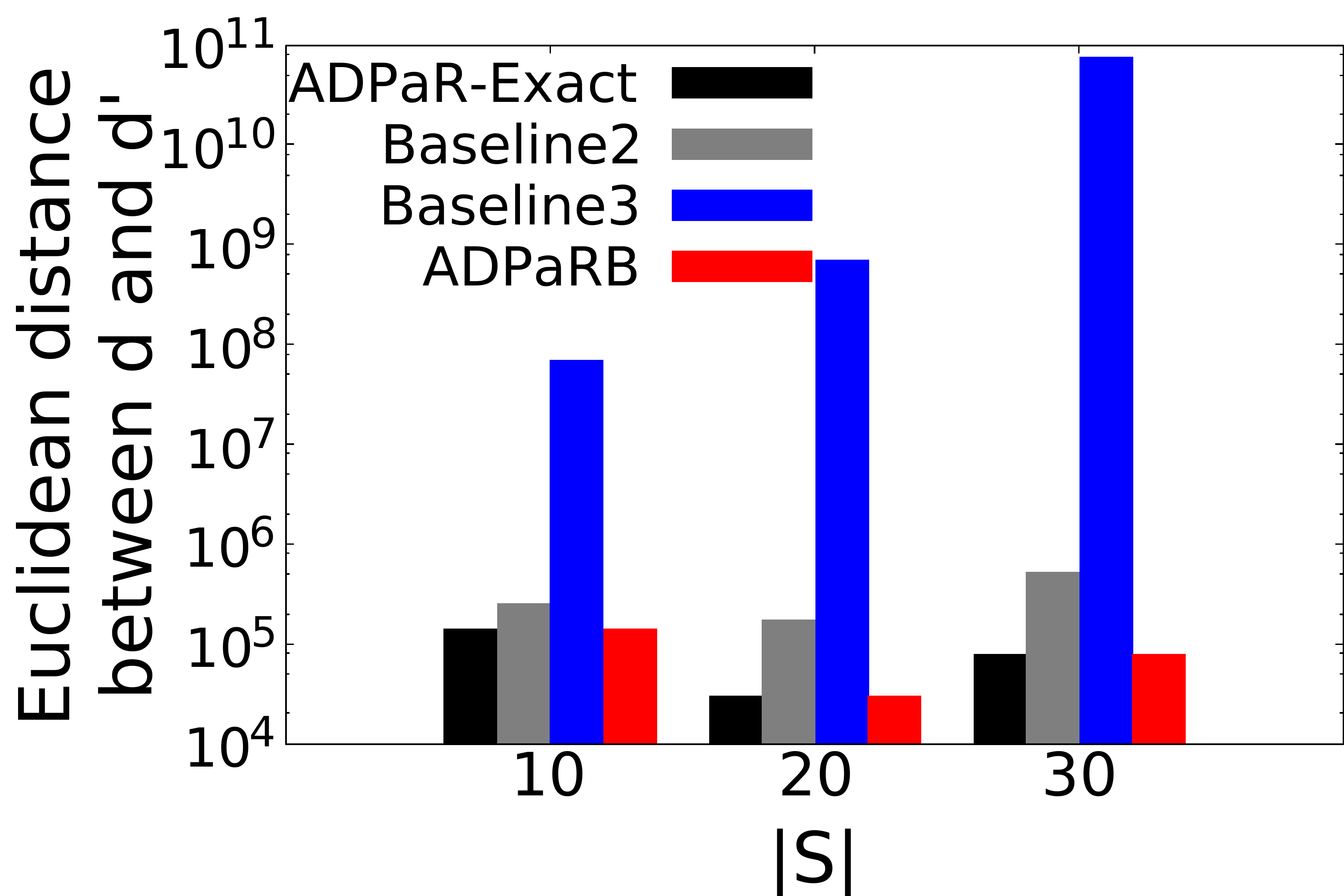}
	}
	\subfloat[without Brute Force]{
		\includegraphics[height=6 pc, width=9 pc]{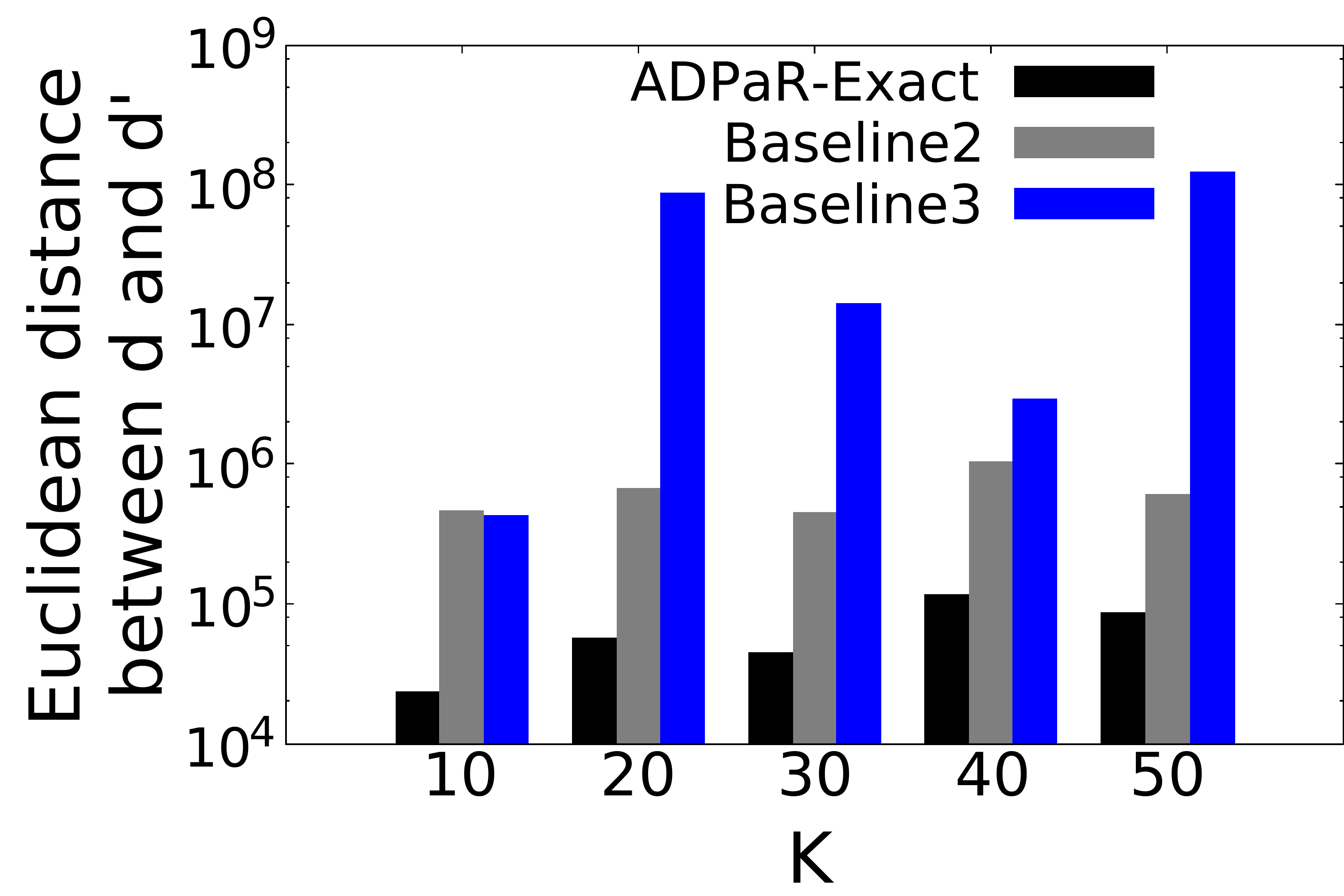}
	}
	\subfloat[with Brute Force]{
		\includegraphics[height=6 pc, width=9 pc]{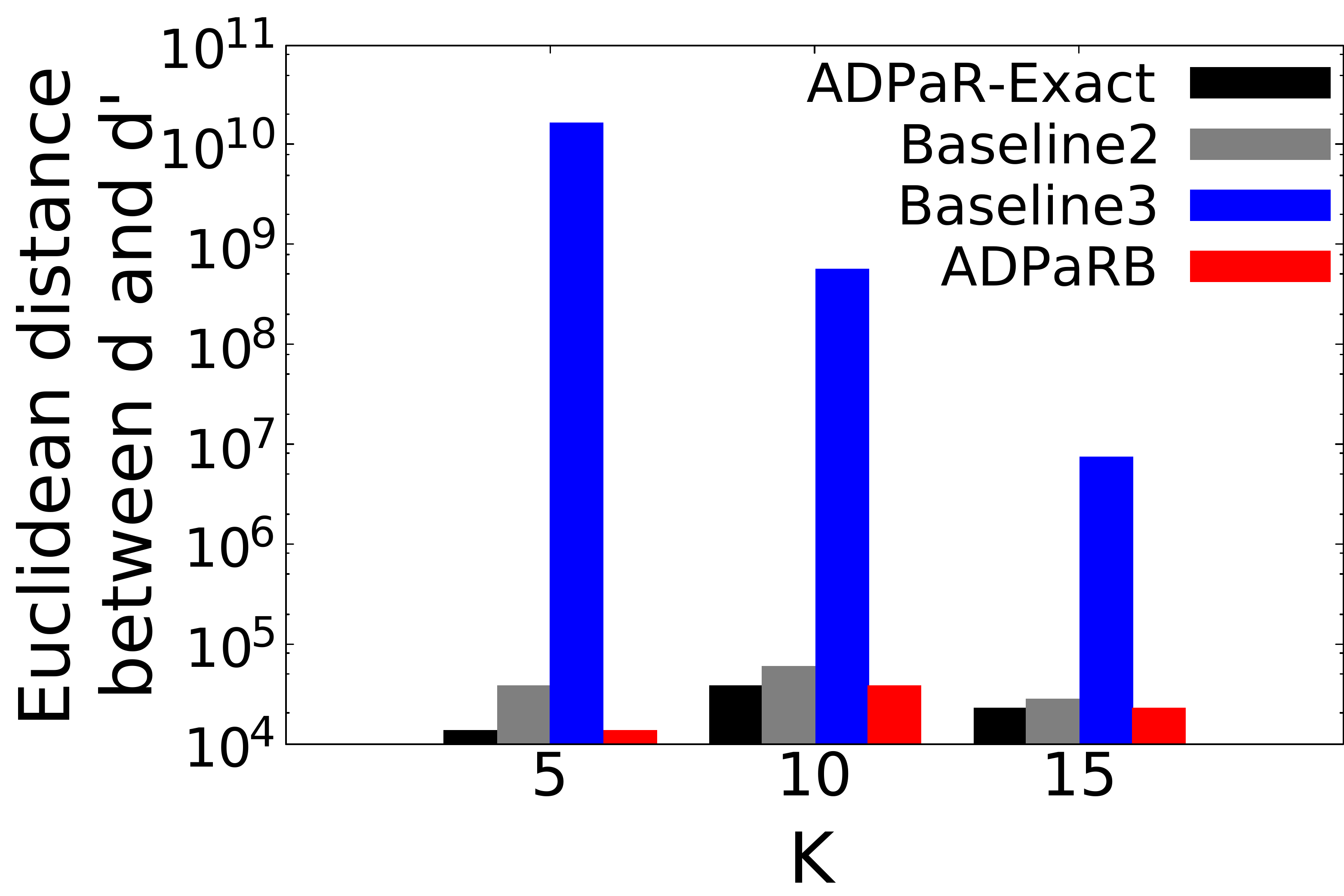}
	}
	\caption{Quality Experiments for \ADPaR}
	\label{Quality}
\end{figure*}

\begin{figure*}[htpb]
\centering
	\subfloat[Batch Deployment Varying $m$]
	{
		\includegraphics[height=6 pc, width=10 pc]{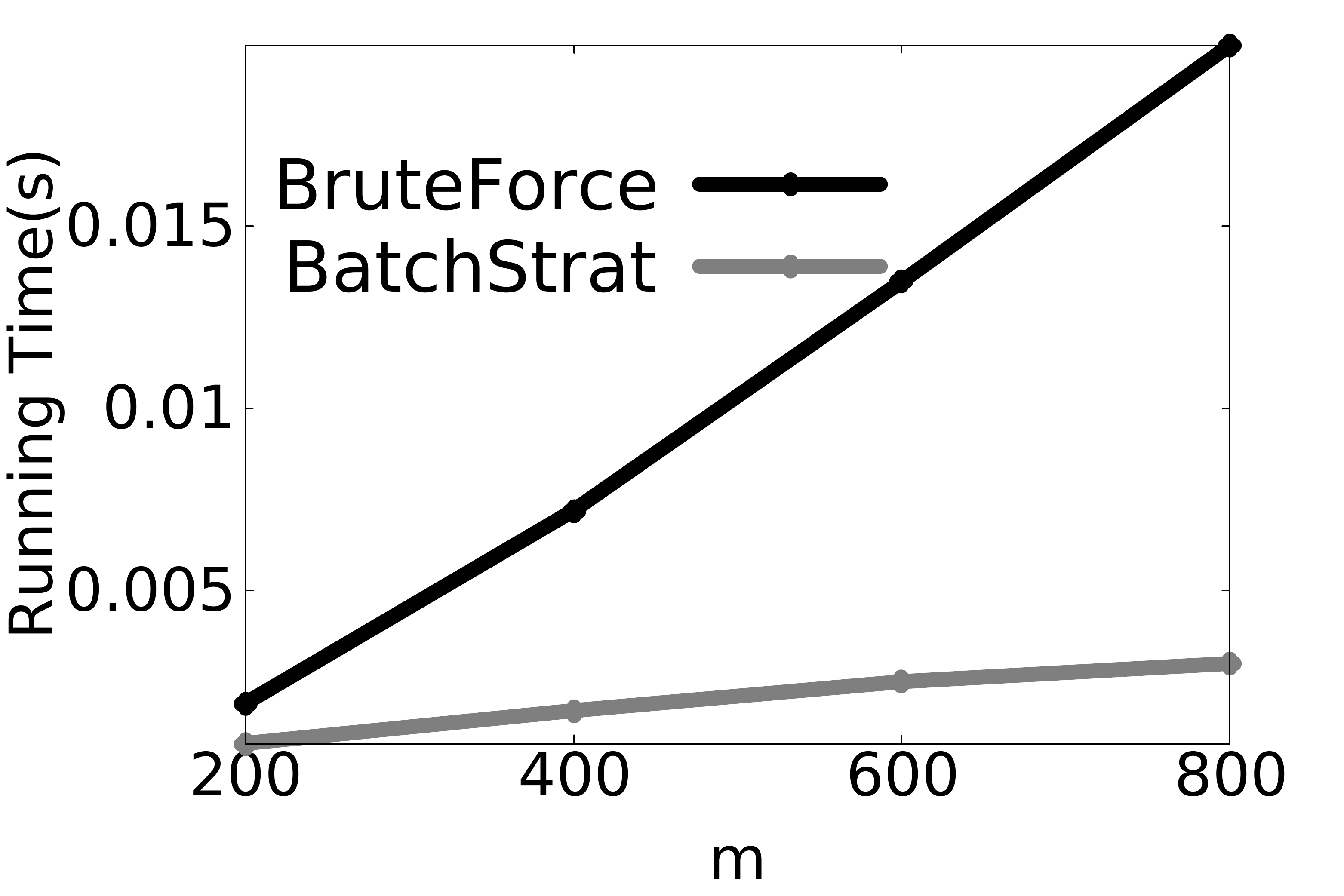}\label{batch}
	}
	\subfloat [\ADPaR Varying $|\mathcal{S}|$]
	{
		\includegraphics[height=6 pc, width=10 pc]{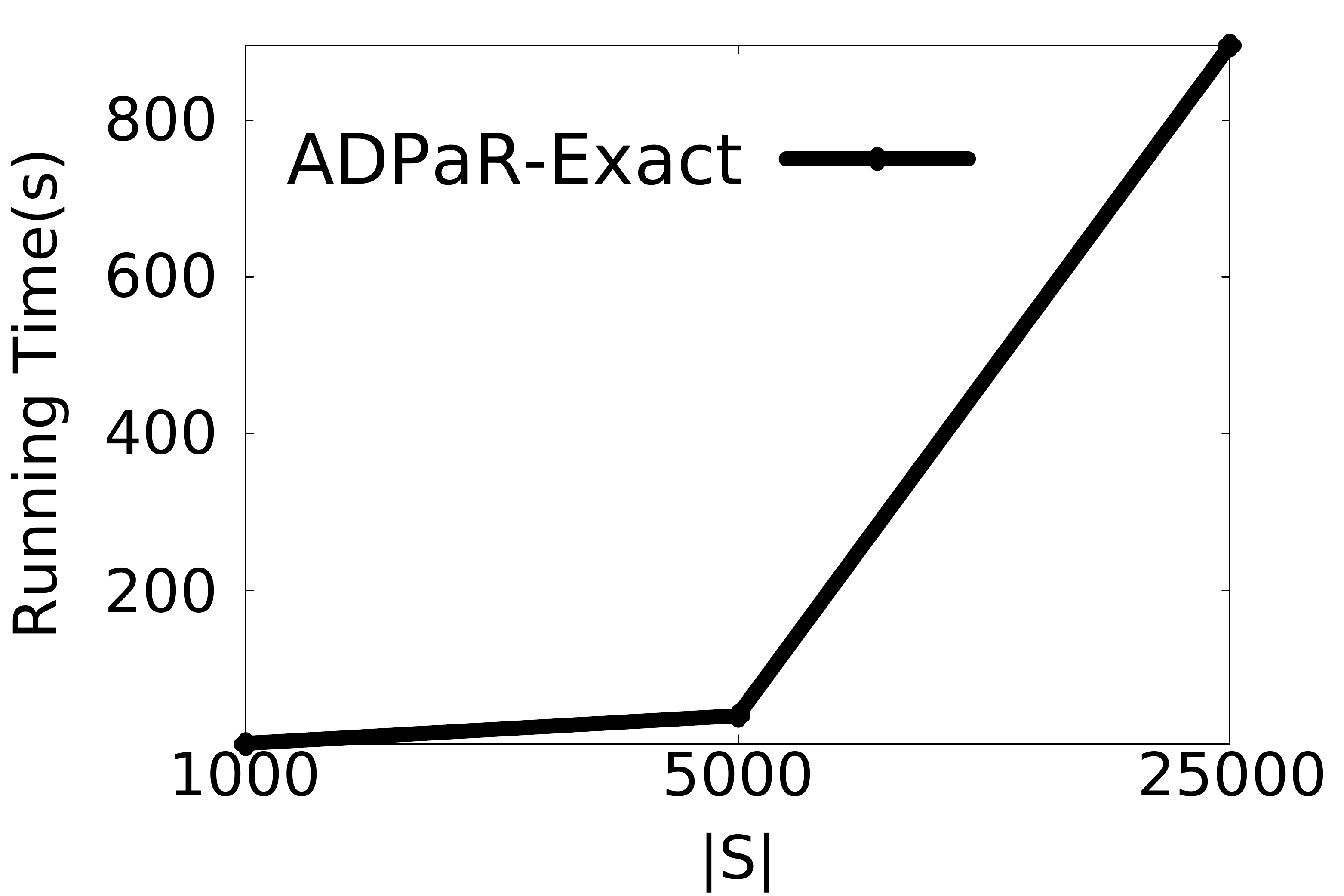}\label{sc1}
	}
	\subfloat[\ADPaR Varying $k$]
	{
		\includegraphics[height=6 pc, width=10 pc]{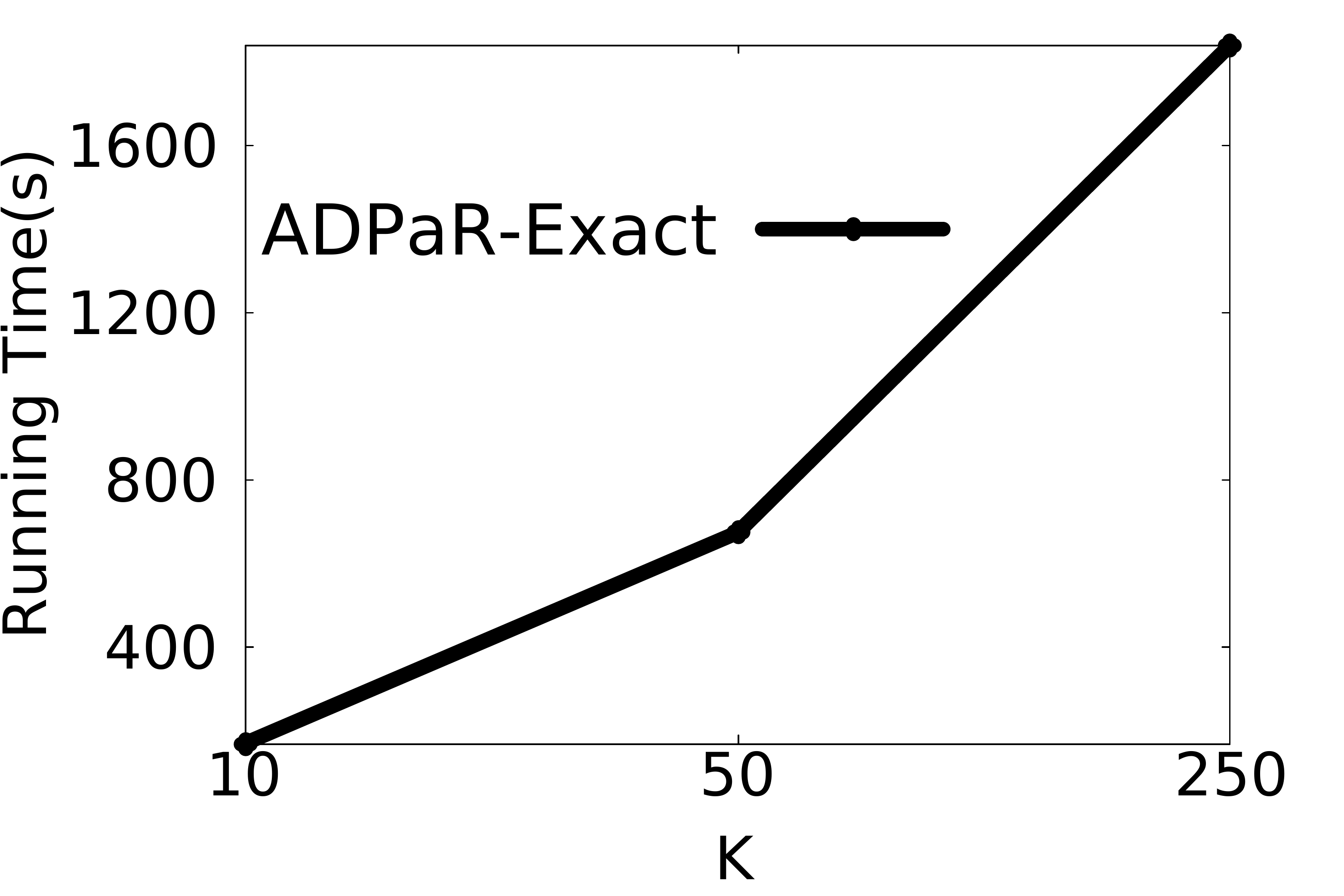}\label{sc2}
	}
	\caption{\small Scalability Experiments}
	\label{Scalability_Q1}
\end{figure*}

Figure~\ref{CC2} shows the percentage of satisfied requests by \BatchStrat with varying $k$, $m$, $|\mathcal{S}|$, $W$. In general, normal distribution performs better than uniform. Upon further analysis, we realize that normal distribution has a very small standard deviation, and is thereby able to satisfy more requests. As shown in Figure~\ref{CC2}(a), the percentage of satisfied requests decreases with increasing $k$, which is expected. Contrarily, the effect of increasing batch size $m$ is less pronounced. This is because all requests use the same underlying distribution, allowing \BatchStrat to handle more of them. With more strategies $|\mathcal{S}|$, as Figure \ref{CC2}(c) illustrates, \BatchStrat satisfies more requests, which is natural, because with increasing $|\mathcal{S}|$, it simply has more choices. Finally, in Figure \ref{CC2}(d), with higher worker availability \BatchStrat  satisfies more requests.  By default, we set $|\mathcal{S}| = 10000, m = 10, k=10, W=0.5$.

Figure \ref{throughout} shows the results of throughput of \BatchStrat by varying $k$, $m$,$|\mathcal{S}|$, compare with the two baselines. Figure \ref{payoff} shows the approximation factor of \BatchStrat and {\tt BaselineG}. \BatchStrat achieves an approximation factor of $0.9$ most of the time. For both experiments, the default values are $k =10, m = 5, |\mathcal{S}| = 30, W = 0.5$ because brute force does not scale beyond that.

{\bf Alternative Deployment Recommendation \ADPaR.}
The goal here is to measure the objective function. Since {\tt ADPaRB} takes exponential time,  to be able to compare with this, we set
$|\mathcal{S}| = 20$, $k = 5, W=0.5$ for all the quality experiments that has to compare with the brute force. Otherwise, the default values are $|\mathcal{S}| = 200$, $k = 5$.

In Figure \ref{Quality}, we vary $|\mathcal{S}|$ and  $k$ and plot the Euclidean distance between $d$ and $d'$ (smaller is better). Indeed, \ADPaRA returns exact solution always.  The other two baselines perform significantly worse, while {\tt Baseline 3} is the worst. That is indeed expected, because these two baselines are not optimization guided, and does not satisfy our goal.  Naturally, the objective function decreases with increasing  $|\mathcal{S}|$, because more strategies mean smaller change in $d'$, making the distance between $d$ and $d'$  smaller. As the results depict, optimal Euclidean distance between $d$ and $d'$ increases with increasing $k$, which is also intuitive, because, with higher $k$ value, the alternative deployment parameters are likely to have more distance from the original ones.

\subsubsection{Scalability Experiments} \label{scaExp}
Our goal is to evaluate the running time of our proposed solutions. Running time is measured in seconds. We present a subset of results that are representative.
\paragraph{Batch Deployment Recommendation}
Since the {\tt BaselineG} has the same running time as that of \BatchStrat (although qualitatively inferior), we only compare the running time between {\tt Brute Force} and \BatchStrat. The default setting for $|\mathcal{S}|$, $k$ and $W$ are $30$, $10$ and $0.75$, respectively.

The first observation we make is, clearly \BatchStrat can handle millions of strategies, several hundreds of batches, and very large $k$ and still takes only a few fractions of seconds to run. It is easy to notice that the running time of this problem only relies on the size of the batch $m$ (or the number of deployment requests), and not on $k$ or $\mathcal{S}$.  As we can see in Figure \ref{batch}, {\tt Brute Force} takes exponential time with increasing $m$, whereas \BatchStrat scales linearly.

\paragraph{Alternative Deployment Recommendation}
We vary $k$ and $|\mathcal{S}|$ with defaults set to $5$ and $10000$ respectively, and evaluate the running time of \ADPaRA. $W$ is set to $0.5$.
As Figures~\ref{sc1} and~\ref{sc2} attest, albeit non-linear, \ADPaRA scales well with $k$ and $|\mathcal{S}|$. We do not present the baselines as they are significantly inferior in quality.

\section{Related Work}\label{related}
{\bf Crowdsourcing Deployment:} Till date, the burden is entirely on the task requester to design appropriate deployment strategies that are consistent with the cost, latency, and quality parameters of task deployment.
A very few related works~\cite{zheng2011task,allen2018design} have started to study the importance of appropriate deployment strategies but these works do not propose an algorithmic solution and are limited to empirical studies. A recent work~\cite{dep1} presents the results of a 10-month deployment of  a  crowd-powered  system  that  uses  a  hybrid  approach  to fast recruitment of workers, called {\em Ignition}. These  results suggest a number of opportunities to deploy work in the online job market.

{\bf Crowdsourcing Applications:}
A number of interactive crowd-powered systems have been developed  to  solve  difficult  problems  and develop applications~\cite{workflow1, turkomatic, workflow3, Bernstein10soylent:a, crowdforge, bio, software, cascade, mozafari2014scaling}. For instance, Soylent uses the crowd to edit and proofread text~\cite{Bernstein10soylent:a}; Chorus recruits a group of workers to hold sophisticated conversations~\cite{chorus}; and Legion allows a crowd to interact with an UI-control task~\cite{legion}. A primary challenge for such interactive systems is to decrease latency without having to compromise with the quality. A comprehensive survey on different crowdsourcing applications could be found at~\cite{survey}. All crowd-powered systems share these challenges and are likely to benefit from \StratRec.

\smallskip \noindent
{\bf Query planning and Refinement:}
The closest analogy of  deployment strategy recommendation is recommending the best query plan in relational databases, in which joins, selections and projections could be combined any number of times. Typical parametric query optimization problems, like \cite{ioannidis1992parametric}, only focus on one objective to optimize. Afterwards, multi-objective problems have been studied, with a focus on optimizing multiple  objectives at the same time \cite{trummer2016multi}. Our work borrows inspiration from that and studies the problem in the deployment context, making the challenges unique and different from traditional query planning.

Query reformulation has been widely studied in Information Retrieval~\cite{query4}. In~\cite{mishra},  authors take users' preference into account and propose an interactive method for seeking an alternative query which satisfies cardinality constraints. This is different from \ADPaR since it only relaxes one dimension at a time. Aris et al.~\cite{query3} proposed a graph modification method to recommend queries that maximize an overall utility.  Mottin et al.~\cite{query2} develop an optimization framework where solutions can only handle Boolean/categorical data.

\smallskip \noindent
{\bf Skyline and Skyband Queries:}
Skyline queries play an essential role in computing favored answers from a database~\cite{sky5,sky1}. Based on the concepts of skylines, other classes of queries arise, especially top-$k$ queries and $k$-skyband problems which aim to bring more useful information than original skylines. Mouratidis et al.~\cite{sky4,sky3} study several related problems. In~\cite{sky4}, sliding windows are used to track the records in dynamic stream rates. In~\cite{sky3}, a geometry arrangement method is proposed for top-$k$ queries with uncertain scoring functions. Because our problem seeks the optimal group of $k$ strategies, it is similar to the top-$k$ queries problem.   However, unlike Skyband or any other related work, \ADPaR recommends alternative deployment parameters. Thus, these solutions do not extend to solve \ADPaR. 
\color{black}

\section{Conclusion}\label{conclusion}
We propose an optimization-driven middle layer to recommend deployment strategies.  Our work addresses multi-faceted modeling challenges through the generic design of modules in \StratRec that could be instantiated to optimize different types of goals by accounting for worker availability.  
We develop computationally-efficient algorithms and validate our work with extensive real data and synthetic experiments. 

This work opens up several important ongoing and future research directions. As an ongoing investigation, we are deploying  additional types of tasks using \StratRec to evaluate its effectiveness. Our future investigation involves adapting batch deployment to optimize additional criteria, such as worker-centric goals, or to combine multiple goals inside the same optimization function.  Understanding the computational challenges of such an interactive system remains to be explored. Finally, how  to design \StratRec for a fully dynamic stream-like setting of incoming deployment requests, where the deployment requests could be revoked, remains to be an important open problem.


\newpage
\bibliographystyle{ACM-Reference-Format}
\bibliography{ref,reference}


\begin{thebibliography}{35}


\ifx \showCODEN    \undefined \def \showCODEN     #1{\unskip}     \fi
\ifx \showDOI      \undefined \def \showDOI       #1{#1}\fi
\ifx \showISBNx    \undefined \def \showISBNx     #1{\unskip}     \fi
\ifx \showISBNxiii \undefined \def \showISBNxiii  #1{\unskip}     \fi
\ifx \showISSN     \undefined \def \showISSN      #1{\unskip}     \fi
\ifx \showLCCN     \undefined \def \showLCCN      #1{\unskip}     \fi
\ifx \shownote     \undefined \def \shownote      #1{#1}          \fi
\ifx \showarticletitle \undefined \def \showarticletitle #1{#1}   \fi
\ifx \showURL      \undefined \def \showURL       {\relax}        \fi
\providecommand\bibfield[2]{#2}
\providecommand\bibinfo[2]{#2}
\providecommand\natexlab[1]{#1}
\providecommand\showeprint[2][]{arXiv:#2}

\bibitem[\protect\citeauthoryear{Allen et~al\mbox{.}}{Allen
  et~al\mbox{.}}{2018}]%
        {allen2018design}
\bibfield{author}{\bibinfo{person}{BJ Allen} {et~al\mbox{.}}}
  \bibinfo{year}{2018}\natexlab{}.
\newblock \showarticletitle{Design Crowdsourcing: The Impact on New Product
  Performance of Sourcing Design Solutions from the Crowd}.
\newblock \bibinfo{journal}{\emph{Journal of Marketing}}
  (\bibinfo{year}{2018}).
\newblock


\bibitem[\protect\citeauthoryear{Anagnostopoulos et~al\mbox{.}}{Anagnostopoulos
  et~al\mbox{.}}{2010}]%
        {query3}
\bibfield{author}{\bibinfo{person}{Aris Anagnostopoulos} {et~al\mbox{.}}}
  \bibinfo{year}{2010}\natexlab{}.
\newblock \showarticletitle{An optimization framework for query
  recommendation}.
\newblock  (\bibinfo{year}{2010}).
\newblock


\bibitem[\protect\citeauthoryear{Beckmann et~al\mbox{.}}{Beckmann
  et~al\mbox{.}}{1990}]%
        {beckmann1990r}
\bibfield{author}{\bibinfo{person}{Norbert Beckmann} {et~al\mbox{.}}}
  \bibinfo{year}{1990}\natexlab{}.
\newblock \showarticletitle{The R*-tree: an efficient and robust access method
  for points and rectangles}. In \bibinfo{booktitle}{\emph{SIGMOD}}. Acm.
\newblock


\bibitem[\protect\citeauthoryear{Bernstein, Little, Miller, Hartmann, Ackerman,
  Karger, Crowell, and Panovich}{Bernstein et~al\mbox{.}}{2010}]%
        {Bernstein10soylent:a}
\bibfield{author}{\bibinfo{person}{Michael~S. Bernstein}, \bibinfo{person}{Greg
  Little}, \bibinfo{person}{Robert~C. Miller}, \bibinfo{person}{Björn
  Hartmann}, \bibinfo{person}{Mark~S. Ackerman}, \bibinfo{person}{David~R.
  Karger}, \bibinfo{person}{David Crowell}, {and} \bibinfo{person}{Katrina
  Panovich}.} \bibinfo{year}{2010}\natexlab{}.
\newblock \showarticletitle{Soylent: A Word Processor with a Crowd Inside}. In
  \bibinfo{booktitle}{\emph{IN PROC UIST'10}}.
\newblock


\bibitem[\protect\citeauthoryear{Borromeo et~al\mbox{.}}{Borromeo
  et~al\mbox{.}}{2017}]%
        {borromeo2017deployment}
\bibfield{author}{\bibinfo{person}{Ria~Mae Borromeo} {et~al\mbox{.}}}
  \bibinfo{year}{2017}\natexlab{}.
\newblock \showarticletitle{Deployment strategies for crowdsourcing text
  creation}.
\newblock \bibinfo{journal}{\emph{Information Systems}} (\bibinfo{year}{2017}).
\newblock


\bibitem[\protect\citeauthoryear{Borzsony et~al\mbox{.}}{Borzsony
  et~al\mbox{.}}{2001}]%
        {sky5}
\bibfield{author}{\bibinfo{person}{Stephan Borzsony} {et~al\mbox{.}}}
  \bibinfo{year}{2001}\natexlab{}.
\newblock \showarticletitle{The skyline operator}. In
  \bibinfo{booktitle}{\emph{ICDE}}. IEEE.
\newblock


\bibitem[\protect\citeauthoryear{Chilton, Little, Edge, Weld, and
  Landay}{Chilton et~al\mbox{.}}{2013}]%
        {cascade}
\bibfield{author}{\bibinfo{person}{Lydia~B Chilton}, \bibinfo{person}{Greg
  Little}, \bibinfo{person}{Darren Edge}, \bibinfo{person}{Daniel~S Weld},
  {and} \bibinfo{person}{James~A Landay}.} \bibinfo{year}{2013}\natexlab{}.
\newblock \showarticletitle{Cascade: Crowdsourcing taxonomy creation}. In
  \bibinfo{booktitle}{\emph{Proceedings of the SIGCHI Conference on Human
  Factors in Computing Systems}}. ACM, \bibinfo{pages}{1999--2008}.
\newblock


\bibitem[\protect\citeauthoryear{Chomicki et~al\mbox{.}}{Chomicki
  et~al\mbox{.}}{2013}]%
        {sky1}
\bibfield{author}{\bibinfo{person}{Jan Chomicki} {et~al\mbox{.}}}
  \bibinfo{year}{2013}\natexlab{}.
\newblock \showarticletitle{Skyline queries, front and back}.
\newblock \bibinfo{journal}{\emph{SIGMOD}} (\bibinfo{year}{2013}).
\newblock


\bibitem[\protect\citeauthoryear{De~Berg et~al\mbox{.}}{De~Berg
  et~al\mbox{.}}{1997}]%
        {de1997computational}
\bibfield{author}{\bibinfo{person}{Mark De~Berg} {et~al\mbox{.}}}
  \bibinfo{year}{1997}\natexlab{}.
\newblock \showarticletitle{Computational geometry}.
\newblock In \bibinfo{booktitle}{\emph{Computational geometry}}.
  \bibinfo{publisher}{Springer}.
\newblock


\bibitem[\protect\citeauthoryear{Garey and Johnson}{Garey and Johnson}{2002}]%
        {garey2002computers}
\bibfield{author}{\bibinfo{person}{Michael~R Garey} {and}
  \bibinfo{person}{David~S Johnson}.} \bibinfo{year}{2002}\natexlab{}.
\newblock \bibinfo{booktitle}{\emph{Computers and intractability}}.
\newblock \bibinfo{publisher}{wh freeman New York}.
\newblock


\bibitem[\protect\citeauthoryear{Gauch et~al\mbox{.}}{Gauch
  et~al\mbox{.}}{1991}]%
        {query4}
\bibfield{author}{\bibinfo{person}{Susan Gauch} {et~al\mbox{.}}}
  \bibinfo{year}{1991}\natexlab{}.
\newblock \bibinfo{booktitle}{\emph{Search improvement via automatic query
  reformulation}}.
\newblock \bibinfo{type}{{T}echnical {R}eport}. \bibinfo{institution}{UNC
  Chapel Hill, Computer Science}.
\newblock


\bibitem[\protect\citeauthoryear{Good and Su}{Good and Su}{2013}]%
        {bio}
\bibfield{author}{\bibinfo{person}{Benjamin~M Good} {and}
  \bibinfo{person}{Andrew~I Su}.} \bibinfo{year}{2013}\natexlab{}.
\newblock \showarticletitle{Crowdsourcing for bioinformatics}.
\newblock \bibinfo{journal}{\emph{Bioinformatics}} \bibinfo{volume}{29},
  \bibinfo{number}{16} (\bibinfo{year}{2013}), \bibinfo{pages}{1925--1933}.
\newblock


\bibitem[\protect\citeauthoryear{Huang and Bigham}{Huang and Bigham}{2017}]%
        {dep1}
\bibfield{author}{\bibinfo{person}{Ting-Hao~Kenneth Huang} {and}
  \bibinfo{person}{Jeffrey~P Bigham}.} \bibinfo{year}{2017}\natexlab{}.
\newblock \showarticletitle{A 10-month-long deployment study of on-demand
  recruiting for low-latency crowdsourcing}. In \bibinfo{booktitle}{\emph{Fifth
  AAAI Conference on Human Computation and Crowdsourcing}}.
\newblock


\bibitem[\protect\citeauthoryear{Ibarra et~al\mbox{.}}{Ibarra
  et~al\mbox{.}}{1975}]%
        {ibarra1975fast}
\bibfield{author}{\bibinfo{person}{Oscar~H Ibarra} {et~al\mbox{.}}}
  \bibinfo{year}{1975}\natexlab{}.
\newblock \showarticletitle{Fast approximation algorithms for the knapsack and
  sum of subset problems}.
\newblock \bibinfo{journal}{\emph{Journal of the ACM (JACM)}}
  (\bibinfo{year}{1975}).
\newblock


\bibitem[\protect\citeauthoryear{Ioannidis, Ng, Shim, and Sellis}{Ioannidis
  et~al\mbox{.}}{1992}]%
        {ioannidis1992parametric}
\bibfield{author}{\bibinfo{person}{Yannis~E Ioannidis},
  \bibinfo{person}{Raymond~T Ng}, \bibinfo{person}{Kyuseok Shim}, {and}
  \bibinfo{person}{Timos~K Sellis}.} \bibinfo{year}{1992}\natexlab{}.
\newblock \showarticletitle{Parametric query optimization}. In
  \bibinfo{booktitle}{\emph{VLDB}}, Vol.~\bibinfo{volume}{92}. Citeseer,
  \bibinfo{pages}{103--114}.
\newblock


\bibitem[\protect\citeauthoryear{Jin et~al\mbox{.}}{Jin et~al\mbox{.}}{2007}]%
        {sky2}
\bibfield{author}{\bibinfo{person}{Wen Jin} {et~al\mbox{.}}}
  \bibinfo{year}{2007}\natexlab{}.
\newblock \showarticletitle{The multi-relational skyline operator}. In
  \bibinfo{booktitle}{\emph{ICDE}}. IEEE.
\newblock


\bibitem[\protect\citeauthoryear{Kadi}{Kadi}{[n.d.]}]%
        {thesis}
\bibfield{author}{\bibinfo{person}{Ouiame Ait~El Kadi}.}
  \bibinfo{year}{[n.d.]}\natexlab{}.
\newblock \showarticletitle{Exploring Crowdsourcing Deployment Strategies
  through Recommendation and Iterative Refinement}.
\newblock \bibinfo{journal}{\emph{MS Research Report}}
  (\bibinfo{year}{[n.\,d.]}).
\newblock


\bibitem[\protect\citeauthoryear{Kittur, Smus, Khamkar, and Kraut}{Kittur
  et~al\mbox{.}}{2011}]%
        {crowdforge}
\bibfield{author}{\bibinfo{person}{Aniket Kittur}, \bibinfo{person}{Boris
  Smus}, \bibinfo{person}{Susheel Khamkar}, {and} \bibinfo{person}{Robert~E
  Kraut}.} \bibinfo{year}{2011}\natexlab{}.
\newblock \showarticletitle{Crowdforge: Crowdsourcing complex work}. In
  \bibinfo{booktitle}{\emph{Proceedings of the 24th annual ACM symposium on
  User interface software and technology}}. ACM, \bibinfo{pages}{43--52}.
\newblock


\bibitem[\protect\citeauthoryear{Kulkarni, Can, and Hartmann}{Kulkarni
  et~al\mbox{.}}{2012}]%
        {turkomatic}
\bibfield{author}{\bibinfo{person}{Anand Kulkarni}, \bibinfo{person}{Matthew
  Can}, {and} \bibinfo{person}{Bj{\"o}rn Hartmann}.}
  \bibinfo{year}{2012}\natexlab{}.
\newblock \showarticletitle{Collaboratively crowdsourcing workflows with
  turkomatic}. In \bibinfo{booktitle}{\emph{Proceedings of the acm 2012
  conference on computer supported cooperative work}}. ACM,
  \bibinfo{pages}{1003--1012}.
\newblock


\bibitem[\protect\citeauthoryear{Kulkarni, Can, and Hartmann}{Kulkarni
  et~al\mbox{.}}{2011}]%
        {workflow3}
\bibfield{author}{\bibinfo{person}{Anand~P Kulkarni}, \bibinfo{person}{Matthew
  Can}, {and} \bibinfo{person}{Bjoern Hartmann}.}
  \bibinfo{year}{2011}\natexlab{}.
\newblock \showarticletitle{Turkomatic: automatic recursive task and workflow
  design for mechanical turk}. In \bibinfo{booktitle}{\emph{CHI'11 Extended
  Abstracts on Human Factors in Computing Systems}}. ACM,
  \bibinfo{pages}{2053--2058}.
\newblock


\bibitem[\protect\citeauthoryear{Lasecki, Kushalnagar, and Bigham}{Lasecki
  et~al\mbox{.}}{2014}]%
        {legion}
\bibfield{author}{\bibinfo{person}{Walter~S Lasecki}, \bibinfo{person}{Raja
  Kushalnagar}, {and} \bibinfo{person}{Jeffrey~P Bigham}.}
  \bibinfo{year}{2014}\natexlab{}.
\newblock \showarticletitle{Legion scribe: real-time captioning by
  non-experts}. In \bibinfo{booktitle}{\emph{Proceedings of the 16th
  international ACM SIGACCESS conference on Computers \& accessibility}}. ACM,
  \bibinfo{pages}{303--304}.
\newblock


\bibitem[\protect\citeauthoryear{Lasecki, Wesley, Nichols, Kulkarni, Allen, and
  Bigham}{Lasecki et~al\mbox{.}}{2013}]%
        {chorus}
\bibfield{author}{\bibinfo{person}{Walter~S Lasecki}, \bibinfo{person}{Rachel
  Wesley}, \bibinfo{person}{Jeffrey Nichols}, \bibinfo{person}{Anand Kulkarni},
  \bibinfo{person}{James~F Allen}, {and} \bibinfo{person}{Jeffrey~P Bigham}.}
  \bibinfo{year}{2013}\natexlab{}.
\newblock \showarticletitle{Chorus: a crowd-powered conversational assistant}.
  In \bibinfo{booktitle}{\emph{Proceedings of the 26th annual ACM symposium on
  User interface software and technology}}. ACM, \bibinfo{pages}{151--162}.
\newblock


\bibitem[\protect\citeauthoryear{Lawler}{Lawler}{1979}]%
        {lawler1979fast}
\bibfield{author}{\bibinfo{person}{Eugene~L Lawler}.}
  \bibinfo{year}{1979}\natexlab{}.
\newblock \showarticletitle{Fast approximation algorithms for knapsack
  problems}.
\newblock \bibinfo{journal}{\emph{Mathematics of Operations Research}}
  (\bibinfo{year}{1979}).
\newblock


\bibitem[\protect\citeauthoryear{Lin, Daniel, and Weld}{Lin
  et~al\mbox{.}}{2012}]%
        {workflow1}
\bibfield{author}{\bibinfo{person}{Christopher~H Lin}, \bibinfo{person}{Mausam
  Daniel}, {and} \bibinfo{person}{S Weld}.} \bibinfo{year}{2012}\natexlab{}.
\newblock \showarticletitle{Dynamically switching between synergistic workflows
  for crowdsourcing}. In \bibinfo{booktitle}{\emph{In Proceedings of the 26th
  AAAI Conference on Artificial Intelligence, AAAI’12}}. Citeseer.
\newblock


\bibitem[\protect\citeauthoryear{Mishra et~al\mbox{.}}{Mishra
  et~al\mbox{.}}{2009}]%
        {mishra}
\bibfield{author}{\bibinfo{person}{Chaitanya Mishra} {et~al\mbox{.}}}
  \bibinfo{year}{2009}\natexlab{}.
\newblock \showarticletitle{Interactive query refinement}. In
  \bibinfo{booktitle}{\emph{EDBT}}. ACM.
\newblock


\bibitem[\protect\citeauthoryear{Mottin et~al\mbox{.}}{Mottin
  et~al\mbox{.}}{2013}]%
        {query2}
\bibfield{author}{\bibinfo{person}{Davide Mottin} {et~al\mbox{.}}}
  \bibinfo{year}{2013}\natexlab{}.
\newblock \showarticletitle{A probabilistic optimization framework for the
  empty-answer problem}.
\newblock \bibinfo{journal}{\emph{VLDB}} (\bibinfo{year}{2013}).
\newblock


\bibitem[\protect\citeauthoryear{Mouratidis et~al\mbox{.}}{Mouratidis
  et~al\mbox{.}}{2006}]%
        {sky4}
\bibfield{author}{\bibinfo{person}{Kyriakos Mouratidis} {et~al\mbox{.}}}
  \bibinfo{year}{2006}\natexlab{}.
\newblock \showarticletitle{Continuous monitoring of top-k queries over sliding
  windows}. In \bibinfo{booktitle}{\emph{SIGMOD}}. ACM.
\newblock


\bibitem[\protect\citeauthoryear{Mouratidis and Tang}{Mouratidis and
  Tang}{2018}]%
        {sky3}
\bibfield{author}{\bibinfo{person}{Kyriakos Mouratidis} {and}
  \bibinfo{person}{Bo Tang}.} \bibinfo{year}{2018}\natexlab{}.
\newblock \showarticletitle{Exact Processing of Uncertain Top-k Queries in
  Multi-criteria Settings}.
\newblock \bibinfo{journal}{\emph{PVLDB}} (\bibinfo{year}{2018}).
\newblock


\bibitem[\protect\citeauthoryear{Mozafari, Sarkar, Franklin, Jordan, and
  Madden}{Mozafari et~al\mbox{.}}{2014}]%
        {mozafari2014scaling}
\bibfield{author}{\bibinfo{person}{Barzan Mozafari}, \bibinfo{person}{Purna
  Sarkar}, \bibinfo{person}{Michael Franklin}, \bibinfo{person}{Michael
  Jordan}, {and} \bibinfo{person}{Samuel Madden}.}
  \bibinfo{year}{2014}\natexlab{}.
\newblock \showarticletitle{Scaling up crowd-sourcing to very large datasets: a
  case for active learning}.
\newblock \bibinfo{journal}{\emph{Proceedings of the VLDB Endowment}}
  \bibinfo{volume}{8}, \bibinfo{number}{2} (\bibinfo{year}{2014}),
  \bibinfo{pages}{125--136}.
\newblock


\bibitem[\protect\citeauthoryear{Pilourdault et~al\mbox{.}}{Pilourdault
  et~al\mbox{.}}{2017}]%
        {julien}
\bibfield{author}{\bibinfo{person}{Julien Pilourdault} {et~al\mbox{.}}}
  \bibinfo{year}{2017}\natexlab{}.
\newblock \showarticletitle{Motivation-aware task assignment in crowdsourcing}.
  In \bibinfo{booktitle}{\emph{EDBT}}.
\newblock


\bibitem[\protect\citeauthoryear{Rahman et~al\mbox{.}}{Rahman
  et~al\mbox{.}}{2018}]%
        {habib}
\bibfield{author}{\bibinfo{person}{Habibur Rahman} {et~al\mbox{.}}}
  \bibinfo{year}{2018}\natexlab{}.
\newblock \showarticletitle{Optimized group formation for solving collaborative
  tasks}.
\newblock \bibinfo{journal}{\emph{The VLDB Journal}} (\bibinfo{year}{2018}),
  \bibinfo{pages}{1--23}.
\newblock


\bibitem[\protect\citeauthoryear{Stol and Fitzgerald}{Stol and
  Fitzgerald}{2014}]%
        {software}
\bibfield{author}{\bibinfo{person}{Klaas-Jan Stol} {and} \bibinfo{person}{Brian
  Fitzgerald}.} \bibinfo{year}{2014}\natexlab{}.
\newblock \showarticletitle{Two's company, three's a crowd: a case study of
  crowdsourcing software development}. In \bibinfo{booktitle}{\emph{Proceedings
  of the 36th International Conference on Software Engineering}}. ACM,
  \bibinfo{pages}{187--198}.
\newblock


\bibitem[\protect\citeauthoryear{Trummer and Koch}{Trummer and Koch}{2016}]%
        {trummer2016multi}
\bibfield{author}{\bibinfo{person}{Immanuel Trummer} {and}
  \bibinfo{person}{Christoph Koch}.} \bibinfo{year}{2016}\natexlab{}.
\newblock \showarticletitle{Multi-objective parametric query optimization}.
\newblock \bibinfo{journal}{\emph{ACM SIGMOD Record}} \bibinfo{volume}{45},
  \bibinfo{number}{1} (\bibinfo{year}{2016}), \bibinfo{pages}{24--31}.
\newblock


\bibitem[\protect\citeauthoryear{Yuen, King, and Leung}{Yuen
  et~al\mbox{.}}{2011}]%
        {survey}
\bibfield{author}{\bibinfo{person}{Man-Ching Yuen}, \bibinfo{person}{Irwin
  King}, {and} \bibinfo{person}{Kwong-Sak Leung}.}
  \bibinfo{year}{2011}\natexlab{}.
\newblock \showarticletitle{A survey of crowdsourcing systems}. In
  \bibinfo{booktitle}{\emph{2011 IEEE Third International Conference on
  Privacy, Security, Risk and Trust and 2011 IEEE Third International
  Conference on Social Computing}}. IEEE, \bibinfo{pages}{766--773}.
\newblock


\bibitem[\protect\citeauthoryear{Zheng et~al\mbox{.}}{Zheng
  et~al\mbox{.}}{2011}]%
        {zheng2011task}
\bibfield{author}{\bibinfo{person}{Haichao Zheng} {et~al\mbox{.}}}
  \bibinfo{year}{2011}\natexlab{}.
\newblock \showarticletitle{Task design, motivation, and participation in
  crowdsourcing contests}.
\newblock \bibinfo{journal}{\emph{International Journal of Electronic
  Commerce}} (\bibinfo{year}{2011}).
\newblock


\end{thebibliography}

\end{document}